\documentclass[acmsmall,10pt]{acmart}
\pdfoutput=1
\usepackage{etex}
\usepackage{hyperref}
\usepackage[setrelation]{math}
\usepackage[modernsign,substopindex,shortmquant,mquantifiertype,
mconnectiveformal,postfixflatinterpret,bracketmodalinterpret,
setfixinterpret,modifopindex,seqinfers,seqoptional,footnotecalculus,abbrseqcontext,shortterms,nosigmaterms,novarterms]{logic}
\usepackage[pretest,nocommandblocks]{progreg}
\usepackage[bracketinterpret,postfixinterpret,bracketmodalinterpret,simplenames]{dL}

\usepackage{tabularx}
\usepackage{booktabs}
\usepackage{paralist}
\usepackage{todonotes}
\usepackage{wrapfig}
\usepackage{stmaryrd}
\usepackage{proof}

\let\qed\relax

\usepackage{pf2}

\newcommand{\bimply}{\leftrightarrow}
\newcommand{\folr}{\ensuremath{FOL_{\mathbb{R}}}}
\usepackage{prettyref}

\newcommand{\ltac}{\ensuremath{\mathcal{L}_{\tt tac}}~}
\newcommand{\sproof}{{\tt SP}}

\newcommand{\hrsub}[2]{{\tt sub}(#1,#2)}
\newcommand{\hreq}[3]{{\tt eq}(#1,#2,#3)}
\newcommand{\hrany}[2]{{\tt any}(#1,#2)}

\newcommand{\hhsub}[3]{#1,{\tt sub}(#2,#3)}
\newcommand{\hheq}[4]{#1,{\tt eq}(#2,#3,#4)}
\newcommand{\hhany}[3]{#1,{\tt any}(#2,#3)}
\newcommand{\arb}[2]{{#1}^{#2}}

\newrefformat{model}{Model\,\ref{#1}}
\newrefformat{listing}{Listing\,\ref{#1}}
\newrefformat{line}{line\,\ref{#1}}
\newrefformat{sec}{Section\,\ref{#1}}
\newrefformat{appendix}{Appendix\,\ref{#1}}
\newrefformat{def}{Definition\,\ref{#1}}
\newrefformat{defn}{Definition\,\ref{#1}}
\newrefformat{thm}{Theorem\,\ref{#1}}
\newrefformat{ax}{\ref{#1}}
\newrefformat{prop}{Proposition\,\ref{#1}}
\newrefformat{lem}{Lemma\,\ref{#1}}
\newrefformat{cor}{Corollary\,\ref{#1}}
\newrefformat{ex}{Example\,\ref{#1}}
\newrefformat{tab}{Table\,\ref{#1}}
\newrefformat{fig}{Figure\,\ref{#1}}
\newrefformat{eqn}{(\ref{#1})}

\newtheorem{model}{Model}

\usepackage{xcolor, colortbl}
\usepackage{color}
\definecolor{gray}{rgb}{0.5,0.5,0.5}

\providecommand{\KeYmaeraX}{KeYmaera X\xspace}

\acmVolume{1}
\acmNumber{1}
\acmArticle{1}
\acmYear{2019}
\acmMonth{1}
\acmDOI{10.1145/nnnnnnn.nnnnnnn}
\startPage{1}
\acmYear{2019}
\acmISBN{978-x-xxxx-xxxx-x/YY/MM}
\acmDOI{10.1145/nnnnnnn.nnnnnnn}
\startPage{1}
\usepackage{listings}
\usepackage{verbatim}
\usepackage{courier} %use courier as monospaced font so that we get bold keywords
\usepackage{upquote} % This package makes backticks (``) display correctly in listings.

\lstdefinelanguage{scala}{
  morekeywords={abstract,case,catch,class,def,%
    else,extends,false,final,finally,%
    for,if,implicit,import,match,mixin,%
    new,null,object,override,package,%
    private,protected,requires,return,sealed,%
    super,this,throw,trait,true,try,%
    type,val,var,while,with,yield},
  otherkeywords={=>,<-,<\%,<:,>:,\#,@,by,&,|,<,?},
  sensitive=true,
  morecomment=[l]{//},
  morecomment=[n]{/*}{*/},
  morestring=[b]",
  morestring=[b]"""
}

\lstdefinelanguage{bellerophon}{
  morekeywords={partial,OnAll,'R,'L,'_},
  otherkeywords={&,|,<,?},
  sensitive=true,
  morecomment=[l]{//},
  morecomment=[n]{/*}{*/},
  morestring=[b]",
  morestring=[b]"""
}

\renewcommand{\dL}{{\sf{dL}}\xspace}
\author{
  Brandon Bohrer \and Andr\'{e} Platzer
}
\begin{document}
\title{Toward Structured Proofs for Dynamic Logics}

\begin{abstract}
  We present Kaisar, a structured interactive proof language for differential dynamic logic (\dL), for safety-critical cyber-physical systems (CPS).
  The defining feature of Kaisar is \emph{nominal terms}, which simplify CPS proofs by making the frequently needed historical references to past program states first-class.
  To support nominals, we extend the notion of structured proof with a first-class notion of \emph{structured symbolic execution} of CPS models.

  We implement Kaisar in the theorem prover \KeYmaeraX and reproduce an example on the safe operation of a parachute and a case study on ground robot control.
  We show how nominals simplify common CPS reasoning tasks when combined with other features of structured proof.
  We develop an extensive metatheory for Kaisar.
  In addition to soundness and completeness, we show a formal specification for Kaisar's nominals and relate Kaisar to a nominal variant of \dL.
\end{abstract}

%\ccsdesc[500]{Theory of computation~Constraint and logic programming}
%\ccsdesc[500]{Software and its engineering~Compilers}
%\ccsdesc[100]{Theory of computation~Type theory}
%\keywords
%{Certifying Compilation, Proof-Producing Compilation, Prolog, Warren Abstract Machine}

\maketitle
\newcommand{\JavaDL}{{\texttt{JavaDL}}}
\section{Introduction}
%TODO good opportunity to cite others too
Many cyber-physical systems (CPS) such as autonomous cars~\cite{DBLP:conf/fm/LoosPN11}, airborne collision-avoidance systems~\cite{DBLP:conf/emsoft/JeanninGKGSZP15} and surgical robots~\cite{DBLP:conf/hybrid/KouskoulasRPK13} are safety-critical, and thus their correctness is of utmost importance.
Differential dynamic logic (\dL)~\cite{DBLP:journals/jar/Platzer08} is a domain-specific logic which expresses correctness theorems such as safety and liveness for hybrid dynamical systems models of CPS, which combine discrete computation with ordinary differential equations (ODEs).

Because safety and liveness for hybrid systems are undecidable,~\cite{DBLP:conf/lics/Henzinger96} achieving strong correctness results for nontrivial systems demands an interactive technique that allows users to provide human insight when automation does not suffice.
%DBLP:conf/popl/CousotC77
The \KeYmaeraX~\cite{DBLP:conf/cade/FultonMQVP15} theorem prover for \dL (in comparison to, e.g. model checking ~\cite{DBLP:conf/cav/FrehseGDCRLRGDM11,DBLP:conf/hybrid/Dreossi17,DBLP:conf/hybrid/Frehse05,DBLP:conf/popl/BokerHR14,DBLP:conf/cav/HenzingerHW97,DBLP:conf/vmcai/KidoCH16,DBLP:conf/icalp/SuenagaH11,DBLP:conf/fmcad/GaoKC13,DBLP:conf/cade/GaoKC13} approaches) enables human insight through \emph{interactive theorem proving}.
The advantage of this interaction-based approach is its flexibility.
For any given hybrid system, the proof difficulty may lie in the complexity of discrete dynamics, continuous dynamics, solving first-order real arithmetic at the leaves, or all three.
Through interactivity, the proof author can address the  difficult aspects of their system, while letting automation handle simpler parts.
Interactive theorem-proving for CPS verification has proven fruitful through a number of case studies~\cite{DBLP:conf/fm/LoosPN11,DBLP:conf/emsoft/JeanninGKGSZP15,DBLP:conf/hybrid/KouskoulasRPK13,DBLP:conf/fm/PlatzerC09,DBLP:conf/icfem/PlatzerQ09,DBLP:conf/hybrid/LoosRP13, DBLP:conf/itsc/MullerMP15,DBLP:conf/itsc/MullerMP15}.

%TODO: cite keymaera, keymaera x, bellerophon
While these results demonstrate the significant potential that deduction has in CPS, it remains to be seen what the most productive way is for writing interactive CPS proofs.
Interactive point-and-click interfaces are available for \KeYmaeraX \cite{DBLP:conf/fide/MitschP16} and its predecessor \KeYmaera \cite{DBLP:conf/cade/PlatzerQ08}, which are useful for learning, but become tedious at scale.
Tactical theorem proving is available for CPS \cite{Belle}, which is useful for programming generic proof search procedures, but requires a certain level of aptitude in the proof system and is harder to scale for complex applications.
%TODO: Cite differential ghosts
We, thus, argue that CPS proofs are to be taken seriously and merit the full attention of a first-class approach.

%These arguments demand tracking historical values across multiple intermediate states with ghost variables.
We focus on the opportunity provided by the hybrid mixture of discrete and continuous invariants inherent to hybrid systems.
Typical proofs combine discrete and continuous invariants, which in turn combine continuously-evolving, discretely-evolving, and constant variables.
CPS invariant proofs often relate the present state to multiple historical system states.
It has long been known that historical reference, implementable with ghost state, is an essential component of proof in many domains~\cite{DBLP:journals/tcs/AptBM79,Owicki1976,DBLP:books/garland/Owicki75,apt2010verification,DBLP:journals/acta/Clint73}.
It has been known equally long~\cite{DBLP:journals/acta/Clarke80} that manual ghost arguments can make proofs clumsy.
For example, stating even a single invariant might require multiple non-local changes to introduce ghost state.
This is especially noticeable for CPS proofs, which typically contain nested invariants with references to multiple states.
%This is even truer for CPS verification: the discrete ghosts needed for historical reference must interact cleanly with continuous differential ghosts, which are necessary~\cite{DBLP:journals/corr/abs-1104-1987} for completeness.

In this paper, we devise a proof language Kaisar for CPS that addresses the problem of automating historical reference.
The distinguishing feature of Kaisar is \emph{nominal proof}: proof authors can give names $t$ to abstract proof states and then refer to the value of any term $\theta$ in proof state $t$ by a \emph{nominal term} $t(\theta)$.
We show through the examples of this paper~\cite{Belle}
% and our reproduction of a robotics case study~\cite{DBLP:conf/rss/MitschGP13} (see artifact) 
that nominal terms simplify the complex historical reasoning necessary for CPS verification:
Instead of cluttering a proof with explicit ghosts of each term, the author can write down values from past states directly when needed, and the Kaisar nominal system will automatically supply the necessary ghost state.
%This provides a simpler interaction style: when the author realizes after-the-fact that they need access to intermediate state, they need not 

The automation of historical reference in Kaisar is supported by \emph{structured symbolic execution}, an extension of the \emph{structured proof} paradigm introduced in Mizar~\cite{conf/mkm/BancerekBGKMNPU15}, and seen in Isabelle's Isar~\cite{DBLP:conf/tphol/Wenzel99} and ${\rm TLA}^+$'s TLAPS~\cite{DBLP:conf/fm/CousineauDLMRV12,DBLP:conf/hybrid/Lamport92}.
Symbolic execution is a first-class proof language feature in Kaisar, 
enabling the language to automatically maintain an abstract execution trace, introducing ghost variables when needed.
This trace enables automatically reducing nominal terms to ghost state, solving the historical reference automation problem.
Hence the name: Kaisar is to \KeYmaeraX as Isar is to Isabelle: A structured proof language tailored to the needs of each prover's logic.
%Hence the name: Kaisar

In structuring proofs around first-class symbolic execution, it is non-obvious but essential that we do not restrict which formulas are provable.
%TODO: Cites
To this end, we develop a comprehensive metatheory for Kaisar, showing it is not only sound, but complete with respect to existing calculi~\cite{DBLP:journals/jar/Platzer08}, which are themselves sound and relatively complete for hybrid systems~\cite{DBLP:journals/jar/Platzer16,DBLP:conf/lics/Platzer12,DBLP:journals/jar/Platzer08}.
We also give a precise semantics to nominal terms via dynamic execution traces: the nominal term computed by Kaisar for $t(\theta)$ equals the value of $\theta$ in the concrete program state corresponding to abstract state $t$.
We relate nominal terms to a nominal dialect of \dL called $\dLh$:~\cite{DBLP:journals/entcs/Platzer07} the meaning of proof-level named states $t$ can be understood via the logic-level nominal states of $\dLh$.
The above applies for both discrete and continuous dynamics and for both initial and intermediate states.

In adopting structured symbolic execution, we import the existing benefits of structured proof to hybrid systems for the first time.
It has long been noted~\cite{DBLP:conf/mkm/Wenzel06} that structured proof languages improve scalability and maintainability via block structure, and improve readability via their declarative style, with readability benefits extending to natural-language structured proofs~\cite{Lamport12,Lamport95}.
This is important for CPS verification as the scale of systems verified continues to grow~\cite{DBLP:conf/emsoft/JeanninGKGSZP15,DBLP:conf/itsc/MullerMP15}.

%,Belle
Lastly, we believe the nominal mechanism of Kaisar is of interest beyond CPS verification.
A number of verification tools ~\cite{KeYBook,DBLP:conf/cade/FultonMQVP15,DBLP:conf/lpar/Leino10,this-is-boogie-2-2,Barnett2005,DBLP:conf/fosad/LeinoMS09,DBLP:books/daglib/p/LeavensBR99} offer ad-hoc constructs without theoretical justification which can reference only the initial program state.
We generalize and formally justify them.
%We offer a formal justification of and a generalization of such features.

We show that Kaisar works in practice with a prototype for Kaisar in \KeYmaeraX.
We evaluate our implementation using the examples in this paper, including a representative proof on parachute control.
We evaluate further by reproducing a case study on the control of ground robots~(Section~\ref{sec:concl}).

% %TODO: Could cite even more case studies.
% %TODO make clearer what this has to do with state relation/ghost. Not very convincing/clear yet.
% %TODO at some point will have to relate to x(0) versus x(T) versus x(99) which is useful in math but not very well-defined in hybrid plus needs more symbolic things that are not concrete numerics like stopped at times t_1, t_2, t_55, t_276. Maybe not yet in intro?
% %TODO aren't there more approaches doing some ad-hoc approximation? Also JML ghost fields? Maybe enough to have in related work not here but it;s a bit us-heavy.

In Section~\ref{sec:dl} we provide a primer on \dL, hybrid systems and dynamic logics, giving an informal structured proof following Lamport~\cite{Lamport12,Lamport95}.
In Section~\ref{sec:fok}, we present the propositional fragment of Kaisar, introducing structural constructs and pattern-matching to aid in making proofs concise.
This fragment borrows from Isar; we present it separately to show its relationship with prior work.
In Section~\ref{sec:trace}, we detail the technical challenge and solution for nominals.
In Section~\ref{sec:ddk} we add that solution to Kaisar, generalizing it to the discrete fragment of \dL by introducing \emph{nominal terms}.
In Section~\ref{sec:dlk} we show that our notion of nominal terms generalizes to full \dL.
In Section~\ref{sec:metatheory} we validate Kaisar by developing its metatheory.
%Further validation is provided by our implementation in \KeYmaeraX and the examples throughout this paper, all of which are supported by the implementation.
%We provide examples throughout which, combined with our implementation of Kaisar in \KeYmaeraX, serve as validation of the design.
In Section~\ref{sec:relatedwork} we compare Kaisar with other proof languages.

\section{Background: Differential Dynamic Logic}
\label{sec:dl}
\emph{Differential dynamic logic} (\dL)~\cite{DBLP:conf/lics/Platzer12,DBLP:journals/jar/Platzer08} is a dynamic logic~\cite{DBLP:conf/focs/Pratt76,Harel} for formally verifying hybrid (dynamical) systems models of cyber-physical systems (CPS), which is implemented in the theorem prover \KeYmaeraX~\cite{DBLP:conf/cade/FultonMQVP15} and has seen successful
application in a number of case studies~\cite{DBLP:conf/rss/MitschGP13,DBLP:conf/emsoft/JeanninGKGSZP15,DBLP:conf/hybrid/KouskoulasRPK13,DBLP:conf/fm/LoosPN11}.
In \dL, a hybrid systems model of a CPS is expressed as a (nondeterministic) program in the language of \emph{hybrid programs}.
The distinguishing feature of hybrid programs is the ability to model continuous physics with ordinary differential equations (ODEs).
Combined with nondeterminism, real arithmetic, and standard discrete constructs, this suffices to express hybrid systems.

As usual in dynamic logic, \dL internalizes program execution with first-class modal operators $[\alpha]\phi$ and $\langle\alpha\rangle\phi$ stating $\phi$ is true after \emph{all} or \emph{some} execution paths of the program $\alpha$, respectively.
Dynamic logics are a generalization of Hoare logics, e.g. any Hoare triple $\{P\}\alpha\{Q\}$ can be expressed in dynamic logic as the formula $P\limply[\alpha]Q$.
Dynamic logic is more general because the modal operators are first-class and may be nested freely.
%%In this paper, we focus on box modalities $[\alpha]\phi$, leaving diamond modalities $\langle\alpha\rangle\phi$ as future work.
%The basic design principles of Kaisar apply to a broad family of dynamic logics, but for concreteness we present Kaisar as specialized to differential dynamic logic (\dL).
Throughout this paper, we let $e,f,g$ range over expressions of \dL.
Expressions are divided into propositions $\phi,\psi$, hybrid programs $\alpha,\beta$, and real-valued terms $\theta$.
The proposition language of \dL combines first-order arithmetic ($\folr{}$) with the dynamic logic modalities $\dbox{\alpha}{\phi}$ and $\ddiamond{\alpha}{\phi}$:
\begin{center}
  $\phi,\psi,P,Q~::=~\phi \wedge \psi\ |\ \phi \vee \psi\ |\ \neg{\phi}\ |\ \forall x~\phi\ |\ \exists x~\phi\ |\ [\alpha]\phi\ |\ \langle\alpha\rangle\phi\ |\ \theta_1 \sim \theta_2$
\end{center}
where $\sim$ is a comparison operator on the (classical) reals, i.e. $\sim\ \in\{<,\leq,=,\geq,>,\neq\}$.

The term language of \dL consists of basic arithmetic operations on real-valued variables with rational literals $q\in\mathbb{Q}$ and rational exponentiation $\theta^q$:
\begin{center}
  $\theta~::=~x\ |\ q\ |\ \theta_1 + \theta_2\ |\ \theta_1 - \theta_2\
  |\ \theta_1 \cdot\ \theta_2\ |\ \theta_1/\theta_2\ |\ \theta^q$
\end{center}
This means that the program-free fragment of \dL is equivalent to first-order classical logic over real-closed fields (\folr), which is decidable but unscalable, requiring doubly-exponential time~\cite{Davenport1988}. The \emph{hybrid programs} of \dL include typical discrete constructs and differential equation systems:
\begin{center}
  $\alpha,\beta~::=~x:=\theta\ |\ \prandom{x}\ |\ ?\phi\ |\ \alpha;\beta\ |\ \alpha \cup \beta\ |\ \alpha^*\ |\ \pevolvein{x' = \theta}{Q}$
\end{center}
%The constructs are:
\begin{itemize}
\item \emph{Assignment} $x:=\theta$ Sets a variable $x$ to the current value of term $\theta$.
\item \emph{Nondeterministic Assignment} $\prandom{x}$ Sets $x$ nondeterministically to an arbitrary real.
\item \emph{Assertion} $?\phi$ Execution transitions to the current state if formula $\phi$ is true, else the program has no transitions.
In typical models $\phi$ is first-order arithmetic.
%In this paper we focus on the common case where $\phi$ is first-order arithmetic.
\item \emph{Sequential Composition} $\alpha;\beta$ Runs program $\alpha$, then runs $\beta$ in some resulting state.
\item \emph{Choice} $\alpha\cup\beta$ Nondeterministically runs either program $\alpha$ or $\beta$.
\item \emph{Iteration} $\alpha^*$ Nondeterministically runs program $\alpha$ zero or more times.
\item \emph{ODE Evolution} $\{x'=\theta \& Q\}$ Evolves the ODE system $x'=\theta$ nondeterministically for any duration $r \in \mathbb{R}_{\geq{0}}$ that never leaves the evolution domain formula $Q$.
%Systems with multiple equations can be written as $\{x'=\theta_1, y'=\theta_2 \& H\},$ but we typically present the single-equation case for notational convenience.
\end{itemize}

We illustrate \dL with a running example in Model~\ref{model:dive}, which models the decision-making of a skydiver opening their parachute at a safe time.
We intentionally chose this example because it demonstrates the techniques necessary for more complex systems within the space provided.
As further validation we reproduced a larger case study~\cite{DBLP:conf/rss/MitschGP13}, see Section~\ref{sec:concl}.
%We intentionally chose this example to be small enough to explain fully.
%This example combines many essential modeling and proof techniques.
%This system allows us to demonstrate continuous invariant reasoning about ODEs, since the drag force on a parachute is non-linear and thus requires non-trivial proof arguments.
%The safety argument also requires nontrivial arithmetic and thus we will use manual arithmetic proofs as well.
%The parachute's drag force is non-linear and thus we will exploit continuous invariant reasoning. 
%A common modeling pattern is that of a control-plant loop, where a system repeatedly make a discrete control decision and then allows the continuous physical environment (also called the plant) to evolve until the next control point.

\newcommand{\rvar}{r}
\newcommand{\ravar}{a}
\newcommand{\rpvar}{p}
\newcommand{\Tvar}{\ensuremath{\varepsilon}}
\begin{model}[Safety specification for the skydiver model]\label{model:dive}
\small
{\begin{flalign}
%\underbrace{
  \rvar = \ravar \land
\underbrace{\left(x \ge 0 \land v < 0\right)}_{\text{dc}} & \land 
\underbrace{\left(g > 0 \land 0 < \ravar  < \rpvar \land \Tvar \ge 0\right)}_{\text{const}} \land
\underbrace{m <  -\sqrt{\frac{g}{\rpvar}} < v}_{\text{dyn}}
%}_{\text{initial condition}} 
\tag{Pre}\label{exmodel:init} \\&
\limply 
[  
\big\{
%\underbrace{
~(?\underbrace{\left(\rvar{}=\ravar{} \land v - g\cdot \Tvar > -\sqrt{\frac{g}{\rpvar}}\right)}_\textit{Dive}~\cup~\humod{r}{\rpvar}); \tag{ctrl}\label{exmodel:ctrl}
%}_{\textit{ctrl}}\ ; 
\\&\phantom{\limply [ \big\{ }
  %\underbrace{
~\ \humod{t}{0};~ \{x'=v,~ v'=r\cdot v^2 - g ~\&~ x \ge 0 \land v<0 \land t \le \Tvar\} \tag{plant}\label{exmodel:plant}
    %}_{\textit{plant}}
\\
&\phantom{\limply [ \big\{}\big\}^*]
%\underbrace{
  (x=0 \limply \abs{v} \le m) \tag{Safe}\label{exmodel:post}
%}_{\text{post cond.}}
\end{flalign}
}\end{model}
% \begin{example}[Skydiver model]\label{ex:dive}
% \(
%     \prepeat{
%       \big(\underbrace{(\ptest{\textit{Dive}} \cup \humod{r}{\rpvar})}_{\textit{ctrl}}\ ;\
%       \underbrace{\{\D{x}=v,~ \D{v}=f(v,g,r)\}}_{\textit{plant}}\big)
%     }
% \)
This structure is typical: the system is a \emph{control-plant} loop, alternating between the skydiver's control decisions and evolution of the environment or \emph{plant} (i.e. gravity and drag).
The diver has a downward velocity $v < 0$, with a maximum safe landing speed $m$.
The chute is initially closed ($\rvar = \ravar$) until opened by the controller ($\humod{\rvar}{\rpvar}$).
The controller can choose not to open the chute if it will not exceed the \emph{equilibrium velocity} $-\sqrt{\frac{g}{\rpvar}}$, where $g$ is gravity and $\rpvar$ is the resistance of an open chute.
The plant evolves altitude $x$ at velocity $v$, which itself evolves at $v' = v^2\cdot{\rvar}-g$: Newton drag is proportional to $v^2$ by a factor $\rvar$ and gravity $g$ is constant.
The domain constraint stops the ODE upon hitting the ground ($x = 0$) or reaching the next control time ($t = \Tvar$ for time limit $\Tvar$).
It also simplifies the proof by providing an assumption $v < 0$, which is always true for a falling diver.
In this hybrid program model, when the diver lands ($x = 0$), they have a safe velocity $v \ge m$.

\subsection{A Structured Natural-Language Proof}
\label{sec:natlang-proof}
%TODO: Cite appendix with formal rules.
%TODO: Use abs in model
%TODO: Define all the subformulas and stuff
%TODO: Introduce "now" notation.
We introduce the reasoning principles for \dL by giving an informal structured proof of skydiver safety.
Beyond discrete program reasoning and first-order real arithmetic, the defining proof techniques of \dL are \emph{differential induction}~\cite{DBLP:journals/jar/Platzer16,DBLP:journals/corr/abs-1206-3357,DBLP:journals/logcom/Platzer10,DBLP:conf/tableaux/Platzer07} and \emph{differential ghosts}~\cite{DBLP:journals/corr/abs-1104-1987}.
Differential induction and ghosts are important because they enable rigorous proofs about ODEs whose solutions are outside the decidable fragment of arithmetic, in our example the drag equation.
Differential induction states that a formula $\phi$ is invariant if it is true initially and its differential $\der{\phi}$ is invariant, e.g. $\theta_1 > \theta_2$ is invariant if $\theta_1 > \theta_2$ initially and $\der{\theta_1} \geq \der{\theta_2}$ is invariant.
Differential ghosts enable augmenting ODEs with continuously-evolving ghost state.
If a formula $\phi$ is invariant, but not inductive, ghosts enable restating it as a formula that holds by differential induction.
This occurs, e.g. if $\theta_1$ converges asymptotically toward $\theta_2$ like velocity converges to the equilibrium in Model~\ref{model:dive}.

In this proof, we loosely follow the Lamport's~\cite{Lamport12,Lamport95} hierarchical style with numbered proof steps, but use our own keywords, such as \textsc{State} for giving names to states and \textsc{Introduce} for introducing differential ghost variables.
We restate in abbreviated form the theorem of Model~\ref{model:dive}:
\begin{theorem}[Skydiver Safety]
\label{thm:safe}
$\mathit{pre} \limply{[(\mathit{ctrl;plant}^*)](x=0\limply\abs{v} \leq m)}\ \text{is valid}$
\end{theorem}
%\usepackage{pf2}
%TODO: Rename T-> ep()/\epsilon
\newcommand{\proofinv}{\textsc{Invariant}:}
\newcommand{\proofbc}{\textsc{Base Case}:}
\newcommand{\proofis}{\textsc{Inductive Case}:}
\newcommand{\proofintro}{\textsc{Introduce}}
\newcommand{\proofstate}[1]{\textsc{State}~{\it #1}:}
\newcommand{\proofassign}{\textsc{Assign}:}
\newcommand{\ctrlcase}[1]{\textsc{Control Case}:~#1}
\pflongnumbers%{1}{0.7em}
\begin{proof}
\step{label-1}{
  \proofstate{init}~\textsc{Assume} $init(pre)$.
%  \assume{$init(pre)$}
}
%  \prove{$[(\mathrm{ctrl;plant})^*](x=0\limply\abs{v} \leq m)$  (by loop invariants and arithmetic)}}
\step{label-2}{\proofinv{}~$\mathit{const}$ vacuously because the free variables of $\mathit{const}$ are not written in $(\mathit{ctrl;plant})$. }
\step{label-3}{\proofinv{}~$\mathit{dc}$  by loop induction.\proofbc{}~by arithmetic and \stepref{label-1}. \\\proofis{}~ $\mathit{dc}\limply[\mathit{(ctrl;\pevolvein{plant}{dc})}]\mathit{dc}$ holds by the domain constraint.}
%        \begin{proof}
%        \step{label-2.2.1.1}{Execute $\mathit{ctrl}$.}
  %      \step{label-2.2.1.2} Call the resulting state \texttt{mid}.
  %      \step{label-2.2.1.2} $now(dc)$ holds vacuously.
%        \step{label-2.2.1.2}{Execute $\mathit{plant}$.}
%        \step{label-2.2.1.3}{$now(dc)$ holds by the domain constrant of $\mathit{plant}$.}
%        \end{proof}
%      \end{proof}
    \step{label-4}{\proofinv{}~$\mathit{dyn}$  by loop induction. \proofbc{} by arithmetic and \stepref{label-2}.
      \proofis{} We establish differential invariants for each control case.}
      \begin{proof}
      \step{label-4.1}{\proofstate{loop}. Here $\mathit{dc}, \mathit{const}, \mathit{dyn}$ hold by
            \stepref{label-2}, \stepref{label-3},\stepref{label-4}}
      \step{label-4.2}{
        \ctrlcase{$?(\rvar =\ravar \land \abs{v}+g\cdot\epsilon > \sqrt{\frac{g}{\rpvar}}); t:=0$}}
%        \pfsketch\ The case follows arithmetically from the differential invariants.
        \begin{proof}
        \step{label-4.2.1}{\proofinv{} $g>0\land{\rpvar>0}$ vacuously.}
        \step{label-4.2.2}{\proofinv{} $\abs{v} < \abs{loop(v)}+g\cdot{t}$ by differential induction.}
          \begin{proof}
          \step{label-4.2.2.1}{\proofbc{} by arithmetic and \stepref{label-4.2}.}
          \step{label-4.2.2.2}{\proofis{} $\D{\abs{v}} \leq \D{\abs{loop(v)}+g\cdot{t}}$ holds because $\der{v} = r\cdot{v^2}-g \geq -g = \der{loop(v)-g\cdot{t}}$ by arithmetic and \stepref{label-4.2.1}.}
          \end{proof}
        \step{label-4.2.3}{\proofinv{} $t\leq{\epsilon}\land{v\geq{0}}$ by domain constraint.}
        \step{label-4.2.4}{\proofinv{} $loop(\abs{v})+g\cdot{t}\leq loop(\abs{v})+g\cdot\epsilon$  by arithmetic, \stepref{label-4.2.1}, and \stepref{label-4.2.3}. }
        \step{label-4.2.5}{\proofinv{} $loop(\abs{v})+g\cdot \epsilon < \sqrt{\frac{g}{\rpvar}}$ by arithmetic and \stepref{label-4.2}.}
        \step{label-4.2.6}{Q.E.D. Because $\abs{v}> \sqrt{\frac{g}{\rpvar}}$ by transitivity, \stepref{label-4.2.2}, \stepref{label-4.2.4}, and \stepref{label-4.2.5}.}
         \pflonglabel{pfstep-trans}
        \end{proof}
      \step{label-4.3}{\ctrlcase{$r:=\rpvar; t:=0$}. Safety is invariant but not inductive: $\D{\abs{v}} > \left(-\sqrt{\frac{g}{\rpvar}}\right)'$.
        We augment the ODE with a \emph{differential ghost} variable, enabling an inductive invariant.}
        \begin{proof}
        \step{label-4.3.1}{\proofinv{} $g>0\land{\rpvar>0}$ is invariant.}
        \step{label-4.3.2}{\proofintro{} $y'=-\frac{\rpvar}{2}\cdot\left(\abs{v}+\sqrt{\frac{g}{\rpvar}}\right)\cdot{y}$ with $loop(y)^2\cdot\left(-\abs{v}+\sqrt{\frac{g}{\rpvar}}\right)=1$ for fresh variable $y$.
          This is sound since $y'$ is linear in $y$ and thus duration of the ODE is unchanged, and since $y$ is fresh the augmented ODE agrees with the original on all other variables.
        }
        \step{label-4.3.3}{\proofinv{} $y^2\cdot\left(-\abs{v}+\sqrt{\frac{g}{\rpvar}}\right)=1$ by differential induction. The base case holds by \stepref{label-4.3.2}; the inductive step holds by construction of $y'$ and arithmetic.}
        \step{label-4.3.4}{\proofinv{} $\abs{v} < \sqrt{\frac{g}{\rpvar}}$ because it is arithmetically equivalent to \stepref{label-4.3.3}.}
        \end{proof}
      \end{proof}
\step{label-5}{Q.E.D. Because \stepref{label-4} implies the postcondition by arithmetic.}
  \end{proof}
%\end{proof}
We mechanize this proof in Kaisar as a guiding example throughout the paper.
% N.B. r is just coefficient of drag, drag force itself is r\cdot v^2 \cdot actually some other stuff too.
% We don't have space to get super precise with the physics but force vs. coefficient might be a useful distinction still.
%\end{example}

\newcommand{\kwassume}{\textbf{assume}}
\newcommand{\kwassert}{\textbf{assert}}
\newcommand{\kwlet}{\textbf{let}}
\newcommand{\kwstate}{\textbf{state}}
\newcommand{\kwsolve}{\textbf{solve}}
\newcommand{\kwshow}{\textbf{show}}
\newcommand{\kwby}{\textbf{by}}
\newcommand{\kwusing}{\textbf{using}}
\newcommand{\kwnote}{\textbf{note}}
\newcommand{\kwhave}{\textbf{have}}
\newcommand{\kwinv}{\textbf{inv}}
\newcommand{\kwghost}{\textbf{Ghost}}
\newcommand{\kwind}{\textbf{Ind}}
\newcommand{\kwpre}{\textbf{Pre}}
\newcommand{\kwfinally}{\textbf{finally}}
\newcommand{\kwassign}{\textbf{assign}}
\newcommand{\kwmid}{\textbf{after}}
\newcommand{\kwfirst}{\textbf{have}}
\newcommand{\kwthen}{\textbf{then}}
\newcommand{\kwcase}{\textbf{case}}
\newcommand{\kwfocus}{\textbf{focus}}

%TODO: Read up on Mizar's semantic correlates
\newcommand{\uproof}{{\tt UP}}
\newcommand{\fproof}{{\tt FP}}
\newcommand{\iproof}{{\tt IP}}
\newcommand{\vproof}{{\tt VP}}
\newcommand{\brule}{{\tt BR}}
\newcommand{\prule}{{\tt PR}}
\newcommand{\drule}{{\tt DR}}
\newcommand{\slet}[3]{\kwlet~#1~=~#2~#3}
\newcommand{\snote}[3]{\kwnote~#1~=~#2~#3}
\newcommand{\have}[4]{\kwhave~#1:#2~#3~#4}
\newcommand{\sshow}[2]{\kwshow~#1~#2}
\newcommand{\bassign}[3]{\kwassign~#1:=#2~#3}
\newcommand{\dassign}[3]{\bassign{#1}{#2}{#3}}
\newcommand{\bassignany}[2]{\kwassign~\prandom{#1}~#2}
\newcommand{\bcase}[4]{\left(\kwcase~#1~\Rightarrow~#2~|~#3~\Rightarrow~#4\right)}
\newcommand{\pcaseol}[4]{\left(\kwcase~#1~\Rightarrow~#2~|~#3~\Rightarrow~#4\right)}
\newcommand{\pcaseor}[3]{\left(\kwcase({#1}\lor{#2})~#3\right)}
\newcommand{\pcaseiffllr}[4]{\left(\kwcase~\limply{#1}~\Rightarrow~#3~|~{#2}\limply~\Rightarrow~#4\right)}
\newcommand{\pcaseifflrl}[4]{\left(\kwcase~\limply{#2}~\Rightarrow~#3~|~{#1}\limply~\Rightarrow~#4\right)}
\newcommand{\pcaseiffr}[4]{\left(\kwcase({#1}\bimply{#2})~{\it{of}}~L~\Rightarrow~#3~|~R~\Rightarrow~#4\right)}
\newcommand{\pcasediaiter}[4]{\left(\kwcase({#2}\lor\ddiamond{#1}{\ddiamond{#1^*}{#2}})~{\it{of}}~L~\Rightarrow~#3~|~R~\Rightarrow~#4\right)}
\newcommand{\pcaseimpl}[4]{\left(\kwcase~\limply{#1}~\Rightarrow~#3~|~{#2}\limply~\Rightarrow~#4\right)}
\newcommand{\pcasear}[4]{\left(\kwcase({#1}\land{#2})~{\it{of}}~L~\Rightarrow~#3~|~R~\Rightarrow~#4\right)}
\newcommand{\bcon}[3]{\kwmid{}~\{#2\}~\kwfirst{}~#1~\kwthen{}~\{#3\}}
\newcommand{\tinycolon}{\colon{}\hspace{-0.02in}}
\newcommand{\dsolve}[5]{\kwsolve~#1~ t:#2~dc:#3~#4~#5}
\newcommand{\bsolve}[4]{\kwsolve~#1~ t\tinycolon{}#2~\ dom\tinycolon{}#3~#4}
\newcommand{\bassert}[3]{\kwassume~#1:#2~#3}
\newcommand{\dassignany}[3]{\kwassign~\prandom{#1}~#2~#3}
\newcommand{\dcon}[3]{\bcon{#1}{#2}{#3}}
\newcommand{\dcase}[3]{\left(\kwcase~(#1\cup#2)~{of}~#3\right)}
\newcommand{\dassert}[4]{\kwassert~#1:#2~\{~#3~\}~#4}
\newcommand{\method}{{\tt method}}
%TODO: Generalize axiom to support metadata
\newcommand{\axiom}[4]{{\tt axiom}~#1~=~#2~#3~#4}
\newcommand{\ident}{{\texttt{ident}}}
\newcommand{\freeident}{{\texttt{ident{\textunderscore}}}}
\newcommand{\wildpat}{{\texttt{\textunderscore}}}
\newcommand{\using}{{\kwusing}}
\newcommand{\termone}{\theta_1}
\newcommand{\termtwo}{\theta_2}
\newcommand{\pat}{p}
\newcommand{\qat}{q}
\newcommand{\by}{\kwby{}}
\newcommand{\smethod}{method}
\newcommand{\simp}{{\tt simp}}
\newcommand{\auto}{{\tt auto}}
\newcommand{\rcf}{\ensuremath{\mathbb{R}}}
\newcommand{\closeid}{{\tt id}}
\newcommand{\finally}[1]{\kwfinally~#1}
\newcommand{\sinv}[5]{\kwinv~#1:#2\{\kwpre~\Rightarrow~#3~|~\kwind~\Rightarrow~#4\}~#5}
\newcommand{\scon}[5]{Inv~#1~\varphi(v)=#2 \{Pre~\Rightarrow~#3~|~Post~\Rightarrow~#4~|~Inv~\Rightarrow~#5\}}
\newcommand{\sghost}[4]{\kwghost~\humod{#1}{#3};#1'=#2~#4}
\newcommand{\sstate}[2]{\kwstate~#1~#2}
\newcommand{\upchk}[3]{#2 : (#1\vdash#3)}
\newcommand{\fpchk}[3]{#2 : (#1\vdash#3)}
\newcommand{\fpchks}[4]{#2:(#1\vdash_{#4}#3)}
% \ppchk{\G}{\sproof}{\Delta}
% \tpckh{H1}{G}{sproof}{Delta}{H2}
\newcommand{\ppchk}[3]{#2 : (#1\vdash{#3})}
\newcommand{\tpchk}[5]{\left({#2}\leadsto{#5}\right)#3 : \left({#1}\vdash{#4}\right)}
\renewcommand{\G}{\Gamma}
\newcommand{\GE}{\Gamma\vdash}
\newcommand{\GDE}{\Gamma;\Delta;H\vdash}
\newcommand{\ematch}[2]{\textbf{match}(#1,#2)}
\newcommand{\emmatch}[3]{\textbf{match}_{#1}(#2,#3)}
\newcommand{\emmmatch}[4]{\textbf{match}_{#1|#2}(#3,#4)}
\newcommand{\match}[3]{\textbf{match}(#1,#2)=#3}
\newcommand{\mmatch}[4]{\textbf{match}_{#1}(#2,#3)=#4}
\newcommand{\mmmatch}[5]{\textbf{match}_{#1|#2}(#3,#4)=#5}
\newcommand{\nmatch}[2]{\match{#1}{#2}{\bot}}
\newcommand{\nmmatch}[3]{\mmatch{#1}{#2}{#3}{\bot}}
\newcommand{\nmmmatch}[4]{\mmmatch{#1}{#2}{#3}{#4}{\bot}}
\newcommand{\eval}{\downarrow}
\newcommand{\Eval}{\Downarrow}
\newcommand{\laze}[3]{#1\vdash#2\eval#3}
\newcommand{\eage}[3]{#1\vdash#2\Eval#3}
\newcommand{\eeag}[2]{\bar{#2}_{#1}}
\newcommand{\longeag}[1]{\overline{#1}}
\newcommand{\longeeag}[2]{\overline{#2}_{#1}}
\newcommand{\eag}[1]{\bar{#1}}
\newcommand{\ext}[3]{#1,#2{\hskip -0.02in}\colon{}{\hskip -0.03in}#3}
\newcommand{\rext}[3]{#1,#2{\hskip -0.02in}\colon{}{\hskip -0.03in}#3}
\newcommand{\lext}[3]{#1{\hskip -0.02in}\colon{}{\hskip -0.03in}#2,#3}
\newcommand{\csing}[2]{\{#1{\hskip -0.02in}\colon{}{\hskip -0.03in}#2\}}
\newcommand{\cemp}{\{\}}
\newcommand{\oldcase}[1]{\textbf{Case}~{#1}:}

\renewcommand{\to}{\rightarrow}
\newcommand{\hhtime}[2]{#1,#2}

%Cite Paulson CPP keynote
%Cite Wenzel here again, don't repeat myself though
%As remarked by Wenzel, the key task of a structured proof language (vs. an unstructured one) is to provide block structuring while managing and respecting scope.

\section{First-Order Kaisar}
\label{sec:fok}
We present the first-order (real-arithmetic, i.e. \folr) fragment of Kaisar alone before considering dynamic logic.
The examples of Figure~\ref{fig:fok-ex} demonstrate the key features of First-Order Kaisar:
\begin{itemize}
\item Block-structuring elements (\kwhave{}, \kwnote{}, \kwlet{}, \kwshow{}) introduce intermediate facts, definitions, and conclusions.
 Facts $\phi$ are assigned names $x:\phi$ for later reference.
\item Unstructured proof methods (\kwusing{} $\mathit{facts}$ \kwby{} $\langle{\mathit{method}}\rangle$) close the leaves of proofs.
\item Backward-chaining propositional rules (\kwassume{}, \kwcase{}) decompose logical connectives.
\item Forward-chaining proof terms (\kwnote{} $x = \fproof$) make arithmetic lemmas convenient.
\item Patterns and abbreviations (\kwlet{} $x\_ = \ldots$) improve conciseness.
\end{itemize}
All of these features have appeared previously in some form in the literature~\cite{DBLP:conf/mkm/Wenzel06}.
Here we make clear any differences specific to Kaisar and in doing so lay a solid foundation for the development of nominals.
We do so by expanding upon the examples with formal syntax and semantics, given as a proof-checking relation.
Note that the examples in this section, because they fall under \folr, are in principle decidable.
For that reason we focus here on techniques that vastly improve the speed of decision procedures (arithmetic proving with \kwnote{} and \kwhave{}), and which generalize to the undecidable fragments of \dL handled in Sections~\ref{sec:ddk} and \ref{sec:dlk} (block-structuring and pattern-matching).

\paragraph{Examples}
%In order to focus on the first-order fragment of Kaisar, we begin with a proof of the main arithmetic lemma in the proof of Model~\ref{model:dive}.
%First, consider the high-level proof outline.
%The two control cases (open or closed parachute) translate into two proof cases, where the argument for the closed case is as follows:
%\begin{enumerate}
%\item Velocity $v$ is bounded below by the velocity one would have if there was no drag ($v - g\cdot{T}$). 
%\item The bound for time $t$ exceeds that at the maximum timestep $T$ ($v - g\cdot{t} \leq v-g\cdot{T}$).
%\item The bound for time $T$ exceeds the open-parachute equilibrium velocity $-\sqrt{g/p_r}$.
%\item By transitivity the velocity is bounded by $-\sqrt{g/p_r}$.
%\end{enumerate}
%Steps (1-3) all use continuous reasoning and thus will be handled in Section~\ref{sec:dlk}.
For a first-order example, consider the transitivity reasoning of Step~\pfref{pfstep-trans} in Section~\ref{sec:natlang-proof}.
Figure~\ref{fig:fok-ex} presents four proofs of the  following arithmetic subgoal in \dL:
%TODO: Use nominals already?
\[\abs{v} \leq loop(\abs{v})+g\cdot{t} \limply 
 loop(\abs{v})+g\cdot{t} \leq loop(\abs{v})+g\cdot{\Tvar} \limply
 loop(\abs{v})+g\cdot{\Tvar} < \sqrt{\frac{g}{\rpvar}} \limply 
 \abs{v}<\sqrt{\frac{g}{\rpvar}}\]
Figure~\ref{fig:fok-ex} develops a proof as one might in interactive proof.
Examples 1a and 1b both appeal immediately to an arithmetic solver, with Example 1b using pattern-matching to introduce concise names, e.g. ${\tt vt\_}=loop(\abs{v})+g\cdot{t}$.
Pattern-matching makes Example 1b more flexible: it will work even if the definition of $\mathit{vt\_}, \mathit{vEps\_},$ or $\mathit{vBound\_}$ changes.
However, the $\mathbb{R}$ will be prohibitively slow if $\mathit{vt\_}, \mathit{vEps\_}, \mathit{vBound\_}$ become too large (all reasoning is performed on the expanded terms).
Example 1c restores speed by isolating the transitivity axiom in the fact {\tt trans} and instantiating it with \kwnote{} to recover the result.
The \kwnote{} step supplies term inputs $\mathit{v}, \mathit{vt\_}, \mathit{vEps\_}, \mathit{vBound\_}$ first, then propositions {\tt v, gt, gEps}, matching the structure of the formula {\tt trans}.
The {\tt id} proof method closes the proof once we provide a proof of the goal with \kwusing{}.

The \kwlet{} statement reuses the pattern-matching mechanism to introduce definitions for readability at will.
Example 1d uses \kwlet{} to ensure that the bound {\tt vBound\_} is specifically $\sqrt{\frac{g}{\rpvar}}$.
This provides machine-checked documentation if, e.g. we did not intend this proof to be general or if the proof depended on the value of {\tt vBound\_}.
%does the same, and simply showcases the interaction between \kwlet{} and pattern-matching.
%introduces new definitions 
% allows us to introduce new definitions. % \kwassume{}
%For a longer example proof, see %TODO Write a section on this.
\begin{figure}
  \centering
  \begin{tabular}{cc}
    \begin{minipage}{0.5\textwidth}
%\begin{verbatim}
\verb|# Example 1a|\\
\kwassume\verb| v:| $\abs{v} \leq loop(\abs{v})+g\cdot{t}$\\
\kwassume\verb| gt:| $loop(\abs{v})+g\cdot{t} \leq loop(\abs{v})+g\cdot\Tvar$\\
\kwassume\verb| gEps:| $loop(\abs{v})+g\cdot\Tvar < \sqrt{\frac{g}{\rpvar}}$\\
\kwshow\ \ $\abs{v} < \sqrt{\frac{g}{\rpvar}}$\\
\verb|  |\kwby\ $\mathbb{R}$\\
\verb||\\
\verb|# Example 1c|\\
\kwassume\verb| v:|    $\abs{v} \leq vt\_$ \\
\kwassume\verb| gt:|   $vt\_ \leq v\Tvar{}\_$\\
\kwassume\verb| gEps:| $v\Tvar{}\_ <  vBound\_$\\
\kwhave\verb| trans:|\\
\verb| |$\forall w x y z~(w\leq{x}\limply{x}\leq{y}\limply{y}<{z}\limply{w}<{z})$\\
\verb|  |\kwby\ $\mathbb{R}$\\
\kwnote\verb| res =| \\
\verb| trans|$\ \ v\ \ vt\_\ \ v\Tvar{}\_\ \ vBound\_\ \ {\tt{v}}\ \ {\tt{gt}}\ \ {\tt{gEps}}$\\
\kwshow\verb| |$\abs{v} < vBound\_$\\
\verb|  |\kwusing\ {\tt{res}} \kwby{} $id$\\
%\end{verbatim}
    \end{minipage}
&    \begin{minipage}{0.5\textwidth}
\verb|# Example 1b|\\
\kwassume\verb| v:|$\abs{v} \leq vt\_$\\
\kwassume\verb| gt:|$vt\_ \leq v\Tvar{}\_$\\
\kwassume\verb| gEps:| $v\Tvar{}\_ < vBound\_$\\
\kwshow\ $v > vBound\_$\\	
\verb|  |\kwby{}\ $\mathbb{R}$\\
\verb||\\
\verb|# Example 1d|\\
\kwlet\verb| vBound_ =| $\sqrt{\frac{g}{\rpvar}}$\\
\kwassume\verb| v:|    $\abs{v} \leq vt\_$\\
\kwassume\verb| gt:|   $vt\_ \leq v\Tvar{}\_$\\
\kwassume\verb| gEps:| $v\Tvar{}\_ < vBound\_$\\
\kwhave\verb| trans:|\\ 
$\forall w x y z~(w\leq{x} \limply x\leq{y} \limply y<z \limply w<z)$\\
\verb|  |\kwby{}\ $\mathbb{R}$\\
\kwnote{}\verb| res = |\\
\verb| trans|$\ \ v\ \ vt\_\ \ v\Tvar{}\_\ \ vBound\_\ \ {\tt{v}}\ \ {\tt{gt}}\ \ {\tt{gEps}}$\\
\kwshow\  $\abs{v} < vBound\_$ \\
\verb|  |\kwusing{}\ {\tt{res}} \kwby{}\ $id$\\
\verb||\\
    \end{minipage}\\[0.5in]
    \begin{minipage}{0.5\textwidth}
    \end{minipage}
&    
  \end{tabular}
    \caption{Kaisar Proofs of First-Order Example}
\label{fig:fok-ex}
\end{figure}

\paragraph{Definitions}
The proof-checking judgements of Kaisar maintain a context $\Gamma$ which maps names to assumptions (or conclusions for succedents $\Delta$) and abbreviations introduced through pattern-matching.
Abbreviation variables are suffixed with an underscore , so $\Gamma \equiv \{xz \mapsto (x > 0), {\tt{y\_}} \mapsto y_0 + 5\}$ means we have the assumption that $x > 0$ via the name $xz$ and we have abbreviated $y_0 + 5$ as $y$.
Throughout the paper, subscripts are mnemonic, e.g. $\G_\phi$ could be a context associated with $\phi$ in any way.
Kaisar expands abbreviations before reasoning to ensure no extensions to \dL are necessary.
%Abbreviations are part of Kaisar only, not \dL.
Substitution of $x$ with $\theta$ in $\phi$ is denoted $\subst[\phi]{x}{\theta}$.
%Abbreviations are an extra-logical feature: they do not effect the expressive power, but can greatly affect convenience.

\newcommand{\GHE}{{\G;H}\vdash}
%We also maintain a trace $H$ of state changes.  Because state does not change in first-order Kaisar, we delay further treatment of traces to Section~\ref{sec:ddk}.
The primary proof-checking judgement $\ppchk{\G}{\sproof}{\Delta}$ says that $\sproof$ is a structured proof of the classical sequent $\G\vdash\Delta$ (i.e. the formula $\bigwedge_{\phi\in\Gamma}{\phi}\limply\bigvee_{\psi\in\Delta}\psi$).
In proving safety theorems $\phi\limply[\alpha]\psi$, the succedent $\Delta$ typically consists of a single formula.
The auxilliary judgement $\fpchk{\G}{\fproof}{\Delta}$ for forward-chaining proofs is analogous.
In Section~\ref{sec:ddk} we will extend the judgement $\ppchk{\G}{\sproof}{\Delta}$ with \emph{static execution traces} to support structured symbolic execution.
%before and after executing $\sproof$, respectively.
%Static execution traces will be discussed further in Section~\ref{sec:ddk}, where they are used to describe variable updates in symbolic execution.
%We delay discussion until then: the first-order fragment binds variables only in forward, not backward proofs. 
%because the propositional backward-chaining rules discussed in this section do not bind variables.
%The resulting state $H'$ is only used to support references to prior states in Hoare-style midcondition reasoning and is conservative in some cases, e.g. in branching proofs where different branches have different traces.
%Note forward proofs are free to instantiate quantifiers --- those instantiations just do not affect the trace, which is connected to the overall backward-chaining proof.
%Even 
%As we see in Section~\ref{sec:ddk}

Expressions written in proofs can contain abbreviations, but sequents, being pure \dL, do not.
We write $\eeag{\G}{e}$ (or, when clear, simply $\eag{e}$) for the \emph{expansion} of $e$ which replaces all abbreviations with their values.
%The above use expansion of extended expressions $\tilde{e}$ (expressions with abbreviations) to proper expressions $e$ with the helper judgement ${\G;H}\vdash \tilde{e}\eval{e}$ and the
We also use matching of patterns $p$ with $\mmatch{\G}{\pat}{e}{\G_\pat}$ where $\G_\pat$ extends $\G$ with bindings produced by matching $e$ against $p$. 

\paragraph{Block Structure}
The top level of a Kaisar proof is a \emph{structured proof} (\sproof)
\begin{align*} 
\sproof ::= ~&\prule\ |\ \slet{p}{e}{\sproof}\ |\ \snote{x}{\fproof}{\sproof}\ |\ \have{x}{e}{\sproof_1}{\sproof_2}\ |\ \sshow{p}{\langle{method}\rangle}
\end{align*}
The \emph{propositional rules} (\prule) \kwassume{} and \kwcase{} perform propositional reasoning.
The \kwcase{} construct is multi-purpose, supporting the connectives $\vee, \bimply, \land$.
Here a vertical bar $|$ outside parentheses separates syntactic productions, while a bar inside parentheses is syntax to separate cases:
\begin{align*}
\prule ::= ~&\bassert{x}{\pat}{\sproof}\ |\ \pcasear{x\tinycolon{}\pat_\phi}{y\tinycolon{}\pat_\psi}{\sproof_1}{\sproof_2}\ \\
 &|\ \pcaseor{x\tinycolon{}\pat_\phi}{y\tinycolon{}\pat_\psi}{\sproof}\ |\ \pcaseiffr{x\tinycolon{}\pat_\phi}{y\tinycolon{}\pat_\psi}{\sproof_1}{\sproof_2}
\end{align*}
We give only right rules: left rules are analogous and derivable with the \kwfocus{} construct of Section~\ref{sec:ddk}.
In each rule we assume fresh variables $x,y$:
\begin{center}
  {\small\begin{tabular}{cc}
      \infer{\ppchk{\G}{(\bassert{x}{\pat}{\sproof})}{(\psi\limply\phi),\Delta}}
      {\mmatch{\G}{\pat}{\psi}{\Gamma_\psi} &
        \ppchk{\ext{\G_\psi}{x}{\psi}}{\sproof}{\phi,\Delta}}
      &
      \infer{\ppchk{\G}{\pcaseiffr{p}{q}{\sproof_1}{\sproof_2}}{(\phi\bimply\psi),\Delta}}
      {   \deduce{\ppchk{\rext{\G_1}{x}{\phi}}{\sproof_1}{\lext{y}{\psi}{\Delta}}\hskip 0.1in \ppchk{\rext{\G_1}{y}{\psi}}{\sproof_2}{\lext{x}{\phi}{\Delta}}} 
             {\mmatch{\G}{p\bimply{q}}{\psi\bimply\phi}{\Gamma_1}} 
        }
%      \infer{\ppchk{\G}{\pcaseiffr{x\colon{}\phi}{y\colon{}\psi}{\sproof_1}{\sproof_2}}{(\phi\bimply\psi),\Delta}}
 %     {\ppchk{\rext{\G}{x}{\phi};H}{\sproof_1}{\lext{y}{\psi}{\Delta}} & \ppchk{\rext{\G}{y}{\psi}}{\sproof_2}{\lext{x}{\phi}{\Delta}}}
\\[0.05in]
      % \infer{\ppchk{\G,\phi\lor\psi;H}{\pcaseol{x:\phi}{y:\psi}{\sproof_1}{\sproof_2}}{\Delta}{H}}
      % {\ppchk{\rext{\G}{x}{\phi};H}{\sproof_1}{\Delta}{H_1} & \ppchk{\rext{\G}{y}{\psi};H}{\sproof_2}{\Delta}{H_2}}\\[0.05in]
      \infer{\ppchk{\G}{\pcaseor{x:\phi}{y:\psi}{\sproof}}{(\phi\lor\psi),\Delta}}
         {\ppchk{\G}{\sproof}{\lext{x}{\phi}{\lext{y}{\psi}{\Delta}}}}
      &
      \infer{\ppchk{\G}{\pcasear{p}{q}{\sproof_1}{\sproof_2}}{(\phi\land\psi),\Delta}}
      {\deduce{\ppchk{\G_1}{\sproof_1}{\lext{x}{\phi}{\Delta}}
               \hskip 0.1in \ppchk{\G_1}{\sproof_2}{\lext{y}{\psi}{\Delta}}}
         {\mmatch{\G}{p\land{q}}{\psi\land\phi}{\Gamma_1}} 
         }
      % \infer{\ppchk{\G,\phi\bimply\psi;H}{\pcaseiffllr{x:\phi}{y:\psi}{\sproof_1}{\sproof_2}}{\Delta}{H}}
      % {\ppchk{\rext{\G}{x}{\phi};H}{\sproof_1}{\Delta}{H_1} & \ppchk{\rext{\G}{y}{\psi};H}{\sproof_2}{\Delta}{H_2}}\\[0.05in]
      % \infer{\ppchk{\G,\phi\bimply\psi;H}{\pcaseifflrl{x:\phi}{y:\psi}{\sproof_1}{\sproof_2}}{\Delta}{H}}
      % {\ppchk{\G;H}{\sproof_1}{\lext{y}{\psi}{\Delta}}{H_1} & \ppchk{\rext{\G}{x}{\phi};H}{\sproof_2}{\Delta}{H_2}}\\[0.05in]
      % \infer{\ppchk{\G,\phi\limply\psi;H}{\pcaseimpl{x:\phi}{y:\psi}{\sproof_1}{\sproof_2}}{\Delta}{H_1}}
      % {\ppchk{\G;H}{\sproof_1}{\lext{x}{\phi}{\Delta}}{H_1} & \ppchk{\rext{\G}{y}{\psi};H}{\sproof_2}{\Delta}{H_2}}\\[0.05in]
    \end{tabular}  }
\end{center}
% (in the succedent). and $\limply$.
% (in the antecedent).
%\pcaseiffllr{\phi}{\psi}{\sproof_1}{\sproof_2}\

%The correct variant of \kwcase{} is determined by pattern-matching.
%TODO: Rest of prules
%
%
% \pcaseol{\phi}{\psi}{\sproof_1}{\sproof_2}
% \pcaseor{\phi}{\psi}{\sproof}
% \pcaseiffllr{\phi}{\psi}{\sproof_1}{\sproof_2}
% \pcaseifflrl{\phi}{\psi}{\sproof_1}{\sproof_2}
% \pcaseiffr{\phi}{\psi}{\sproof_1}{\sproof_2}
% \pcaseimpl{\phi}{\psi}{\sproof_1}{\sproof_2}
% \pcasear{\phi}{\psi}{\sproof}
%
%\newcommand{\pcaseol}[4]{\left(\kwcase{}~#1~\Rightarrow~#3~|~#2~\Rightarrow~#4\right)}
%\newcommand{\pcaseor}[3]{\left({\tt case #1 \ensuremath{\rm\vee} #2 }~~\Rightarrow~#3\right)}
%\newcommand{\pcaseiffllr}[4]{\left(\kwcase{}~\limply{#1}~\Rightarrow~#3~|~{#2}\limply~\Rightarrow~#4\right)}
%\newcommand{\pcaseifflrl}[4]{\left(\kwcase{}~\limply{#2}~\Rightarrow~#3~|~{#1}\limply~\Rightarrow~#4\right)}
%\newcommand{\pcaseiffr}[4]{\left(\kwcase{}~{#1}\limply{#2}~\Rightarrow~#3~|~{#2}\limply{#1}~\Rightarrow~#4\right)}
%\newcommand{\pcaseimpl}[4]{\left(\kwcase{}~\limply{#1}~\Rightarrow~#3~|~{#2}\limply~\Rightarrow~#4\right)}
%\newcommand{\pcasear}[4]{\left(\kwcase{}~#1~\Rightarrow~#3~|~#2~\Rightarrow~#4\right)}
%
%
%We give the propositional rules in Figure~\ref{fig:prop}.
Pattern-matching reduces verbosity vs. writing complete formulas.
In Example 1b, we shorten the Kaisar proof of Example 1a by applying variable patterns $x\_$ and $y\_$ to the assumption (thus matching any conjunction) and a wildcard pattern \verb|_| to the goal (thus matching any goal).
%Here $\G,\G_1$ denotes $\G$ extended with any definitions introduced by the pattern-match $\match{p}{\phi}{\G_1}$:
%TODO: Maybe better not figure for space reasons
\begin{center}
  {\small\begin{tabular}{ccc}
\infer{\ppchk{\G}{\have{x}{\psi}{\sproof_1}{\sproof_2}}{\phi,\Delta}}
  {
 \ppchk{\G}{\sproof_1}{\eag{\psi}} & 
 \ppchk{\ext{\Gamma}{x}{\eag{\psi}}}{\sproof_2}{\phi,\Delta}}
&    \infer{\ppchk{\G}{\slet{p}{e}{\sproof}}{\phi,\Delta}}
  {\mmatch{\G}{p}{\eag{e}}{\G_1} & \ppchk{\G_1}{\sproof}{\phi,\Delta}}
&\infer{\ppchk{\G}{\snote{x}{\fproof}{\sproof}}{\phi,\Delta}}
  {\fpchk{\G}{\fproof}{\psi} & {\ppchk{\ext{\Gamma}{x}{\psi}}{\sproof}{\phi,\Delta}}}
  \end{tabular}}
\end{center}
%\[\infer{\ppchk{\G}{\have{x}{\psi}{\sproof_1}{\sproof_2}}{\phi,\Delta}}
%  {
% \ppchk{\G}{\sproof_1}{\eag{\psi}} & 
% \ppchk{\ext{\Gamma}{x}{\eag{\psi}}}{\sproof_2}{\phi,\Delta}}\]
The \kwhave{} construct cuts, proves, and names an intermediate fact $\psi$ with $\sproof_1$, then continues the main proof $\sproof_2$.
%In Example 1c, note the concrete syntax for \kwhave{}: if the fact $\phi$ is provable in one step by $\kwshow{}$, we write the contents of $\kwshow{}$ directly
%TODO: Explain good
When $\sproof_1$ immediately \kwshow{}'s $\phi$, as in Example 1c, we omit the \kwshow{} keyword in the concrete syntax.
The use of \kwhave{} in Example 1c is typical: the direct application of $\mathbb{R}$ in Example 1(a,b) does not scale if $\mathit{vt\_}, mathit{v\Tvar{}}, \mathit{vBound\_}$ are large terms.
By isolating the ${\mathit{trans}}$ axiom with \kwhave{}, we enable the arithmetic reasoning to scale.
%Furthermore, proving facts which we believe ought to be true is an easy way to check our intuition even in simpler proofs.

The \kwlet{} keyword, as used in Example 1d, performs a general-purpose pattern-match, which in this case simply binds $goal{\tt\_}$ to $v > -\sqrt{\frac{g}{\rpvar}}$, then continues the proof $\sproof$.
%\[\infer{\ppchk{\G}{\slet{p}{e}{\sproof}}{\phi,\Delta}}
%  {\match{p}{\eag{e}}{\G_1} & \ppchk{\G,\G_1}{\sproof}{\phi,\Delta}}\]
The \kwnote{} construct is similar to \kwhave{}, except that the intermediate fact is proven by a forward-chaining proof term.
It is often convenient for instantiating (derived) axioms or performing propositional reasoning, as in Example 1c.
Because \kwnote{} uses a forward-chaining proof, we need not specify the proven formula, but rather the proof-checking judgement synthesizes it as an output.

%In Example 1d we use it to prove $x+y+z>1$ by conjunction elimination:
%\[\infer{\ppchk{\G}{\snote{x}{\fproof}{\sproof}}{\phi,\Delta}}
%  {\fpchk{\G}{\fproof}{\psi} & {\ppchk{\ext{\Gamma}{x}{\psi}}{\sproof}{\phi,\Delta}}}\]
%TODO: Optional formula in note
%Thus, any optional pattern annotation in \kwnote{} functions as in \kwshow{}, merely allowing the user to confirm that they proved what they intended to.

%TODO: Should we even call this unstructured proof?
\paragraph{Unstructured Proof Methods}
As shown in Example 1, \kwshow{} closes a proof leaf by specifying facts (\kwusing{}) and an automatic proof method ($\mathbb{R}$, {\tt id}, {\tt auto}).
In contrast to other languages~\cite{DBLP:conf/mkm/Wenzel06}, these methods are the only unstructured language construct, as we have found no need for a full-fledged unstructured language.
The \ident{} method is extremely fast but expects the conclusion to appear verbatim in the context or \kwusing{} clause.
The $\mathbb{R}$ method invokes a $\folr{}$ decision procedure~\cite{PCAD,CAD}.
We typically assist $\mathbb{R}$ on difficult goals by specifying facts with a \kwusing{} clause for speed.
%To finish within an acceptable time,  is typically assisted on difficult goals by .
The \kwusing{} clause can specify both assumptions from the context and additional facts by forward proof.
When the \kwusing{} block is empty, it defaults to the entire context.
The {\tt auto} method applies general-purpose but incomplete proof heuristics including symbolic execution and also benefits from \kwusing{}.
Note the {\tt auto} method does not have a simple proof rule.
This is okay at the unstructured proof level: any sound proof method is adequate for closing leaves of a proof.\footnote{The {\tt auto} method (tactic) of \KeYmaeraX is sound because it uses only operations of a sound LCF-style core~\cite{DBLP:conf/cade/FultonMQVP15}.}

\newcommand{\upfacts}[2]{{\textbf{facts}(#1,#2)}}
$\upfacts{{\pat}s}{{\fproof}s}$ defines the facts available to the proof method (assumptions and \fproof{} conclusions).
% denotes the union of the conclusions of each $\fproof$ with assumptions matching any $\pat_i$.
The pattern $q$ selects a conclusion from $\Delta$, else we default to the entire succedent.
% uses an unstructured proof method to finish a proof after the preceding proof steps have simplified it.
%TODO: Use good notation
%   \Gamma,\Gamma_1;H\vdash \uproof:\Delta_1,\phi,\Delta_2}
 \begin{center}
$ \upfacts{{\pat}s}{{\fproof}s} \equiv \{\phi\in\G~|~\mexists{i}\emmatch{\G}{\pat_i}{\phi}\}\cup\{\phi~|~\mexists i~\fpchk{\G}{\fproof_i}{\phi}\}$\\[0.12in]
   \begin{tabular}{cc}
     \infer{\ppchk{\sshow{\qat}{\using~{\pat}s~{\fproof}s~\by~\mathbb{R}}}{\G}{\phi,\Delta}}
     {\ematch{\qat}{\phi} & \upfacts{{\pat}s}{{\fproof}s}~\text{valid in}~\folr{} }&
     \infer{\ppchk{\sshow{\qat}{\using~{\pat}s~{\fproof}s~\by~\closeid}}{\G}{\phi,\Delta}}
     {\ematch{\qat}{\phi} & \phi \in \upfacts{{\pat}s}{{\fproof}s}}
   \end{tabular}
 \end{center}
%The $\closeid$ method is very fast but works only if the exact conclusion is also an assumption, while $\rcf$ uses a complete decision procedure for real-closed fields, and {\tt auto} applies general-purpose but incomplete proof heuristics including symbolic execution.
%TODO: Torn about proof modedness
%\begin{align*} 
%\uproof ::=&~[\using~(\pat~|~\fproof)^*]~\by~\smethod&\\
%\smethod ::=&~\simp\ |\ \auto\ |\ \rcf\ |\ \closeid\\
%\end{align*}
%TODO: get rid of UP and method judgements.
%\[\infer{\GHE \using~\pat_1\cdots\pat_n~\fproof_1\cdots\fproof_m~\method : \Delta}
%%{\Gamma_1;H\vdash \method: \Delta &(\Gamma_1\equiv \{\phi \in \Gamma~|~\exists i~ (\ematch{\pat_i}{\phi} \}\cup\{\phi~|~\exists i~\GHE\fproof_i: \phi\})}\]
%The \kwusing{} block can contain both patterns and forward proofs.
%Any patterns provided are used to search the context for relevant assumptions. 
%Any forward proofs are executed and cut in.
%\begin{center}
%{\small\begin{tabular}{c@{\hskip 0.5in}c}
%%\infer{\Gamma,\phi\vdash \closeid : \phi,\Delta}{} & \infer{\Gamma\vdash \mathbb{R}:\Delta}{(\wedge_{\Gamma}\limply \vee_{\Delta}~\text{valid in}~\folr)}
%  \end{tabular}}
%\end{center}
%The {\tt id} method applies the hypothesis rule, while $\mathbb{R}$ invokes a decision procedure for $\folr$.

\paragraph{Pattern Matching}
Throughout the examples of Figure~\ref{fig:fok-ex}, we use pattern-matching to describe the shapes of expressions and select assumptions for use in automation.
The above features suffice for structural proof steps.
%The above features are sufficient for decomposing programs.
The pattern language of Kaisar is defined inductively:
\begin{align*} 
\pat ::=& \ident\ |\ p(vars)\ |\ p(\neg vars)\ |\ \_\ |\ \pat\cup\pat\ |\ \pat\cap\pat\ |\ \neg\pat\ |\ \otimes(e,f)
\end{align*} 
Pattern-matching is formalized as a judgement $\mmatch{\G}{\pat}{e}{\G_\pat}$, where $\G_\pat$ extends $\G$ with bindings resulting from the match (we omit $\G_\pat$ when it is not used and omit $\G$ when it is clear from context).
We write $\nmatch{\pat}{e}$ when no pattern-matching rule applies.
In the definition below, we use the notation $\otimes(e,f)$ to generically say that all operators $\otimes$ (where the arguments are expressions $e,f$) of the \dL language are supported in patterns.
The \dL operators are matched structurally.
\[\infer{\mmatch{\G}{\otimes(\pat,\qat)}{\otimes(e,f)}{\G_\qat}}
 {\mmatch{\G}{\pat}{e}{\G_\pat} & \mmatch{\G_\pat}{\qat}{f}{\G_\qat}% & \Sigma(\pat_i) \cap \Sigma(\pat_j) = \emptyset
}\]
The meaning of a variable pattern $\freeident$ depends on whether it is bound in $\G$.
If $\ident{}$ is free, the pattern matches anything and binds it to $\ident{}$.
If $\ident{}$ is bound, the pattern matches only the value of $\ident{}$.
Wildcard patterns \verb|_| match anything and do not introduce a binding.
\begin{center}
  \begin{tabular}{c@{\hskip 0.1in}c@{\hskip 0.1in}c}
  \infer{\mmatch{\G}{\freeident}{e}{\G}}
        {\Gamma(\ident) = e} & 
  \infer{\mmatch{\G}{\freeident{}}{e}{\rext{\G}{\freeident{}}{e}}}{{\ident} \notin \G} &
  \infer{\mmatch{\G}{\wildpat}{e}{\G}}{}
  \end{tabular}
\end{center}
The above patterns often suffice for selecting individual facts, as done in forward proofs.
However, when referencing a large number of facts (e.g. in \kwshow{}), it helps to select facts in bulk, for which the following patterns are also useful.
Variable occurrence patterns $p(vars)$ and $p(\neg vars)$ select all formulas $\phi$ where the given variables do or do not occur in its free variables $\freevars{\phi}$, respectively:
\begin{center}
  \begin{tabular}{c@{\hskip 0.1in}c}
\infer{\mmatch{\G}{p(vars)}{e}{\G}}
        {vars\subseteq\freevars{e}} &     
\infer{\mmatch{\G}{p(\neg vars)}{e}{\G}}
      {vars\cap\freevars{e}=\emptyset}
  \end{tabular}
\end{center}
Patterns support set operations.
Set patterns proceed left-to-right. Union short-circuits on success and intersection short-circuits on failure.
Negation patterns $\neg\pat$ match only when $\pat$ fails to match, so we require that $\pat$ binds no variables in this case for the sake of clarity.
This is not a restriction because any free variable patterns in $\pat$ can be replaced with wildcards, which never bind.
\begin{center}{\small
  \begin{tabular}{c@{\hskip 0.2in}c}
  \infer{\match{\pat\cup\qat}{e}{\G_\pat}} 
        {\match{\pat}{e}{\G_\pat}
        %&\match{\pat_2}{e}{\Delta_2} 
        %&\Sigma(\pat_1) = \Sigma(\pat_2)
} & 
\infer
  {\mmatch{\G}{\pat\cap\qat}{e}{\G_\qat}}
    {\mmatch{\G}{\pat}{e}{\G_\pat} \hskip 0.1in \mmatch{\G_\pat}{\qat}{e}{\G_\qat}}
%\G_1(x) = \G_2(x)~(\text{for} \ensuremath{x\in\Sigma(\G_1)\cap\Sigma(\G_2))}}
\\[0.1in]
\infer{\match{\pat\cup\qat}{e}{\G_\qat}}
   {\nmatch{\pat}{e} & 
    \match{\qat}{e}{\G_\qat}} & 
    \infer{\mmatch{\G}{\neg\pat}{e}{\G}}
       {\nmatch{\pat}{e} & \boundvars{\pat}=\emptyset}
  \end{tabular}}
\end{center}

\paragraph{Extended Expression Evaluation}
%TODO: What does proper even mean anymore
In order to keep abbreviations {\tt ident\_} outside the core language of \dL, we automatically expand \emph{extended terms} featuring abbreviations {\tt ident\_} to proper terms
with the \emph{term expansion} function $\eeag{\G}{e} = f$ (we omit $\G$ when clear).
%\[e ::= ?ident~|~\cdots\]
As in pattern-matching, expression constructors map through homomorphically and identifiers are substituted with their values:
\begin{center}
  {\small
\begin{tabular}{c@{\hskip 0.5in}c}
$\longeag{\otimes(e_1,e_2)}=\otimes(\longeag{e_1},\longeag{e_2})$ &
$\longeeag{\G}{\ident}=e$\ \ \ when\ ${\Gamma(\ident) = e}$
\end{tabular}
}
\end{center}

\paragraph{Forward-Chaining Proof Terms}
A comprehensive structured language should provide both backward and forward-chaining proof.
Wenzel~\cite{DBLP:conf/tphol/Wenzel99} observes that backward chaining is often most natural for major steps and forward chaining more natural for minor intermediate steps.
Backward chaining works well when the proof can be guided either by the structure of a formula (e.g. during symbolic execution) or by human intuition (e.g. when choosing invariants).
The addition of forward chaining becomes desireable when we wish to experiment with free-form compositions of known facts, e.g. when trying to assist an arithmetic solver with manual simplifications.
%We agree with an observation of Wenzel it is often most natural to use backward-chaining rules for the major proof steps (in Kaisar, program decomposition) and forward-chaining rules for minor intermediate steps (in Kaisar, arithmetic and propositional simplifications).
In Kaisar, forward-chaining proofs are built from atomic facts in $\G$ and a standard library of first-order logic rules, which are composed with application $(\fproof~\fproof)$ and instantiation $(\fproof~\theta)$.
%Rules pertaining to programs are purposely omitted, so that only the backward-chaining fragment of the language need concern itself with state change:
\begin{align*} 
\fproof ::=& \pat\ |\ (\fproof~\fproof)\ |\ (\fproof~\theta)
\end{align*} 
The judgement $\GE_\Sigma e : \phi$ says $e$ is a proof of $\phi$ using assumptions $\Gamma$, where $\Sigma$ is a library of builtin propositional rules.
%using rules drawn both from $\Gamma$ (user-provided rules)
%TODO: Fix conflict with other notation
Facts and rules can be selected from $\G$ and $\Sigma$ with patterns:
\begin{center}{\small
  \begin{tabular}{c@{\hskip 0.2in}c@{\hskip 0.2in}c}
% \infer{\GHE_\Sigma pat:\phi}{\phi \in \Gamma & \Gamma;H\vdash({\phi} \match e)=\cdot}    
\infer{\fpchks{\G}{\pat}{\phi}{\Sigma}}
      {\phi \in (\Sigma~\cup~\G) & \ematch{\phi}{\pat}} &
\infer{\fpchks{\G}{\fproof_1~\fproof_2}{\psi}{\Sigma}}
      {\fpchks{\G}{\fproof_2}{\phi}{\Sigma}&\fpchks{\G}{\fproof_1}{(\phi\limply\psi)}{\Sigma}} &
\infer{\fpchks{\G}{\fproof~\theta}{\subst[\phi]{x}{\eag{\theta}}}{\Sigma}}
      {\fpchks{\G}{\fproof}{\forall x~\phi}{\Sigma}}
  \end{tabular}}
\end{center}
\section{Static and Dynamic Execution Traces}
\label{sec:trace}
We describe	 the static execution trace mechanism used to implement nominals and automate historical reference.
We present the challenges of historical reference by an easy example with sequent calculus proofs for assignment.
Consider the sequent calculus assignment rules:
%We elaborate on the use of static execution traces to implement nominals, and the challenges solved.
%Before presenting the dynamic logic features of Kaisar in full detail, we elaborate on the static execution trace mechanism that supports nominals.
%Recall the motivation behind nominal terms: hybrid systems are inherently stateful and their proofs must frequently relate current state to past state, therefore a proof language should make simultaneous reference to current and past state easy.
%TODO: Mention sequent-level state early.
%expand in greater detail on the design of nominals and static traces, and the problems they solve.
%Let us focus briefly on reasoning about discrete assignments, which are the quintessential state-modifying operation.
%Before extending Kaisar to dynamic logic, we first demonstrate the traceability problem for dynamic logic in greater detail and show how traces provide its solution.
%There are several possible rules for eliminating assignments in a dynamic logic sequent calculus:
\begin{center}
\begin{tabular}{c@{\hskip 0.2in}c}
\infer[{[:=]sub}]{\GE[x:=\theta]\phi,\Delta}{\Gamma\vdash\subst[\phi]{x}{\theta},\Delta}
&
\infer[{[:=]eq}]{\GE[x:=\theta]\phi,\Delta}{\subst[\Gamma]{x}{x_i},x=\theta\vdash \phi,\subst[\Delta]{x}{x_i}}
\end{tabular}
\end{center}
In $[:=]sub$, reference to the initial state stays simple, while the final value does not.
In the resulting goal $\Gamma\vdash\subst[\phi]{x}{\theta},\Delta$, the variable $x$ refers to the initial value, but we must write $\theta$ (which may be a large term) to refer to the final value.
In $[:=]eq$, the opposite is true: $x$ now refers to the final value, but the initial value of $x$ is stored in a fresh ghost variable $x_i$, which we must remember.

Complicating matters further, in practice we wish to use a combination of $[:=]sub$ and $[:=]eq$.
%TODO: Cite also very original CAD paper
The $[:=]sub$ rule only applies when the substitution $\subst[\phi]{x}{\theta}$ is admissible, e.g. when $\freevars{\theta}\cap\boundvars{\phi}=\emptyset$.
However, we wish to use it whenever it applies to reduce the total number of variables and formulas in the context, which are essential to the performance of real-arithmetic decision procedures~\cite{PCAD}.
Therefore the natural approach is to use $[:=]sub$ when it applies and $[:=]eq$ otherwise.

After a number of such reasoning steps, the meaning of a variable $x$ in a sequent (we call this the \emph{sequent-level meaning}) may disagree with the value of both the initial and final values of the program variable $x$ (we call these the \emph{program-level meaning} in the initial and final states).
To observe this issue in action, consider the following (trivial) sequent proof:
\irlabel{eq|$[:=]eq$}
\irlabel{sub|$[:=]sub$}
\irlabel{closeId|id}
\begin{sequentdeduction}[array]
\linfer[closeId]
{\lclose}
{\linfer[sub]
  {\lsequent{x_0=2,x=1}{x+5=3}}
  {\linfer[eq]{\lsequent{x_0=2,x=1}{[x:=x+5]x=3}}
          {\linfer[eq]{\lsequent{x=2}{[x:=1][x:=x+5]x=3}}
            {\lsequent{ }{[x:=2][x:=1][x:=x+5]x=3}}}}
}
\end{sequentdeduction}

There are four program states in the above proof: one before each assignment and one at the end.
Throughout the first two steps, the sequent-level meaning of $x$ corresponds exactly with its program-level meaning in the current state.
%The first two steps $[:=]eq$ maintain the correspondence between the sequent and program-level meanings of $x$.
%However, as soon as we apply the $[:=]sub$ step, this correspondence is lost.
At the final state of the program, the value of $x$ in the program corresponds to $x + 5$ in the sequent, whereas $x$ in the sequent refers to the value of $x$ from its \emph{previous} state.
%TODO: Talk about program-level reasoning more
This is a problem: It is non-trivial to reference initial and final, let alone intermediate state in a proof regime that mixes $[:=]eq$ and $[:=]sub$, yet we want them all.
%We say the proof is \emph{untraceable} because the variable $x$ in each sequent (we call this the \emph{sequent-level meaning} of $x$) is not necessarily equal to the value of $x$ in the corresponding program state (we call this the \emph{program-level meaning}).
%In particular, the first two $[:=]eq$ steps maintain traceability, but this correspondence is lost as soon as we perform a single $[:=]sub$ step (which we wish to do as often as possible).
%This discrepancy is the core of the tracability problem.
\newcommand{\ddia}[2]{\langle{#1}\rangle{#2}}
\newcommand{\interp}[1]{\ivaluation{}{#1}{}}
\newcommand{\tinterps}[2]{\ivaluation{}{#1}{#2}}
\newcommand{\mkint}[1]{\mathit{interp}(#1)}
\newcommand{\dLN}{\dLh}
\newcommand{\trfst}[1]{{\tt fst}(#1)}
\newcommand{\trlst}[1]{{\tt last}(#1)}
\newcommand{\dom}[1]{\ensuremath{\mathrm{Dom}(#1)}}
\newcommand{\corr}[2]{#1\sim{#2}}
\newcommand{\hemp}{\epsilon}
\newcommand{\tsing}[1]{(#1)}
\newcommand{\ssubst}[3]{#1_{#2}^{#3}}
\newcommand{\now}[1]{\ensuremath{\mathrm{now}(#1)}}
\newcommand{\hnow}[2]{\ensuremath{\mathrm{now}_{#1}(#2)}}
\newcommand{\hnom}[3]{\ensuremath{\mathrm{#1}_{#2}(#3)}}
\newcommand{\seqstate}[2]{\ensuremath{\mathfrak{S}(#1;~#2)}}
\newcommand{\sseqstate}[2]{\ensuremath{\mathfrak{S}(#1;#2)}}
\newcommand{\sval}[3]{\seqstate{#1}{#2}(#3)}
\newcommand{\ssval}[3]{\sseqstate{#1}{#2}(#3)}
%{\ensuremath{~\vdash\hspace{-0.05in}(#1;#2;#3)}}
\newcommand{\mknom}[3]{nom({#1},{#2})={#3}}
\newcommand{\ttrunc}[2]{#1(#2)}
%\paragraph{Static Execution Traces Automate State Change Navigation}

We address this problem by automating state-change bookkeeping in a \emph{static execution trace} data structure.
Static execution traces automate state navigation by providing a static, finitary abstraction of a dynamic program execution trace.
\begin{definition}[Static Traces]
\label{def:stat-tr}
A static trace is an ordered list of four kinds of \emph{trace records} (tr):
\[tr ::= \hrsub{x}{\theta}\ |\ \hreq{x}{x_i}{\theta}\ |\ \hrany{x}{x_i}\ |\ t\]
We denote the empty trace by $\hemp{}$.
For any state name $t$ appearing in $H$ (i.e. $t \in \dom{H}$) we denote by $\ttrunc{t}{H}$ the unique prefix of $\eta$ ending at state $t$.
\end{definition}
By maintaining a substitution record $\hrsub{x}{\theta}$ for each substitution, we can automatically translate between the sequent-level and (current) program-level meaning of an expression.
For example, if you wish to know the sequent-level value of the program-level term $x^2$ in the final state, after $\hrsub{x}{x+5}$, it suffices to compute $\subst[(x^2)]{x}{x+5} = (x + 5)^2$.
%TODO: equational -> assumption
We enable nominal references to past states by adding a $t$ record at each named state $t$ and an $\hreq{x}{x_i}{\theta}$ any time the $[:=]eq$ rule is used to rename $x_i=x$ and introduce an assumption $x=\theta$.
This allows us to determine, e.g. that the second value of $x$ was ultimately renamed to $x_0$.
The case for $\prandom{x}$ is marked with $\hrany{x}{x_i}$, which is analogous to $\hreq{x}{x_i}{\theta}$, without any assumption on the $x$ value.

%TODO: Relate to following sections
Given a trace, we can reconstruct the value at any proof state by replaying the composition of all substitutions since the renaming of interest.
We begin by defining the pseudo-nominal $\hnow{H}{x}$ which computes the \emph{current} sequent-level equivalent for a program-level variable $x$ at the end of trace $H$.
All expressions are by default assumed to occur at state {\tt now}.
Because expressions depend only on the values of variables, it suffices to define the variable case:
  \begin{center}
\begin{minipage}{0.4\linewidth}
    \begin{align*}
      &\hnow{\hemp{}}{x} &&= x                &\hnow{\hhsub{H}{x}{\theta}}{x} &&= \theta\\
      &\hnow{\hhany{H}{x}{x_i}}{x} &&= x      & \hnow{\hheq{H}{x}{x_i}{\theta}}{x} &&= x \\
      &\hnow{H,{\tt hr}}{x} &&= \hnow{H}{x} & \text{(for all other {\it{hr}})}~\ldots&&
      % &\hnow{H,t}{x} &&= \hnow{H}{x} &\hnow{\hhsub{H}{y}{\theta}}{x}
      % &&= \hnow{H}{x} \hnow{\hheq{H}{y}{y_i}{\theta}}{x} &&=
      % \hnow{H}{x}\\
    \end{align*}
\end{minipage}
  \end{center}
We recurse until we find a record for the $x$ of interest.
\newcommand{\kwsub}{{\bf{sub}}}
\newcommand{\kweq}{{\bf{eq}}}
If it is a \kwsub{} record, we stop immediately and return $\theta$: even if the trace contains multiple \kwsub{}s for the same $x$, they are cumulative (the last record contains the composition of all \kwsub{}s).
If it is a \kweq{} record, then we use the current value of $x$.

To compute a nominal $\hnom{t}{H}{x}$, we determine the name $x$ has at state $t$ in the history $H$, which is either $x$ or some ghost $x_i$ (if $x$ has been ghosted since state $t$).
%TODO: There's notation for this isn't there?
We then compute $\hnow{H',t}{x}$ or $\hnow{H',t}{x_i}$ accordingly where $H'$ is the prefix of $H$ preceding state $t$.
As in the program variable case, $\hnow{H}{x_i}$ can either be exactly $x_i$ or the result of a substitution.
%This occurs when  and thus $t(x)$ depends on a 
%In addition to program variables $x$, the argument of \emph{now} can be a ghost variable $x_i$,
  \begin{center}
\begin{minipage}{0.7\linewidth}
\begin{align*}
  &\hnom{t}{H,t}{x}               &&= \hnow{H}{x}       &\hnom{t}{\hheq{H}{x}{x_i}{\theta}}{x}  &= \hnom{t}{H}{x_i}\\
  &\hnom{t}{\hhany{H}{x}{x_i}}{x} &&= \hnom{t}{H}{x_i}  &\hnom{t}{H,{\tt{hr}}}{x}               &= \hnom{t}{H}{x}~\ \ \text{(for all other {\it{hr}})}~\ldots
\end{align*}
\end{minipage}
\end{center}
It is now natural to ask the question: If $\hnow{H}{x}$ converts between the program-level state and sequent-level state, can we give a precise meaning to the notion of sequent-level state?
The answer is yes, but in general the sequent-level state will not be identical to any specific state the program passed through, but rather each variable might take its meaning from different past states.
To this end, we define a notion of \emph{dynamic trace} encapsulating all past program states.

\begin{definition}[Dynamic Traces] 
\label{def:dyn-tr}
A \emph{dynamic trace} $\eta$ is a non-empty list of program states $\omega$, possibly interleaved with state names $t$ (but always containing at least one state).
The first state is denoted $\trfst{\eta}$, the last state $\trlst{\eta}$.
As with static traces, for any state name $t$ appering in the trace (i.e. $t\in\dom{\eta}$), we denote by $\ttrunc{t}{\eta}$ the longest prefix of $\eta$ preceding state name $t$.
We denote singleton traces $\tsing{\omega}$.
\end{definition}
\begin{definition}[Sequent-Level State]
\label{def:seq-state}
We define \emph{sequent-level state} $\seqstate{\eta}{H}$ for dynamic and static traces $\eta$ and $H$.
Recall that after $\hreq{x}{x_i}{\theta}$ the variable $x$ represents the end state of the assignment, while after $\hrsub{x}{\theta}$ it represents the start state.
We take each program variable $x$ from its most recent $\hreq{x}{x_i}{\theta}$ or $\hrany{x}{x_i}$ state, or the initial state if none exists.
Each ghost is assigned at most once and takes its value from the state in which it was assigned.
We give an inductive definition:
%that which maps each $x$ to its sequent-level value $\sval{\eta}{H}{x}$ 
%defined  inductively by:
\begin{center}
\begin{minipage}{0.70\linewidth}
\begin{align*}
%\subst{x}{r}
  \seqstate{\tsing{\omega}}{\hemp{}}                &=\omega & \\
  \seqstate{\eta,\omega}{\hhany{H}{x}{x_i}}         &=\subst[\seqstate{\eta}{H}]{x_i}{\sseqstate{\eta}{H}(x)}~_{x}^{\omega(x)} & \text{($x_i$ fresh)}\\
  \seqstate{\eta,\omega}{\hheq{H}{x}{x_i}{\theta}}  &=\subst[\seqstate{\eta}{H}]{x_i}{\sseqstate{\eta}{H}(x)}~_{x}^{\omega(x)} & \text{($x_i$ fresh)}\\
  \seqstate{\eta,\omega}{\hhsub{H}{x}{\theta}}      &=\subst[\seqstate{\eta}{H}]{x}{\omega(x)}&\\
  \seqstate{\eta,t}{H,t}                            &=\seqstate{\eta}{H}&
%  \sval{\tsing{\omega}}{\hemp{}}{x} &= \omega(x)\\
  % \sval{\eta,\omega}{\hrany{H}{x}{x_i}}{x} &= \omega(x)\\
  % \sval{\eta,\omega}{\hreq{H}{x}{x_i}{\theta}}{x} &= \omega(x)\\
  % \sval{\eta,\omega}{\hrsub{H}{x}{\theta}}{x} &= \sval{\eta}{H}{x}\\
  % \sval{\eta,\omega}{\hrany{H}{y}{y_i}}{x} &= \sval{\eta}{H}{x}\\
  % \sval{\eta,\omega}{\hreq{H}{y}{y_i}{\theta}}{x} &= \sval{\eta}{H}{x}\\
  % \sval{\eta,\omega}{\hrsub{H}{y}{\theta}}{x} &= \sval{\eta}{H}{x}\\
  % \sval{\eta,t}{H,t}{x} &= \sval{\eta}{H}{x}\\
\end{align*}
\end{minipage}
\end{center}
\end{definition}
\section{Discrete Dynamic Kaisar}
\label{sec:ddk}
We now extend Kaisar with its core feature: nominal terms.
We add a  construct {\kwstate{} $t$} which gives a name $t$ to the current abstract proof state, after which we can write nominal terms $t(\theta)$ to reference the value of term $\theta$ at state $t$ from future states. 
Nominal terms are supported by {\it structured symbolic execution} rules for each program construct, which automatically maintain the corresponding static execution trace.
As before we proceed from examples to syntax and proof-checking rules.

%We expand further on the proof of Model~\ref{model:dive} with loop invariants and discrete program reasoning, but we leave ODE reasoning for Section~\ref{sec:dlk}.
%We wish to prove that under the appropriate preconditions, we have not exceeded the safe speed $m$ when we land ($x=0$).
%In this section we prove invariance of the domain constraint ({\tt dc}) and constant propositions ({\tt const}) since those do not require continuous reasoning.
\paragraph{Examples} 
We continue the proof of Model~\ref{model:dive}, augmenting it with loop invariants and other discrete program reasoning, but we leave differential equation reasoning for Section~\ref{sec:dlk}.
Recall the program and statement of Theorem~\ref{thm:safe} (Skydiver Safety for Model~\ref{model:dive}):
\begin{align*}
  \mathit{Pre}   & \equiv \left(\mathit{dc} \land \mathit{const}\right) \land \mathit{dyn} & \mathit{plant} &\equiv \humod{t}{0};~ \{x'=v,~ v'=\rvar\cdot v^2 - g ~\&~ x \ge 0 \land v<0 \land t \le \Tvar\} \\
  \mathit{dc}    &\equiv x \ge 0 \land v < 0 & \mathit{const} & \equiv g > 0 \land 0 < \ravar  < \rpvar \land \Tvar \ge 0 \\
%\end{align*}
%\begin{align*}
  \mathit{dyn}   &\equiv  \abs{v} < \sqrt{\frac{g}{\rpvar}} < m & \mathit{ctrl}  &\equiv~ ?\left(\rvar=\ravar \land v - g\cdot \Tvar > -\sqrt{\frac{g}{\rpvar}}\right) \cup \humod{\rvar}{\rpvar}%o\\
%& &   
\end{align*}
\begin{proposition}\label{prop:safe}$\rvar = \ravar\land{dc}\land{const}\land{dyn}\limply[\{\mathit{ctrl};\mathit{plant}\}^*](x=0\limply v \leq m)$ is valid. \end{proposition}
%To ease the proof we divide the preconditions ({\tt Pre}) into those which remain true by the domain constraint of our ODE ({\tt dc}), those which depend only on constants ({\tt const}) and those which depend on the dynamics expressed in the ODE ({\tt dyn}):

Examples 2(a-c) are proofs of this Proposition~\ref{thm:safe} (with differential equation reasoning postponed until Section~\ref{sec:dlk}).
As before, we proceed from basic to advanced proof techniques.
%, in different styles.
Examples 2(a,b) both use a single loop invariant, and begin by splitting the conclusion into $[\{ctrl;plant\}](dc \land const)$ and $[\{ctrl;plant\}]dyn$ to separate the discrete and continuous reasoning.
Example 2a splits eagerly on the control decision, which is straightforward but often requires duplication of proofs about the plant.
Example 2b uses Hoare-style composition reasoning instead, which reduces duplication.
In general, Hoare-style composition adds the cost of the user supplying a composition formula $I$, but $I$ is trivial in this case. 
% then splits eagerly on the control decision, where 2b uses Hoare-style sequential composition reasoning (\kwmid{}) to reduce complexity.
Example 2c proves the loop invariants $\mathit{dc}, \mathit{const}, \mathit{dyn}$ separately, which is useful in interactive proofs when some invariants ($\mathit{dc}, \mathit{const}$) prove trivially while others ($\mathit{dyn}$) are complex.
%Example 2c proves the same theorem, but proves each loop invariant separatly, emphasizing the simplicity of proving {\tt dc} and {\tt const}.
Example 2d proves a slightly different theorem which establishes a bound on the position $x$, showcasing discrete nominals ($init(v), init(x)$).
More advanced uses of nominals are in the full proof in Section~\ref{sec:dlk}.
%\[x=2\cdot y\limply \left[\left\{\{?(x<-2);x:=x^2\}\cup y:=\frac{1}{3}\cdot y\right\}x:=x\cdot 2\right]x > 2\cdot y\]
\begin{figure}
  \centering
  \begin{tabular}{cc}
    \begin{minipage}{0.5\textwidth}
\verb|# Example 2a|\\
\kwassume\verb| ra:|        $r=a$\\
\kwassume\verb| domInit:|   $dom\_$\\
\kwassume\verb| constInit:| $const\_$\\
\kwassume\verb| dynInit:|   $dyn\_$\\
\kwinv\verb| J:| $(dom\_ \land const\_) \land dyn\_~\{$ \\
\kwpre$~\Rightarrow~\mathbb{R}$\\
\kwind$~\Rightarrow~\{$\\
\verb| |\kwcase\ $(dom\_ \land const\_) \Rightarrow~\{$\\
\verb|  |\kwcase\ $?(\_)~\Rightarrow$ \\
\verb|   |\kwassume\verb| slowEnough:|\\
\verb|     |$(vEps\_ > vBound\_ \land r=a)$\\
\verb|   |\kwshow\ $\_$ \kwby{} $\mathbb{R}$\\
\verb+  |+$\humod{\rvar}{\rpvar}~\Rightarrow$\\
\verb|   |\kwassume$~\humod{x}{\rpvar}$\\
\verb|   |\kwshow\ $\_$ \kwby{} $\mathbb{R}$\\
$\}$\\
%\verb| }|\\
\verb| |\kwcase\ $dyn\_ \Rightarrow~\{$\\
\verb|  |\kwcase\ $?(\_) \Rightarrow$\\
\verb|   |\kwassume\verb| slowEnough:|\\
\verb|     |$(vEps\_ > vBound\_ \land r=a)$\\
%\verb|    (vEps_ > vBound_ & r=a)|\\
\verb|      |$\ldots$\\
%\verb|    ...|\\
\verb+  |+$\humod{\rvar}{\rpvar}~\Rightarrow~\ldots$\\
$\}\}\}$\\
%\verb|}}}|\\
\kwshow$ (x=0 \limply |v| \leq m)$\\
\kwusing\ {\tt{J}} \kwby{} $auto$\\
\verb||\\
\verb||\\
\verb|#Example 2c|\\
\kwassume\verb| ra:|  $r=a$\\
\kwassume\verb| domInit:|   $dom\_$\\
\kwassume\verb| constInit:| $const\_$\\
\kwassume\verb| dynInit:|   $dyn\_$\\
\kwinv\verb| DOM:|   $dom\_$   \\
\kwinv\verb| CONST:| $const\_$\\
\kwinv\verb| DYN:|   $dyn\_~\{$\\
\kwpre\ $\Rightarrow~\mathbb{R}$\\
\kwind\ $\Rightarrow~\{$\\
\verb| |\kwcase\\
\verb|  |$?(vEps\_ > vBound\_ \land r=a) \Rightarrow~\ldots$\\
\verb+  |+$\humod{\rvar}{\rpvar}~\Rightarrow~\ldots$\\
\verb|}}}|\\
\kwfinally\ \kwshow\ $(x=0 \limply |v| \leq m)$\\
\kwusing\ {\tt{DYN}} \kwby\ $auto$\\
% \begin{verbatim}
% # Example 2a
% assume xy:(x=2*y)
% case (?(x<-2); _)
%   assume x:(x<-2)
%   assign x:=x^2
%   show(x > 2*y)
%     by auto
% case (y := _)
%   assign (y:=(1/3)*y)
%   assign x:=x^2
%   show (x > 2*y)
%     by auto
% \end{verbatim}
\end{minipage}
			&
%TODO: Can I work in nominals somewhere reasonably.
\begin{minipage}{0.5\textwidth}

\verb|# Example 2b|\\
\kwassume\verb| ra:|        $r=a$\\
\kwassume\verb| domInit:|   $dom\_$\\
\kwassume\verb| constInit:| $const\_$\\
\kwassume\verb| dynInit:|   $dyn\_$\\
\kwinv\verb| J:| $(dom\_ \land const\_) \land dyn\_ \{$\\
\kwpre$~\Rightarrow~ \mathbb{R}$\\
\kwind$~\Rightarrow~\{$\\
\verb| |\kwcase$~(dom\_ \land const\_)~\Rightarrow~$\\
\verb|# Forget everything about ctrl|\\
\verb|  |\kwmid\ $\{$\verb| |\kwshow\ $\_$\ \kwby\ $\mathbb{R}~\}$\\
\verb|   |\kwfirst\ $I$:$~true$\\
\verb|   |\kwthen\ $\{$\verb| |\kwshow\verb| |$\_$ \kwby{} $auto~\}$\\
\verb| |\kwcase\ $~dyn\_~\Rightarrow~\{$\\
\verb|  |\kwcase\\
\verb+   |+$?(\_)~\Rightarrow~$\\
\verb|    |\kwassume\verb| slowEnough:|\\
\verb|       |$(vEps\_ > vBound\_ \land r=a)$\\
\verb+   |+$\humod{\rvar}{\rpvar}~\Rightarrow~\ldots$\\
$\}\}\}$\\
\kwshow\verb| |$(x=0 \limply v \leq |m|)$\\
\kwusing\verb| |{\tt{J}}\verb| |\kwby\verb| |$auto$\\
\verb||\\
%\verb||\\
\verb|#Example 2d|\\
\kwassume\verb| ra:|        $r=a$\\
\kwassume\verb| domInit:|   $dom\_$\\
\kwassume\verb| constInit:| $const\_$\\
\kwassume\verb| dynInit:|   $dyn\_$\\
\kwassume\verb| top:|       $(x \leq xmax)$\\
\kwstate\verb| init|\\
\kwinv\verb| DOM:|   $dom\_$\\
\kwinv\verb| CONST:| $const\_$\\
\kwinv\verb| DYN:|   $dyn\_~\{$\\
\kwpre$~\Rightarrow~\mathbb{R}$\\
\kwind$~\Rightarrow~\{$\\
\verb| |\kwcase\\
\verb+   +$~?(vEps\_ > vBound\_ \land r=a)~\Rightarrow~\ldots$\\
\verb+  |+$~\humod{\rvar}{\rpvar}~\Rightarrow~\ldots$\\
$\}\}\}$\\
\kwinv\verb| VDECR:|$(v \leq {\it{init}}(v))~\{~\ldots~\}$\\
\kwinv\verb| X:|$(x \leq {\it{init}}(x))~\{~\ldots~\}$\\
\kwfinally\ \kwshow\ $(x \leq xmax)$\\
\kwusing\ {\tt{VDECR}}\ {\tt{X}}\ \kwby{} $auto$\\

% \begin{verbatim}
% # Example 2b
% assume xy:(x=2*y)
% mid J:(x>y)
%   show ([_]x>y) by auto
% assign x:=x*2
% show (x > 2y) using J by auto
% \end{verbatim}
\end{minipage}
\end{tabular}
% \begin{verbatim}
% # Example 2c
% assume xy:(x=2*y)
% mid J:(x>y)
%   show ([_]x>y) by auto
% state pre-assign
% assign x:=x*2
% have xs:(x >= 2*pre-assign(x) & pre-assign(x) > y)
%   by auto
% show _ using J by auto
% \end{verbatim}
% Kaisar Proofs of First-Order Example
\caption{Kaisar Proofs of Skydiver Discrete Fragment}
\label{fig:ddk}
\end{figure}
\paragraph{Structured Symbolic Execution}
Symbolic execution is implemented by adding to the class of structured proofs (\sproof) a set of box rules (\brule{}s) for proving formulas of the form $\dbox{\alpha}{\phi}$ and a set of diamond rules (\drule{}) for proving formulas of form $\ddiamond{\alpha}{\phi}$.
The diamond rules are largely symmetric to the box rules, so we only present the box rules here and give the full list of rules in Appendix~\ref{app:diamond}:
% \bsolve{dc}{\tilde{\phi}}{\sproof}\ |\    \ |\ \iproof
\begin{align*} 
\brule ::=~&\bassert{x}{\pat_\phi}{\sproof}\ |\ \bassign{x}{\pat_\theta}{\sproof}\ |\ \bassignany{x}{\sproof}\\
        |\ &\bcase{\pat_\alpha}{\sproof}{\pat_\beta}{\sproof}\ |\ \bcon{\phi}{\sproof}{\sproof}
\end{align*}
%TODO: Add implicit rule for currying
To improve concision, many proof languages automate steps deemed obvious~\cite{DBLP:journals/jar/Rudnicki87,DBLP:conf/ijcai/Davis81}.
For us, these include the rules for the $\alpha;\beta$ and $?(\phi)$ connectives, i.e. reducing sequential compositions $[\alpha;\beta]\phi$ to nested modalities $[\alpha][\beta]\phi$ and assertions $[?P]Q$ to implications $P\limply{Q}$. 
Negations are implicitly pushed inside other connectives, e.g. $\neg(P\land{Q})\bimply(\neg{P})\lor(\neg{Q})$ and $\neg\dbox{\alpha}{\phi}\bimply\ddiamond{\alpha}{\neg\phi}$.
%TODO mention soundness, maybe make a table
These implicit rules reduce verbosity by automating obvious steps.
This also enables us, for example, to reuse the $(\bassert{x}{\pat}{\sproof})$ rule for implication as if it applied to tests as well, as in Example 2a.

Because structured symbolic execution rules affect the trace, we now \emph{extend the $\sproof$ checking judgment} to $\tpchk{\G}{H_1}{\sproof}{\Delta}{H_2}$ where $H_1$ is the initial trace and $H_2$ is the final trace.
The final trace helps reference the internal states of one subproof within another: see the \kwmid{} rule.
%, e.g. in the rule \kwmid{} for $[\alpha;\beta]\phi$ to use the internal states of $\alpha$ within the proof for $\beta$.
\[\infer{\tpchk{\G}{H}{\bassert{x}{\pat}{\sproof}}{[?\psi]\phi,\Delta}{H_\phi}} 
{\mmatch{\G}{\pat}{\psi}{\Gamma_\psi} & \tpchk{\ext{\Gamma_\psi}{x}{\psi}}{H}{\sproof}{\phi,\Delta}{H_\phi}}\]
The assignment rule itself is completely transparent to the user, but its presence as an explicit rule aids readability and supports the implementation of nominals.
As discussed in Section~\ref{sec:trace}, assignments update the trace because they modify the state.
How they update the trace depends on whether we can perform assignment by substitution or whether we must add an equality to $\G$:
\begin{center}
  \begin{tabular}{cc}
\infer[\text{if}~\phi_{x}^{\eeag{\G}{\theta}}~\text{admissible}]{\tpchk{\G}{H}{\bassign{x}{\theta}{\sproof}}{[x:=\eeag{\G}{\theta}]\phi,\Delta}{H_\phi}}
      {\tpchk{\G}{\hhsub{H}{x}{\eeag{\G}{\theta}}}{\sproof}{\subst[\phi]{x}{\eeag{\G}{\theta}},\Delta}{H_\phi}}
\\[0.1in]
\infer[\text{if}~x_i~\text{fresh}]{\tpchk{\G}{H}{\bassign{x}{\theta}{\sproof}}{[x:=\eeag{\G}{\theta}]\phi,\Delta}{H_\phi}}
      {\tpchk{\subst[\G]{x}{x_i},x=\eeag{\G}{\theta}}{\hheq{H}{x}{x_i}{\eeag{\G}{\theta}}}{\sproof}{\phi,\subst[\Delta]{x}{x_i}}{H_\phi}}
  \end{tabular}
\end{center}
%In both cases, we add trace records {\tt sub} and {\tt eq} which will later be used to resolve nominal terms.
%In the equational case we rename the previous variable $x$ to some fresh $x_i$.
%Between these two rules, we maintain the invariant that the unardorned variable $x$ always refers to the value of $x$ as of the most recent equational step, while older values are stored in $x_j$ for some $j$.
Nondeterministic assignment is analogous to the equality case of assignment:
\[\infer[\text{if}~x_i~\text{fresh}]{\tpchk{\G}{H}{\bassignany{x}{\sproof}}{[\prandom{x}]\phi,\Delta}{H_1}}
      {\tpchk{\subst[\G]{x}{x_i}}{\hhany{H}{x}{x_i}}{\sproof}{\phi,\subst[\Delta]{x}{x_i}}{H_1}}\]
%TODO: Say more?
Nondeterministic choices are proven by proving both branches, matched by patterns $p$ and $q$:
%Nondeterministic choices are proven by proving both branches, which are matched by patterns $p$ and $q$ respectively:
\[\infer{\tpchk{\G}{H}{\bcase{p}{\sproof_\alpha}{q}{\sproof_\beta}}{[\alpha\cup\beta]\phi,\Delta}{H}}
{
\deduce{\tpchk{\G_\alpha}{H}{\sproof_\alpha}{[\alpha]\phi,\Delta}{H_\alpha} \hskip 0.1in \tpchk{\G_\beta}{H}{\sproof_\beta}{[\beta]\phi,\Delta}{H_\beta}}
{\mmatch{\G}{p}{\alpha}{\G_\alpha} &
 \mmatch{\G}{q} {\beta}{\G_\beta}}}\]
This \kwcase{} rule is notable because it produces \emph{non-exhaustive} final traces.
In general, an execution of $\alpha\cup\beta$ executes $\alpha$ or $\beta$, but not both.
We return the input trace $H$ because the final trace only contains changes which (are syntactically obvious to) occur in every branch.
%The final trace contains only changes that occur in every trace, and thus we return only the input trace $H$.
This means any states introduced in $\sproof_\alpha$ or $\sproof_\beta$ have local scope and cannot be accessed externally.
\begin{definition}[Abstraction]
When executing certain programs $\alpha$, it is not known exactly which variables are bound on a given run of $\alpha$.
In these cases, we can reason by abstraction over all bound variables $\boundvars{\alpha}$: we treat their final values as arbitrary.
Abstraction is denoted with superscripts $\arb{\phi}{\alpha},$ not to be confused with subscripts, which are mnemonic.
Let $\boundvars{\alpha} = x_1,\ldots,x_n$ and $y_1,\ldots,y_n$ fresh ghost variables.
We define $\arb{H}{\alpha}=H,\hrany{x_1}{y_1},\ldots\hrany{x_n}{y_n}$, 
$\arb{\phi}{\alpha}=\subst[\phi]{x_1}{y_1}\cdots{}\subst[{~}]{x_n}{y_n}$, and 
$\arb{\omega}{\alpha}=\eta,\subst[\omega]{x_1}{\omega(y_1)}\cdots{}\subst[~]{y_n}{\omega(x_n))}$.
In Section~\ref{sec:metatheory} we show soundness and nominalization results for abstraction.
\end{definition}
As shown in Example 2a, using \kwcase{} too soon increases the complexity of a proof: in a proof of $\{\alpha\cup\beta\};\gamma,$ the proof of $\gamma$ may be duplicated.
% to reduce the length of a proof using \emph{midcondition} reasoning with a
Example 2b reduces proof size with Hoare-style~\cite{Hoare:1969:ABC:363235.363259} composition by specifying an intermediate condition $\psi$ which holds between $\{\alpha\cup\beta\}$ and $\gamma$.
%TODO: Talk about scope early please
\[\infer{\tpchk{\G}{H}{\bcon{\psi}{\sproof_\psi}{\sproof_\phi}}{[\alpha][\beta]\phi,\Delta}{H_\phi}}
{\tpchk{\G}{H}{\sproof_\psi}{[\alpha]\eag{\psi},\Delta}{H_\psi} \hskip 0.1in \tpchk{\arb{\G}{\alpha},\eag{\psi}}{\arb{H_\psi}{\alpha}}{\sproof_\phi}{[\beta]\phi,\arb{\Delta}{\alpha}}{H_\phi}}
\]
%, thus any states in $H_\psi$ can freely be accessed while verifying $\beta$.
Hoare composition is notable because $H_\psi$ contains only changes that happened with certainty: the bound variables of $\alpha$ may have been modified in ways not reflected by $H_\alpha$.
Thus we treat the values of bound variables after running $\alpha$ as arbitrary, abstracting over them.
%For this purpose it suffices to ghost all bound variables of $\alpha$ in $\G$ and $\Delta$, and to augment $H_\psi$ with an $\hrany{x}{x_i}$ record fo
%r each $x\in\boundvars{\alpha}$ for fresh $x_i$.
%\newcommand{\arb}[2]{{#1}^{#2}}
%We let superscripts $\arb{\G}{\alpha}, \arb{\Delta}{\alpha},$ and $\arb{H}{\alpha}$ denote the resulting contexts and histories after ghosting.
%When reasoning about $[\alpha;\beta]\phi$ by composition, 
%the proof of $[\beta]\phi$ can freely refer to the states of $\alpha$, but the trace of $
%w
%In the style of nondeterministic assignment, Hoare composition adds a ghost variable in the trace for all variables changed during $\alpha$ (which are then recorded in the trace):
%TODO: Shouldn't have nominals in \tilde{\psi}
%{H_\psi^* = H_\psi,\hrany{x}{x_i},\ldots,\hrany{y}{y_j} & \Delta^* = \Delta_x^{xi}\cdots_{y}^{yj} & \G^* = \G_x^{xi}\cdots_{y}^{yj}~(\text{for all vars in} \boundvars{\alpha})}

%\ |\ \sstate{t}{\sproof}\ |\ \ident\ |\ \brule\ |\ \drule
The \kwstate{} construct gives a name to the current program state. 
This has no effect on the proof state, but allows that state to be referenced later on by nominal terms, as shown in Example 2c:
\[\infer{\tpchk{\G}{H} {\sstate{t}{\sproof}}{\phi,\Delta}{H_\phi}}
  {\tpchk{\G}{\hhtime{H}{t}}{\sproof}{\phi,\Delta}{H_\phi}}\]
\paragraph{Invariant Proofs}
We verify  discrete loops via invariants.
%To verify programs containing discrete loops, we state and prove loop invariants.
Consider the proofs in Figure~\ref{fig:ddk}.
In Examples 2a and 2b we prove a single loop invariant, where the base case proves automatically, as is often the case.
In Example 2c we subdivide the proof into several invariants which we prove successively.
These styles of proof are interchangeable, but the latter is convenient during proof development to separate simple cases from difficult cases.
% \[x=0\land{y=1}\limply[\{y:= (1/2)\cdot{y}\}^*;\{x:=x+y;y:=(1/2)\cdot{y}\}^*;\{x:=x+y\}^*]x \geq 0\]
% \begin{figure}[h]
% \centering
% \begin{minipage}{\textwidth}
% \begin{verbatim}
% #Example 3
% assume xy:(x=0&y=1)
% time init
% Inv J1:(y > 0)            { Pre => show _ by R | Inv => show _ by R }
% time t1
% Inv J2:(y>0 & x>=init(x)) { Pre => show _ by R | Inv => show _ by R }
% time t2
% Inv J3:(y>0 & x>=t2(x))   { Pre => show _ by R | Inv => show _ by R }
% show (y >= 0) using J1 J2 J3 by R
% \end{verbatim}
% \end{minipage}
% \caption{Loop Example for Discrete Dynamic Kaisar}
% \end{figure}
%In Example 3, each loop was easily handled with a single invariant.
\newcommand{\invlist}{\ensuremath{J}{\rm{s}}}
If (as in Example 2c) an invariant is provable automatically, we may omit the branches \kwpre{} and \kwind{}.
After  proving invariants, the \kwfinally{} keyword returns us to a standard structured proof with all invariants available as assumptions.
\begin{align*} 
\iproof ::=&  \ \kwinv{}~x:\phi~\{\kwpre{}\Rightarrow\sproof~|~\kwind{}\Rightarrow\sproof\}~{\tt \iproof}\ |\ \finally{\sproof}%\ |\ {\tt inv}~[\phi]~\sproof\\
\end{align*}

While checking invariant proofs, we add a context $\invlist{}$ of all the invariants, which are made available both while proving further invariants and at the end of the invariant chain.
As in Hoare composition, we abstract over the history because the inductive step must work after any number of iterations.
\[\infer{\tpchk{\G,[\alpha^*]\invlist{}}{H}{\finally{\sproof}}{[\alpha^*]\phi,\Delta}{H_\phi}}
 {\tpchk{\arb{\G}{\alpha},\invlist{}}{\arb{H}{\alpha}}{\sproof}{\phi,\arb{\Delta}{\alpha}}{H_\phi} 
%\hskip 0.1in \GHE\tilde{\psi}\eval\psi \hskip 0.1in {H_1 = H,\hrany{x}{x_i},\hrany{y}{y_j},\ldots}
}
\]
%TODO: vert space
\[\infer{\tpchk{\G,[\alpha^*]\invlist{}}{H}{\sinv{x}{\psi}{\sproof_{Pre}}{\sproof_{Inv}}{\iproof}}{[\alpha^*]\phi,\Delta}{H}}
{\deduce     {\tpchk{\Gamma,\invlist{}}{H}{\sproof_{Pre}}{\eag{\psi},\Delta}{H_{Pre}} 
 \hskip 0.1in  \tpchk{\arb{\G}{\alpha},\invlist{},\eag{\psi}}{\arb{H}{\alpha}}{\sproof_{Inv}}{[\alpha]\eag{\psi},\arb{\Delta}{\alpha}}{H_{Inv}}}
{\tpchk{\Gamma,[\alpha^*]\invlist{},x\tinycolon{}[\alpha^*]\eag{\psi}}{H}{\iproof}{[\alpha^*]\phi,\Delta}{H_{\it{Tail}}}}}\]
\paragraph{Focus}
The box rules presented here implicitly operate on the first formula of the succedent.
In the common case of proving a safety theorem $\phi\limply\dbox{\alpha}{\psi}$ where all tests $?(\phi)$ contain only first-order arithmetic, this is enough.
%TODO: Cite logical model-predictive control
Hewever, this does not provide completeness for liveness properties $\ddiamond{\alpha}{\phi}$ which produce multi-formula succedents, or for tests containing modalities.
%However, there are important cases where this does not suffice: when proving liveness properties $\ddiamond{\alpha}{\phi}$ we will have multiple modal formulas in the succedent, and in advanced modeling techniques such as \emph{logical model-predictive control}
%TODO: VERY DOUBLE BLIND VERY IMPORTANT DELETE
%(@Andre: neither of the LICS'12's had this, maybe just not mention lmpc explicitly?)
% we may test a modal formula $?(\dbox{\alpha}{\phi})$, bringing a modality into the antecedent.
We restore completeness for these cases, extending the class \sproof{} with a \emph{focus} construct which brings an arbitrary formula (selected by pattern-matching) to the first succedent position:
\newcommand{\sfocus}[2]{\kwfocus{}~#1~#2}
%\[\sproof ::= \cdots\ |\ \sfocus{\pat}{\sproof}\]
A \kwfocus{} in the antecedent is the inverse of $\neg{R}$; in the succedent it is the exchange rule:
%In the antecedent, \kwfocus{} reduces to the inverse of the $\neg{R}$ sequent calculus rule, and in the succedent it reduces to exchange:
\begin{center}
\begin{tabular}{cc}
  \infer{\tpchk{\G}{H}{\sfocus{\pat}{\sproof}}{\Delta_1,\phi,\Delta_2}{H}}
        {\emmatch{\G}{\pat}{\phi} & \tpchk{\G}{H}{\sproof}{\phi,\Delta_1,\Delta_2}{H_{\phi}}} 
& \infer{\tpchk{\G_1,\phi,\G_2}{H}{\sfocus{\pat}{\sproof}}{\Delta}{H}}
{\emmatch{\G}{\pat}{\phi} & \tpchk{\G_1,\G_2}{H}{\sproof}{\neg{\phi},\Delta}{H_{\neg{\phi}}}}
\end{tabular}
\end{center}
Recall that negations are pushed inward implicitly, so upon \kwfocus{}ing a formula $\dbox{\alpha}{\phi}$ from the antecedent, we will ultimately have $\ddiamond{\alpha}{\neg{\phi}}$ in the succedent, for example.
As with \kwcase{}, the subproof $\sproof$ can access both the initial trace $H$ and any local changes from the proof of $\neg\phi$, but any such changes leave scope here.
Regardless of the origin of $\neg\phi,$ any structured symbolic execution proof can employ state-based reasoning, but as with \kwcase{} it does not follow that those state changes remain meaningful in any broader context.
%Note also that the subproof $\sproof$ has access to the current trace $H$, but that any additional changes $H_1$ go out of scope here.
%One intuition for this behavior is that \kwfocus{} forks off an ``alternate timeline'': within a \kwfocus{} we are free to add and reference new states as we wish, but they are only meaningful within this alternate timeline, which ends with the end of the \kwfocus{}.

In Section~\ref{sec:metatheory} we show that \kwfocus{}, combined with the execution rules for boxes and diamonds in the succedent, provides completeness.
This formulation minimizes the core proof calculus, but \kwfocus{}-based derived rules for antecedent execution may be useful in practice.
%but in practice one may want derived constructs based on \kwfocus{} in order to directly execute formulas in the antecedent.
The completeness proof of Section~\ref{sec:metatheory} provides intuition for how such constructs would be derived.

\paragraph{Extended Expressions and Patterns}
Discrete Dynamic Kaisar adds nominal terms $t(\theta)$ (where $t$ is the name of some named state) to the language of expressions $e$.
%\[e ::= t(\theta)~|~\cdots\]
This change raises a design question: when defining an abbreviation, should program variables refer to their bind-time values, or their expand-time values?
We choose bind-time evaluation as the default, so  $\eag{e}$ evaluates nominals, performing structural recursion and using the rules of Section~\ref{sec:trace} in the variable case. 
When a variable $x$ appears outside a nominal, it is interpreted at the current state, which we denote here using the notation $now(x)$. 
For pattern-matching to work with nominals, matching against (program) variable patterns performs expansion before matching:
%For pattern-matching to make sense in this context, pattern-matching against (program) variable patterns performs expansion before matching:
\newcommand{\labs}[2]{\mathit{mob}(#1,#2)}
\newcommand{\llabs}[3]{\mathit{mob}_{#1}(#2,#3)}
\newcommand{\labseq}[3]{\labs{#1}{#2}=#3}
\newcommand{\llabseq}[4]{\llabs{#1}{#2}{#3}=#4}
\begin{center}
  {\small
%TODO: Improve inference rule sillyness
\begin{tabular}{c@{\hskip 0.1in}c@{\hskip 0.1in}c@{\hskip 0.1in}c}
$t(\otimes(\termone,\termtwo)) = \otimes(t(\termone),t(\termtwo))$
&$t(q) = q~(\text{for}\ q\in\mathbb{Q})$
&$\eag{x} = \longeag{now(x)}$
%TODO: improve notation
&$%\infer{\ematch{x}{e}}
        \ematch{x}{\longeag{now(x)}}{}$
\end{tabular}}
\end{center}
We also want the option to mix bind-time and expand-time reference, for example in Section~\ref{sec:dlk}.
This enables reusable definitions that still refer to fixed past values.
%In abbreviations, we sometimes wish to refer to both the state as of binding-time and as of expansion-time.
%References to bind-time state allow us to remember and refer back to a fixed past value, while references to expand-time state allow the reuse of a common definition in multiple states.
%Reference to expand-time state is useful for any definition which needs to be discussed across multiple states, such as a loop invariant.
%Mixed references are useful when those invariants relate the current state to some fixed state.
%For an example of both, see Section~\ref{sec:dlk}.
We support this with a new functional variant of the \kwlet{} construct, which is parameterized by a state $t$.
Any subterm $\theta$ under the nominal $t(\_)$ uses the expand-time state, while plain subterms use the bind-time state:
\[\infer{\tpchk{\G}{H}{\slet{t({\tt{x\_}})}{e}{\sproof}}{\phi,\Delta}{H_\sproof}}
  {\tpchk{\ext{\Gamma}{t({\tt{x\_}})}{\llabs{\cemp{}}{e}{t_{now}}}}{H,t_{now}}{\sproof}{\phi,\Delta}{H_\sproof}}\]
%\paragraph{$\llabs{H}{\theta}{t}$}
 Functional let uses a \emph{let mobilization} helper judgment, which (a) expands references to current-state variables and (b) wraps all references to arguments in the $now(x)$ pseudo-nominal:
\[\llabseq{H}{x}{t}{\hnow{H}{x}} ~~~\ \ \ \llabseq{H}{t(\theta)}{t}{now(\theta)} ~~~\ \ \  
\llabseq{H}{\otimes(\termone,\termtwo)}{t}{\otimes(\llabs{H}{\termone}{t},\llabs{H}{\termtwo}{t})}\]
%_{now}
%This is useful, for example, in Example 4 of Section~\ref{sec:dlk}.
Note that due to the addition of functional let, the context $\G$ may now contain extended terms.
As before, any unadorned variable $x$ in $\G$ refers to the current sequent-level meaning of $x$.
Elements of $\G$ can contain extended subterms $now(\theta)$ (for proper terms $\theta$), which are adequately resolved by recursively evaluating any extended terms found during expansion:
\[\longeeag{\G}{\ident}=\eeag{\G}{e}\ \ \text{if}\ \ \G(\ident) = e\]

\section{Differential Dynamic Kaisar}
%\label{sec:ddk}
\label{sec:dlk}
We extend Kaisar to support differential equations, the defining feature of differential dynamic logic.
The examples of this section illustrate the necessity of historical reference to both initial \emph{and} intermediate states.
We show that nominals work even when mixing discrete and continuous invariants, continuous ghosts necessary for ODEs and  first-order reasoning necessary for arithmetic.  
\paragraph{Examples}
First we complete the proof of Model~\ref{model:dive}.
Recall the statement of Theorem~\ref{thm:safe}:
%TODO: Use nice var name macros
{\small\begin{align*}
  \mathit{Pre}   &\equiv \left(\mathit{dc} \land \mathit{const}\right) \land \mathit{dyn}                     \hskip 0.5in  \mathit{dc}    \equiv x \ge 0 \land v < 0    &\mathit{const}  &\equiv g > 0 \land 0 < \ravar < \rpvar \land \Tvar \ge 0\\
  \mathit{ctrl}  &\equiv  (?\Big(\rvar=\ravar \land v - g\cdot \Tvar > -\sqrt{\frac{g}{\rpvar}}\Big) \cup \humod{r}{\rpvar})  &\mathit{dyn}   &\equiv \abs{v} < \sqrt{\frac{g}{\rpvar}} < m \\
  \mathit{plant} &\equiv \humod{t}{0};~ \{x'=v,~ v'=r\cdot v^2 - g ~\&~ x \ge 0 \land v<0 \land t \le \Tvar\}&&&
\end{align*}}
\begin{theorem}[Skydiver Safety]$\rvar = \ravar \land{dc}\land{const}\land{dyn}\limply[(\mathit{ctrl};\mathit{plant})^*](x=0\limply v \leq m)$\ valid\end{theorem}
The proof mirrors the natural-language proof of Section~\ref{sec:natlang-proof} and builds upon Examples 1 and 2.
Recall that we use \emph{differential invariants} to reason about the drag equation, because it does not have a closed-form solution in decidable real arithmetic.
In the open-parachute case, recall that while {\tt dyn} is invariant ($\abs{v}$ never reaches the bound $\sqrt{\frac{g}{\rpvar}}$) it is not inductive because it approaches the bound asymptotically.
%(never reaches the bound $\sqrt{\frac{g}{\rpvar}}$), but not inductive, meaning the velocity $v$ approaches but .
Adding a \emph{differential ghost}~\cite{DBLP:journals/corr/abs-1104-1987} variable $y$, makes it possible to write an equivalent invariant that holds inductively.
An equivalent invariant can be derived mechanically: in this case $y^2\cdot\left(v+\sqrt{\frac{g}{\rpvar}}\right)=1$ which implies $\abs{v} <\sqrt{\frac{g}{\rpvar}}$.
\begin{figure}[h]
  \centering
      \begin{minipage}{0.9\textwidth}
\noindent\verb|#Example 3|\\
\kwassume\verb| assms:| $r=\ravar \land dc\_ \land const\_ \land dyn\_$\\
\kwinv\verb| DCCONST:| $dc\_ \land const\_$\\
\kwinv\verb| DYN:|     $dyn\_~\{$\\
\verb| |\kwind$~\Rightarrow~\{$\\
\verb|  |\kwstate\verb| loop|\\
\verb|  |\kwcase$~?(\_)~\Rightarrow$\\
\verb|    |\kwassume\ $vEps > vBound\_ \land r=ar$\\
\verb|    |\kwind\verb| pr:|       $g>0 \land \rpvar>0$\\
\verb|    |\kwind\verb| vBig:|     $\abs{v} \leq {\it{loop}}(\abs{v}) + g\cdot{t}$\\
\verb|    |\kwind\verb| vLoopBig:| ${\it{loop}}(\abs{v}) + g\cdot{\Tvar{}} < \sqrt{\frac{g}{\rpvar}}$\\
\verb|    |\kwfinally\ \kwhave\verb| tBound:| ${\it{loop}}(\abs{v}) + g\cdot{t} \leq {\it{loop}}(\abs{v}) + g\cdot\Tvar{}$\\
\verb|     |\kwusing\ $const$\ \kwby\ $\mathbb{R}$\\
\verb|    |\kwhave\verb| trans:|\\
\verb|       |$\forall w x y z~(w\leq{x}\limply{x}\leq{y}\limply{y<z}\limply{w<z})$ \kwby{} $\mathbb{R}$\\
\verb|    |\kwnote\verb| res = trans|$\ \ v\ \ vt\_\ \ v\Tvar{}\_\ \ vBound\_\ \ {\tt{v}}\ \ {\tt{gt}}\ \ {\tt{gEps}}$\\
%\verb|  note res = trans v vt_ vEps_ vBound_ vBig tBound vInitBig|\\
\verb|    |\kwshow{} $\_$ \kwusing{} {\tt{res}} \kwby{} $id$\\
\verb+  |+$\humod{r}{\rpvar}~\Rightarrow$\\
\verb|    |\kwassign$~\humod{r}{\rpvar}$\\
\verb|    |\kwinv\verb| consts:|$\rpvar>0\land{g>0}$\\
\verb|    |\kwlet\ $bound\_$ \verb|=| $-\sqrt{\frac{g}{\rpvar}}$\\
\verb|    |\kwghost{} $y=0, y'= -\frac{1}{2}\cdot{p}\cdot(v+bound\_)$\\
\verb|    |\kwinv\verb| ghostInv:| $y^2\cdot(v+bound\_)=1$\\
\verb|    |\kwfinally\ \kwshow\ $\_$ \kwusing{} \verb|ghostInv| \kwby{} $\mathbb{R}$\\
$\}\}\}$\ \kwfinally\ \kwshow\ $(x=0 \limply \abs{v} < m)$ \kwusing{} \verb|DCCONST DYN| \kwby{} $auto$\\
\end{minipage}
  \caption{Kaisar Proof of Skydiver Safety}
\label{fig:dlk-para}
\end{figure}

\newcommand{\velvar}{v}
Having finished the proof of Model~\ref{model:dive}, we consider a second example system that \emph{does} have solvable continuous dynamics, in which case ODEs can be symbolically executed directly without appealing to differential invariants.
Consider a one-dimensional model of a bouncing ball, with vertical position $y$, vertical velocity $\velvar$, acceleration due to gravity $g$ and initial height $H$.
This perfectly-elastic bouncing ball discretely inverts its velocity whenever it hits the ground ($y = 0$).
Because it started with $\velvar = 0$, we will prove that it never exceeds the initial height.
At the same time, we prove that it never goes through the floor ($y \geq 0$):
\begin{model}[Safety specification for bouncing ball]\label{model:ball}
\begin{align*}
  g>0~\land~&H>0\land y\leq H\land  \velvar=0 \limply\\
[&\{\{ ?(y>0\vee \velvar\geq 0)  \cup  \{?(y\leq 0\wedge \velvar < 0); \velvar := -\velvar\}\}\\
 &\ \ \ \ \{y'=\velvar,\velvar'=-g \& y\geq 0\}\\
&\}^*](0 \leq y \wedge y \leq H)
\end{align*}
\end{model}
The proof in Figure~\ref{fig:dlk} follows physical intuitions: total energy ($E\_$) is conserved, from which we show arithmetically that the height bound always holds.
\begin{figure}[h]
  \centering
\begin{minipage}{0.6\textwidth}
\verb|# Example 4|\\
\kwassume\verb| assms:|  $g>0 \land y\geq{H} \land H>0 \land v=0$\\
\kwlet\ $t(E\_)$ \verb|=| $t(\frac{v^2}{2} + H)$\\
\kwstate\verb| init|\\
\kwinv\verb| J:| $y \geq 0 \land E\_ = {\tt{init}}(E\_)~\{$\\
\verb|  |\kwind$~\Rightarrow$\\
\verb|   |\kwstate\verb| loop-init|\\
\verb|   |\kwmid\ $\{$\verb| |\kwshow\ $[\_~\cup~\_]\_$ \kwby\ $auto~\}$\\
\verb|     |\kwfirst\ $I:~E\_ = {\textit{loop-init}}(E\_)$\\
\verb|     |\kwthen$~\{$\\
\verb|       |\kwsolve\ $(\_ \land dom\_)$ \verb|t:| $t\geq{0}$ \verb|dom:|~$dom\_$\\
\verb|       |\kwshow\ $\_$ \kwby\ $auto~\}\}$\\
\kwshow\ $\_$  \kwusing\ \verb|J assms| \kwby\ $auto$\\
\end{minipage}  
  \caption{Kaisar Proof of Bouncing Ball Safety}
  \label{fig:dlk}
\end{figure}
Because this example has a solvable ODE, it suffices to add a construct for solving ODEs (below, {\tt dom} is short for domain constraint):
\[\infer[y_0=x,~y'=\theta(y)]{\tpchk{\G}{H}{\bsolve{\pat\ }{\pat_{t}\ }{\pat_{dom}\ }{\sproof}}{[\{x'=\theta~\&~Q\}]\phi,\Delta}{H_\phi}}
{\deduce
{\tpchk{\ext{\ext{\Gamma_Q}{dom}{\left(\forall{s}\in[0,t]~Q(s)\right)}}{t}{\left(t\geq 0\right)}}{\hhsub{H}{x}{y(t)}}{\sproof}{\subst[\phi]{x}{y(t)},\Delta}{H_\phi}}
{
 \match{\pat}{\{x'=\theta~\&~Q\}}{\Gamma_{x}}
&\mmatch{\G_{x}}{\pat_t}{t\geq{0}}{\Gamma_{t}}
&\mmatch{\G_{t}}{\pat_{dom}}{Q}{\G_Q}
} }\]
%Lastly, consider a toy example which not only has no polynomial solution, but also features exponential decay.
%\[x>0\wedge{y>0}\limply[\{x'=-x,y'=x\}]y>0\]
Lastly, consider the proof in Example 3.
%DBLP:books/daglib/0025392
We reason about unsolvable ODEs using \emph{differential invariants}~\cite{DBLP:journals/jar/Platzer16,DBLP:journals/corr/abs-1206-3357,DBLP:journals/logcom/Platzer10} and \emph{differential ghosts}~\cite{DBLP:journals/corr/abs-1104-1987},
which we add to the syntax of invariant proofs:
%TODO: Manualize DI
% In examples like this one, which exhibit decay properties, it is generally necessary to introduce auxilliary variables before an invariant can be established.
% This is done using the {\tt Ghost} operation, which can be mixed freely with invariants:
% \begin{figure}[h]
%   \centering
%   \begin{minipage}{0.55\linewidth}
% \begin{verbatim}
% assume assms:(x>0 & y>0)
% time init
% Ghost z'=x/2, z = (1/x)^(1/2)
% Inv JG:(x*z^2 = 1)
% Inv Jx:(x > 0)
% Inv Jy:(y > init(y))
% show (Jy > 0) using assms Jy by auto
% \end{verbatim}
%   \end{minipage}
%   \caption{Differential Invariant Example}
%   \label{fig:di}
% \end{figure}
\begin{align*}
  \iproof ::= & \cdots\ |\ \sghost{y}{\theta_1}{\theta_2}{\iproof}
\end{align*}
%TODO: Cite
Unlike in loops, it is essential for soundness that we do \emph{not} assume the current invariant (only previous invariants) while proving it.
\emph{Differential invariant}~\cite{DBLP:conf/tableaux/Platzer07} reasoning uses the \emph{differential} of a formula $(\phi)'$ to compute its Lie derivative, and then proves it to be inductive.
Traces are general enough to support loops and differental equations uniformly.
Because the differential equation $\alpha$ will modify $x$, we abstract over the variables of $\alpha\equiv{\{x'=\theta~\&~Q\}}$ (i.e. $x$):
%Trace management works as in loops, demonstrating the generality of the trace mechanism:
%Trace management, however, is as in loops:
%TODO: Very silly indexing here...
\[\infer{\tpchk{\G,[\alpha]\invlist{}}{H}{\sinv{x}{\psi}{\sproof_{Pre}}{\sproof_{Inv}}{\iproof}}{[\alpha]\phi,\Delta}{H_{x'}}}
   {\deduce{\tpchk{\Gamma,\ext{[\alpha]\invlist{}}{x}{[\alpha]\eag{\psi}}}{H}{\iproof}{[\{x'=\theta~\&~Q\}]\phi,\Delta}{H_{x'}}}
{\tpchk{\Gamma,\invlist{}}{H}{\sproof_{Pre}}{\eag{\psi},\Delta}{H_{Pre}}
   \hskip 0.1in\tpchk{\arb{\G}{\alpha},\invlist{},Q}{\arb{H}{\alpha}}{\sproof_{Inv}}{[\humod{x'}{\theta}](\eag{\psi})',\arb{\Delta}{\alpha}}{H_{Inv}}
}}\]
\[\infer{\tpchk{\G,[\alpha]\invlist{}}{H}{\finally{\sproof}}{[\alpha]\phi,\Delta}{H_{x'}}}
        {\tpchk{\arb{\G}{\alpha},\invlist{},Q}{\arb{H}{\alpha}}{\sproof}{\phi,\arb{\Delta}{\alpha}}{H_{x'}}}\]
        \[\infer[y\ \text{fresh}]{\tpchk{\G,\invlist{}}{H}{\sghost{y}{\theta_{y'}}{\theta_{y}}{\iproof}}{[\{x'=\theta_{x'}~\&~H\}]\phi,\Delta}{H_{x'}}}
 {\tpchk{\G,y=\eag{\theta_{y}}}{\Delta;H}{\iproof}{[\{x'=\eag{\theta_{x'}},y'=\eag{\theta_{y'}}~\&~H\}]\phi,\Delta}{H_{x'}} & \eag{\theta_{y'}}~{\rm linear\ in}~y}
\]
%TODO: More ghosts
When introducing a new variable $y$, we ensure the right-hand side of $y'$ is linear in $y$ to ensure the existence interval of the ODE does not change, which is essential for soundness~\cite{DBLP:journals/corr/abs-1104-1987}.
\newcommand{\GDHE}{\Gamma;\Delta;H\vdash}
\section{Metatheory}
\label{sec:metatheory}
\newcommand{\interps}[2]{\tinterps{#1}{#2}}
The value of a nominal $t(\theta)$ in the sequent-level state agrees with the value of $\theta$ in the corresponding program state.
We begin here with the simplest case, pseudo-nominals of variables $\hnow{H}{x}$, from which we then derive nominals of variables $\hnom{t}{H}{x}$ and arbitrary nominals $\hnom{t}{H}{\theta}$.

\begin{lemma}
For all $\corr{\eta}{H}$ and all variables $x$, $\interps{\hnow{H}{x}}{\seqstate{\eta}{H}} = \trlst{\eta}(x)$.
\end{lemma}
\begin{proof}
Straightforward induction on the derivation of $\corr{\eta}{H}$.

\oldcase{$\corr{\tsing{\nu}}{\hemp{}}$} Then $\interps{\hnow{H}{x}}{\seqstate{\eta}{H}} = \interps{\hnow{\tsing{\nu}}{x}}{\nu}=\nu(x)=\trlst{\eta}{x}$.

\oldcase{$\corr{\eta',\omega}{\hhany{H'}{x}{x_i}}$} Then $\interps{\hnow{H}{x}}{\seqstate{\eta}{H}} =\sval{\eta}{H}{x} =  \omega(x) = \trlst{\eta}(x)$.

\oldcase{$\corr{\eta',\omega}{\hheq{H'}{x}{x_i}{\theta}}$} Then $\interps{\hnow{H}{x}}{\seqstate{\eta}{H}} = \sval{\eta}{H}{x} =  \omega(x) = \trlst{\eta}(x)$.

\oldcase{$\corr{\eta',\omega}{\hhsub{H'}{x}{\theta}}$} Then $\interps{\hnow{H}{x}}{\seqstate{\eta}{H}} = \interps{\theta}{\seqstate{\eta'}{H'}} = \trlst{\eta}(x)$, where the last equation is from the definition of $\corr{}{}$.

\oldcase{$\corr{\eta',t}{H',t}$} Then $\interps{\hnow{H}{x}}{\seqstate{\eta}{H}}=\interps{\hnow{H'}{x}}{\seqstate{\eta'}{H'}}=\trlst{\eta'}(x)=\trlst{\eta}{x}$.
The remaining cases are symmetric.
\end{proof}
This lemma generalizes to arbitrary nominals, but first it will require the generalization of the coincidence theorem for \dL formulas ~\cite{DBLP:conf/cade/Platzer15} to nominals:

\begin{lemma}[Coincidence for ``now'']
\label{lem:coinc-now}
If $\corr{\eta}{H}$ and $y\notin\freevars{H}$ then $\interps{\hnow{H}{x}}{\seqstate{\eta}{H}} = \interps{\hnow{H}{x}}{\subst[\seqstate{\eta}{H}]{y}{r}}$ for all $y\neq{x}$ and $r\in\mathbb{R}$.
\end{lemma}
\begin{proof}
  By induction on $\corr{\eta}{H}$.

\oldcase{$\corr{\tsing(\omega)}{\hemp{}}$} $\interps{\hnow{H}{x}}{\seqstate{\eta}{H}} = \omega(x) = \subst[\omega]{y}{r}(x)$ because $x\neq{y}$.

\oldcase{$\corr{\eta',\omega}{\hhany{H'}{x}{x_i}}$} $\interps{\hnow{H}{x}}{\seqstate{\eta}{H}} = \interps{x}{\seqstate{\eta}{H}}=\omega(x)=\subst[\omega]{y}{r}(x)$ because $y\neq{x}$.

\oldcase{$\corr{\eta',\omega}{\hheq{H'}{x}{x_i}{\theta}}$} $\interps{\hnow{H}{x}}{\seqstate{\eta}{H}} = \interps{x}{\seqstate{\eta}{H}}=\omega(x)=\subst[\omega]{y}{r}(x)$ because $y\neq{x}$.

\oldcase{$\corr{\eta',\omega}{\hhsub{H'}{x}{\theta}}$} $\interps{\hnow{H}{x}}{\seqstate{\eta}{H}} = \interps{\theta}{\seqstate{\eta'}{H'}}=\interps{\theta}{\subst[\seqstate{\eta'}{H'}]{y}{r}}$ by term coincidence~\cite{DBLP:conf/cade/Platzer15} and because $y\notin\freevars{\theta}$ when $y\notin\freevars{H}$.

\oldcase{$\corr{\eta',t}{H,t}$} $\interps{\hnow{H}{x}}{\seqstate{\eta}{H}} = \interps{\hnow{H'}{x}}{\seqstate{\eta'}{H'}} = \interps{\hnow{H'}{x}}{\subst[{\seqstate{\eta'}{H'}}]{y}{r}} = \interps{\hnow{H}{x}}{\subst[{\seqstate{\eta}{H}}]{y}{r}}$. The remaining cases are symmetric.

\end{proof}
\begin{lemma}[Coincidence for Nominals]
\label{lem:coinc-nom}
If $\corr{\eta}{H}$ and $y\neq{\freevars{H}}$ and $x\neq{y}$ then $\interps{\hnom{t}{H}{x}}{\seqstate{\eta}{H}} = \interps{\hnom{t}{H}{x}}{\subst[\seqstate{\eta}{H}]{y}{r}}$.
\end{lemma}
\begin{proof}
  \oldcase{$\corr{\tsing{\omega}}{\hemp{}}$} True by contradiction.

  \oldcase{$\corr{\eta,\omega}{\hhany{H}{x}{x_i}}$} $\interps{\hnom{t}{H}{x}}{\seqstate{\eta}{H}} =_{def} \interps{\hnom{t}{H'}{x_i}}{\seqstate{\eta'}{H'}} =_{IH} \interps{\hnom{t}{H'}{x_i}}{\subst[{\seqstate{\eta'}{H'}}]{y}{r}} =_{def} \interps{\hnom{t}{H}{x}}{\subst[{\seqstate{\eta'}{H'}}]{y}{r}}$ by invariants for $\corr{}{}$. The proof for $\hheq{h}{x}{x_i}$ is symmetric.

  \oldcase{$\corr{\eta,\omega}{\hhsub{H}{x}{\theta}}$} $\interps{\hnom{t}{H}{x}}{\seqstate{\eta}{H}} = \interps{\hnom{t}{H'}{x}}{\seqstate{\eta'}{H'}} = \interps{\hnom{t}{H'}{x}}{\subst[{\seqstate{\eta'}{H'}}]{y}{r}}=\interps{\hnom{t}{H}{x}}{\subst[{\seqstate{\eta}{H}}]{y}{r}}$.
All other cases except $H,t$ are symmetric.

  \oldcase{$\corr{\eta,t}{H,t}$} $\interps{\hnom{t}{H}{x}}{\seqstate{\eta}{H}} = \interps{\hnow{H}{x}}{\seqstate{\eta'}{H'}}=\interps{\hnow{H}{x}}{\subst[{\seqstate{\eta'}{H'}}]{y}{r}}=\interps{\hnom{t}{H}{x}}{\subst[{\seqstate{\eta}{H}}]{y}{r}}$ by Lemma~\ref{lem:coinc-now}.
\end{proof}
Given the coincidence lemmas above, we have correspondence for nominals of variables:

\begin{lemma}[Correspondence for Nominals]
\label{lem:corr-nom}
For all $\corr{\eta}{H}$, all state names $t\in\dom{H}$, and all variables $x$, $\interps{\hnom{t}{H}{x}}{\seqstate{\eta}{H}} = t(\eta)(x)$.
\end{lemma}
\begin{proof}
  By induction on the derivation $\corr{\eta}{H}$.%, appealing to the lemma above in the case $H = H',t$.

\oldcase{$\corr{\tsing{\nu}}{\hemp{}}$} $\interps{\hnom{t}{H}{x}}{\seqstate{\nu}{H}} = \nu(x) = t(\eta)(x)$.

%TODO: Split off another lemma
\oldcase{$\corr{\eta',\omega}{\hhany{H'}{x}{x_i}}$} $\interps{\hnom{t}{H}{x}}{\seqstate{\eta}{H}} = \interps{\hnom{t}{H'}{x}}{\seqstate{\eta'}{H'}_x^{\omega(x)}} = \interps{\hnom{t}{H'}{x}}{\seqstate{\eta'}{H'}} = t(\eta')(x_i)=t(\eta')(x) = t(\eta)(x)$, by definition of $\hnom{t}{H}{x}$ and $\mathfrak{S}$, induction, definition of $\corr{}{}$ and definition of $t(\eta)$ respectively.

\oldcase{$\corr{\eta',\omega}{\hheq{H'}{x}{x_i}{\theta}}$} Symmetric.

\oldcase{$\corr{\eta',\omega}{\hhsub{H'}{x}{\theta}}$} $\interps{\hnom{t}{H}{x}}{\seqstate{\eta}{H}} = \interps{\hnom{t}{H'}{x}}{\seqstate{\eta'}{H'}} = t(\eta')(x) = t(\eta)(x)$.

\oldcase{$\corr{\eta',t}{H,t}$} $\interps{\hnom{t}{H}{x}}{\seqstate{\eta}{H}}=\interps{\hnow{H'}{x}}{\seqstate{\eta'}{H'}}=\trlst{\eta'}(x)=t(\eta)(x)$ by Lemma~\ref{lem:coinc-nom}.

\oldcase{other} Symmetric.
\end{proof}
Furthermore, note that prefixes preserve correctness of nominalization, thus nominals behave correctly even when evaluated from an intermediate state of $\eta$:

\begin{corollary}
 For all $\corr{\eta}{H}$, all state names $s\leq{t}\in\dom{H}$, all variables $x$, $\interps{\hnom{s}{t(H)}{x}}{\seqstate{\eta}{t(H)}} = s(t(\eta))(x)$
\end{corollary}
\begin{proof}
  From Lemma~\ref{lem:corr-nom}, it suffices to show for any $\corr{\eta}{H}$ and $t\in\dom{H}$ that $\corr{t(\eta)}{t(H)}$, which holds by a trivial induction since $\corr{\eta}{H}$ contains a derivation $\corr{\eta'}{H'}$ for all same-length prefixes $\eta'$ and $H'$, including $t(\eta)$ and $H(\eta)$. Then note since $s\leq{t}\in{H}$ then $s\in\dom{t(H)}$, so the preconditions of the lemma are satisfied.
%  By induction on $\corr{\eta}{H}$ using the lemma above.
\end{proof}

Because the meaning of nominal terms is uniquely determined by the meaning of nominal variables, the above lemmas suffice to show that \emph{all} nominal and pseudo-nominal terms are well-behaved:

\begin{theorem}[Nominal Term Correspondence]
%TODO: Please both now and nom
\label{thm:ntc}
For all $\corr{\eta}{H}$, all state names $s\leq{t}\in\dom{H}$, all terms $\theta$, $\interps{\hnom{s}{t(H)}{\theta}}{\seqstate{\eta}{t(H)}} = \interps{\theta}{s(t(\eta))}$.
\end{theorem}
\begin{proof}
  By induction on the structure of terms $\theta$.

\oldcase{$\theta=x$} Then $\interps{\hnom{s}{t(H)}{\theta}}{\seqstate{\eta}{t(H)}} = \interps{\hnom{s}{t(H)}{x}}{\seqstate{\eta}{t(H)}} = s(t(\eta))(x) = \interps{\theta}{s(t(\eta))}$ by corollary above.

\oldcase{$\theta=q\in\mathbb{Q}$} Then $\interps{\hnom{s}{t(H)}{\theta}}{\seqstate{\eta}{t(H)}} = q = \interps{\theta}{s(t(\eta))}$.

\oldcase{$\theta=\otimes(\theta_1,\theta_2)$ for any operator $\otimes$} Then $\interps{\hnom{s}{t(H)}{\theta}}{\seqstate{\eta}{t(H)}} = \interps{\hnom{s}{t(H)}{\otimes(\theta_1,\theta_2)}}{\seqstate{\eta}{t(H)}} = \otimes(\interps{\hnom{s}{t(H)}{\theta_1}}{\seqstate{\eta}{t(H)}},\interps{\hnom{s}{t(H)}{\theta_2}}{\seqstate{\eta}{t(H)}}) = \otimes(\interps{\theta_1}{s(t(\eta))},\interps{\theta_2}{s(t(\eta))}) = \interps{\otimes(\theta_1,\theta_2)}{s(t(\eta))} = \interps{\theta}{s(t(\eta))}.$
\end{proof}
%TODO: Define and explain nominalization
%TODO: Soundness and totality of normalization
%TODO: Prove that nominalized program transition exists because we defined the dLN interpretation correctly.
%TODO: Prove that the "sequent-state-ness" of \Gamma is preserved at each inductive application
In the above theorems, we assumed that the dynamic and static traces are always in correspondence  $\corr{\eta}{H}$ holds.
We now show that this is always the case at every proof state within a proof:

%TODO: Is there a way to show truth in the end state specifically?
%TODO: Mention this is stronger because inductively holds, also welp can't do any better if you're not proving a box modality, but oh also there's a diamond version of this.
\begin{theorem}[Intermediate Nominal Term Correspondence]
 For any program $\alpha$ with $(\omega,\nu)\in\interp{\alpha}$ and any $\corr{\eta_{\omega}}{H}$ where $\trlst{\eta_\omega} = \omega$ and $\tpchk{\G}{H}{\sproof}{\hnow{H}{[\alpha]\phi}}{H_1}$ there exists $\eta_\nu$ where $\trfst{\eta_\nu} = \nu$ and $\corr{\eta_\omega\eta_\nu}{H_1}$.
\end{theorem}
%This statement has several subtleties: first notice correspondence is maintained regardless of any assumptions on $\Gamma$, but it is essential that an execution of $\alpha$ exists, for the remainder of the dynamic trace $\eta_\nu$ is generated from those states.
% The notation $\eta'_\omega \supseteq\eta_\omega$ says the traces agree on $\dom{\eta_\omega}$ and have the same length and structure, but that we are allowed to extend the states of $\eta'_\omega$ with fresh ghost variables $x_i$ --- this is necessary because $\eta_\nu$ may perform ghosting, and recall that ghosts have a constant value throughout an entire trace.
%Note that we do not enforce any relation between $\nu$ and $\trlst{\eta\nu}$ because in general there may not be any:
%%As shown in their proof-checking rules, the nondeterministic constructs $\alpha\cup\beta$ and $\alpha^*$ ``forget'' any internal state definitions and do not make them available in the output trace $H'$, as they are only meaningful in the local context.
%This does not weaken the theorem however: this theorem holds inductively at every proof state and thus shows that even for choices and loops, nominals have their intended meaning in the local context.
Note that  $\trlst{\eta_\nu}$ is not always $\nu$: 
Following the proof-checking rules, any states defined inside the nondeterministic constructs $\alpha\cup\beta$ and $\alpha^*$ leave scope and are absent in the final trace.
Since the theorem holds inductively at every proof state it thus shows that even for choices and loops, nominals have their intended meaning in the local context.
\begin{proof}
The proof is by induction on the derivation $\tpchk{\G}{H}{\sproof}{[\alpha]\phi}{H'}$.
The only significant cases are those which modify the trace.

\oldcase{show} Let $\eta_\omega'=\eta_\omega$ and $\eta_\nu = (\trlst{\eta_\omega}),$ then $\eta_\omega'\eta_\nu=\eta_\omega$ and $H'=H$ so $\corr{\eta_\omega'\eta_\nu}{H'}$ because $\corr{\eta_\omega}{H}$ by assumption.

\oldcase{$\bcase{p_\alpha}{\sproof_1}{p_\beta}{\sproof_2}$} Let $\eta_\omega'=\eta_\omega$ and $\eta_\nu = (\trlst{\eta_\omega}),$ then $\eta_\omega'\eta_\nu=\eta_\omega$ and $H'=H$ so $\corr{\eta_\omega'\eta_\nu}{H'}$ because $\corr{\eta_\omega}{H}$ by assumption.

\oldcase{$\bassign{x}{\tilde{\theta}}{\sproof}$, $\alpha=x:=\hnow{H}{\theta};\hnow{H}{\alpha'}, x\notin\boundvars{\alpha'}$} 
Let $H^*=\hhsub{H}{x}{\now{H}{\theta}}, \eta^*=\eta,\omega,\subst[\omega]{x}{\interps{\theta}{\omega}}$. 
To show $\corr{\eta^*}{H^*}$ it suffices to show $\interps{\now{H}{\theta}}{\seqstate{\eta,\omega}{H}=\interps{\theta}{\omega}}$ which holds by Theorem~\ref{thm:ntc}.
To apply the IH, lastly observe $\hnow{H^*}{\alpha'}=\subst[\hnow{H}{\alpha'}]{x}{\hnow{H}{\theta}}$.
Then by IH, $\exists \eta^*_\nu~\trfst{\nu}=\subst[\omega]{x}{\interps{\theta}{\omega}}$, so let $\eta_\nu=\omega,\eta^*_\nu$ and observe $\eta,\omega,\subst[\omega]{x}{\interps{\theta}{\omega}},\eta^*_\nu=\eta,\omega,\eta^*$ so $\corr{\eta,\omega,\eta^*}{H'}$ by IH.

\oldcase{$\bassign{x}{\tilde{\theta}}{\sproof}$, $\alpha=x:=\hnow{H}{\theta};\hnow{H}{\alpha'}, x\in\boundvars{\alpha'}$} 
By symmetry since the rule for $\corr{\eta,\omega}{\hheq{H}{x}{x_i}{\theta}}$ is symmetric with $\corr{\eta,\omega}{\hhsub{H}{x}{\theta}}$.
To apply the IH observe  $hnow{H^*}{\alpha'}=\hnow{H}{\alpha'}$.

\begin{lemma}[Ghosting]
\label{lem:intermed-ghost}
  If $\corr{\eta,\omega}{H}$ 
  and $(\omega,\nu)\in\interp{\alpha}$  
  then $\exists\eta_\nu~\corr{(\eta,\omega,\eta_\nu)}{\arb{H}{\alpha}}$.
\end{lemma}
\begin{proof}
  Let $x_1,\ldots,x_n=\boundvars{\alpha}$.
  Define $\omega_1=\subst[\omega]{x_1}{\nu(x_1)}, \omega_i=\subst[\omega_{i-1}]{x_i}{\nu(x_i)}$ for all $i\leq n$.
  Then $\omega_n=\nu$ because by bound effect $\omega$ and $\nu$ differ only by $\boundvars{\alpha}$.
  By definition, $\arb{H}{\alpha}=H,\hrany{x_1}{x_{1}^*},\ldots,\hrany{x_n}{x_{n}^*}$ for ghosts $x_{i}^*$.
  Then let $\eta_\nu= \eta,\omega,\omega_1,\ldots,\omega_n$ and the result holds.
\end{proof}

\oldcase{($(\bcon{x:\phi}{\sproof_1}{\sproof_2})$ and $(\alpha=\alpha_1;\alpha_2)$)} 
By inversion on $(\omega,\nu)\in\interp{\alpha}, \exists \mu~(\omega,\mu)\in\interp{\alpha_1}$ and $(\mu,\nu)\in\interp{\alpha_2}$.
By the IH on $\sproof_1, \exists \corr{\eta_\mu}{H'}$. Let $\vec{x}=x_1,\ldots,x_n=\boundvars{\alpha}$ then define $H^*=H,\hrany{x_1}{x_{1,i}},\ldots,\hrany{x_n}{x_{n,i}}$ and define $\omega_1=\subst[\omega]{x_1}{\mu(x_1)}, \omega_i=\subst[\omega_{i-1}]{x_i}{\mu(x_i)}$ for all $i\leq n$, then observe $\omega_n = \mu$ because by bound effect~\cite{DBLP:conf/cade/Platzer15} for programs, $\omega$ and $\mu$ differ only on $\boundvars{\alpha}$.
Let $\eta^*=\eta,\omega,\omega_1,\ldots,\omega_n$ and observe $\corr{\eta^*}{H^*}$ so we can apply the IH on $\sproof_2$, yielding $\exists \corr{\eta^*_\nu}{H''}$.
Now let $\eta_\nu = \hrany{x_1}{x_{1,i}},\ldots,\hrany{x_n}{x_{n,i}},\eta^*_\nu$ and observe $\eta,\omega,\eta^* = \eta,\omega,\hrany{x_1}{x_{1,i}},\ldots,\hrany{x_n}{x_{n,i}},\eta^*_\nu$ so
$\corr{\eta,\omega,\eta_\nu}{H''}$, concluding the case.

\oldcase{$\bassert{x}{\tilde{\phi}}{\sproof}, \alpha=?(\phi);\alpha_1$} Let $H^*=H,\eta^*=\eta$ then $\corr{\eta^*}{H^*}$ and by IH $\exists \eta^*_\nu$ where $\corr{\eta^*\eta^*_\nu}{H'}$.
Let $\eta_\nu=\eta^*_\nu$ and the result holds by definition of $\eta^*$ and $\eta_\nu$.

\oldcase{$\slet{p}{\tilde{e}}{\sproof}$} Symmetric.

\oldcase{$\slet{t(?x)}{\tilde{e}}{\sproof}$} Symmetric.

\oldcase{$\snote{x}{\fproof}{\sproof}$} Symmetric.

\oldcase{$\have{x}{\tilde{e}}{\sproof_1}{\sproof_2}$} Symmetric, except apply the IH on $\sproof_2$.

\oldcase{$\sstate{t}{\sproof}$} Let $H^*=H,t$ and $\eta^*=\eta,t$ so $\corr{\eta^*}{H^*}$ so the result follows by IH.

%TODO: Might want to make this more detailed 
\oldcase{$\bsolve{\pat_{ode}}{\pat_{t}}{\pat_{dom}}{\sproof}$} 
By inversion on $(\omega,\nu)\in\interp{x'=\theta;\alpha_1}$ there exists $\mu = \subst[\omega]{x}{y(t)}$ for some $t \geq 0$ such that $\forall s~\in~[0,t]~Q(s)$ and where $(\omega,\mu)\in\interp{x'=\theta}$ and $(\mu,\nu)\in\interp{\alpha_1}$.
Let $H^*=\hhsub{H}{x}{y(t)}$ and $\eta^*=\eta,\omega,\subst[\omega]{x}{\interps{y(t)}{\omega}}$ then have $\corr{\eta^*}{H^*}$ by Theorem~\ref{thm:ntc} saying $\interps{y(t)}{\omega} = \interps{\hnow{H}{y(t)}}{\seqstate{\eta,\omega}{H}}$ for any such $t$.
Then we can apply the IH since $\hnow{H^*}{\alpha_1}=\hnow{H}{\subst[\alpha_1]{x}{\interps{y(t)}{\omega}}}$, yielding $\eta^*_\nu$ where $\corr{\eta,\omega,\subst[\omega]{x}{\interps{y(t)}{\omega}},\eta^*_\nu}{H'}$.
So let $\eta_\nu=\subst[\omega]{x}{\interps{y(t)}{\omega}},\eta^*_\nu$ and the result holds.

\oldcase{$\sinv{J}{\tilde{\phi}}{\sproof_1}{\sproof_2}{\iproof}, \alpha =\alpha^*_1;\alpha_2$}
Let $H^*=H,\eta^*=\eta$ then $\corr{\eta^*}{H^*}$. by IH $\exists \eta^*_\nu$ where $\corr{\eta,\omega,\eta^*_\nu}{H'}$ so let $\eta_\nu=\eta^*_\nu$ and then $\corr{\eta,\omega,\eta_\nu}{H'}$.

\oldcase{$\finally{\sproof}, \alpha =\alpha_1^*;\alpha_2$}
By inversion, exists $\mu$ where $(\omega,\mu)\in\interp{\alpha_1^*}$ and $(\mu,\nu)\in\interp{\alpha_2}$.
By Lemma~\ref{lem:intermed-ghost} have $\corr{\arb{\eta,\omega}{\alpha}}{\arb{H}{\alpha}}$ 
(where $\arb{(\eta,\omega)}{\alpha} = \eta,\omega,\omega_1,\ldots,\omega_n$)
and by bound effect~\cite{DBLP:conf/cade/Platzer15} lemma have $\omega_n=\mu$.
By IH have $\eta^*_\nu$ where $\corr{\eta,\omega,\omega_1,\ldots,\omega_n,\eta^*_\nu}{H'}$ so let $\eta_\nu=\omega_1,\ldots,\omega_n,\eta^*_\nu$.

\oldcase{$\sinv{J}{\tilde{\phi}}{\sproof_1}{\sproof_2}{\iproof}, \alpha =\{x'=\theta\};\alpha_2$}
Symmetric to the case for loops.

%TODO: Check this again when more awake.
\oldcase{$\sghost{y}{\theta_2}{\theta_3}{\iproof}, \alpha =\{x'=\theta\};\alpha_2$}
By inversion, $\exists \mu = \subst[\omega]{x}{\varphi(t)}$ for some $t\geq 0$ such that $\varphi(t)$ is the solution of $x'=\theta$ for time $t$ and $\forall s\in[0,t]~\subst[\omega]{x}{\varphi(t)}\in\interp{Q}$ 
and $(\omega,\mu)\in\interp{\{x'=\theta\}}$ and $(\mu,\nu)\in\interp{\alpha_2}$.
By linearity of $\{y'=\theta_2\}$, the solution of $\{x'=\theta,y'=\theta_2\}$ exists for the same time $t$ and the solution agrees on $x$.
Since $y$ is a ghost variable then $y\notin\freevars{\alpha_2}$ and so $(\mu,\nu)\in\interp{\alpha_2}$ so the IH applies and $\exists \eta^*_\nu$ where $\corr{\eta\omega,\eta^*_\nu}{H'}$ and letting $\eta_\nu=\eta^*_\nu$ the result holds.

\oldcase{$\finally{\sproof}, \alpha =\{x'=\theta\};\alpha_2$}
Symmetric to the case for loops, except note $\boundvars{\{x'=\theta\}} = \{x,x'\}$ so we apply the bound effect~\cite{DBLP:conf/cade/Platzer15} lemma only for $x$.
We assume as a side condition that $x'\notin{\freevars{\eta}}$ in which case $\corr{\eta^*}{H^*}$ still holds when setting $\omega_n(x')=\mu(x')$.
\end{proof}

\subsection{Nominalization}
The relation $\corr{\eta}{H}$ relates static and dynamic traces, but not the executed program $\alpha$.
%While the relation $\corr{\eta}{H}$ adequately describes the relation between static and dynamic traces, it does not tell us how the output trace $\eta_\nu$ is related to the executed program $\alpha$.
We wish for the trace $\eta_\nu$  to assign each nominal $t$ the actual state that $\alpha$ had during the state $t$ in the proof.\footnote{Since we can add names at any point in a proof, this implies that the proof and program agree at every state.}
%Intuitively, we wish to say that the trace $\eta_\nu$ assigns to each nominal $t$ the actual state that $\alpha$ was in during the state which was named $t$ in the proof.\footnote{Since we can add names at any point in a proof, this implies that the proof and program agree at every state.} 
%While we could simply take traces $\eta$ as a definition of the semantics of $\dL$, that would be a poor solution since they abstract away a significant amount of information (such as the details of continuous evolutions).
%Rather it is better to relate them directly to an external semantics for hybrid programs which both captures the full dynamics of hybrid programs and also allows us to describe internal states.

We show this with semantics of the nominal hybrid logic $\dLN$,~\cite{DBLP:journals/entcs/Platzer07} which extends the logic of \dL with propositions $t$ which are true iff the current state is the unique state identified by $t$.\footnote{\dLN also adds propositions $@_i\phi$ indicating truth of $\phi$ in state $\phi$, but they are not needed to specify our metatheory.}
Accordingly, we extend the semantics  with interpretations that assign a specific state to each $t$.
We write the interpretation corresponding to a trace $\eta$ as $\mkint{\eta}$, i.e. the interpretation which maps each $t$ to $\trlst{t(\eta)}$.
Thus $\omega\in\interp{\phi}{\mkint{\eta}}$ means $\phi$ holds in state $\omega$ in the interpretation constructed from $\eta$, and likewise for $(\omega,\nu)\in\interp{\alpha}{\mkint{\eta}}$.
We give a \emph{nominalization judgment} $\mknom{\alpha}{\sproof}{\alpha_h}$ which augments a program $\alpha$ with a test $?(t)$ for each named state $\sstate{t}{}$ in $\sproof$, producing a $\dLN$ program $\alpha_h$.
Every transition of $\alpha_h$ is a transition of $\alpha$ because we do not introduce state mutation.
We show the converse holds too: all our additional tests (which depend solely on $\mkint{\eta}$) pass.
This formally justifies the claim that the states of $\eta$ match the states of $\alpha$:
%The transitions of $\alpha_1$ are a subset of those for $\alpha$ since it simply adds some tests.

\newcommand{\mkfnom}[2]{nom(#1,#2)}
\begin{theorem} If $\tpchk{\G}{H}{\sproof}{[\alpha]\phi,\Delta}{H_\alpha}$ for $\corr{\eta,\omega}{H}$ and $\corr{\eta,\omega,\eta_\nu}{H_\alpha}$ then let $\nu$ such that $(\omega,\nu)\in\interp{\alpha}$.
Then in $\dLN$, $(\omega,\nu)\in\interp{\mkfnom{\alpha}{\sproof}}{\mkint{\eta}}$.
\end{theorem}

\begin{proof}
  By induction on the derivation of $\tpchk{\G}{H}{\sproof}{[\alpha]\phi}{H'}$.
  The essence of the proof is the definition of $\mknom{\alpha}{\sproof}{\alpha'}$:
\begin{align*}
%TODO: Add the remaining cases
\mknom{\alpha}{\sstate{t}{\sproof}}{?t;\mkfnom{\alpha}{\sproof}}&\\
\mknom{\alpha\cup\beta}{\sproof}{\alpha\cup\beta}&\\
\mknom{\alpha^*;\beta}{\sinv{pre}{\sproof_1}{ind}{\sproof_2}{\sproof}}{\mkfnom{\alpha^*;\beta}{\sproof}}&\\
%\mknom{\alpha^*;\beta}{\finally{\sproof}}{\alpha^*;\beta'}&~ \text{for}~ \mknom{\beta}{\sproof}{\beta}\\
\mknom{\alpha^*;\beta}{\finally{\sproof}}{\alpha^*;\mkfnom{\beta}{\sproof}}&\\
\mknom{?(\phi);\alpha}{\bassert{x}{\phi}{\sproof}}{?(\phi);\mkfnom{\alpha}{\sproof}}&\\
\mknom{\alpha}{\have{x}{\phi}{\sproof_1}{\sproof_2}}{\mkfnom{\alpha}{\sproof_2}}\\
\mknom{x:=\theta;\alpha}{\bassign{x}{\tilde{\theta}}{\sproof}}{x:=\theta;\mkfnom{\alpha}{\sproof}}&\\
\mknom{\prandom{x};\alpha}{\bassignany{x}{\sproof}}{\prandom{x};\mkfnom{\alpha}{\sproof}}&\\
\mknom{\pevolvein{x'=\theta}{H};\alpha}{\bsolve{\pat_{ode}}{\pat_{t}}{\pat_{dom}}{\sproof}}{\pevolvein{x'=\theta}{H};\mkfnom{\alpha}{\sproof}}&\\
\mknom{\pevolvein{x'=\theta}{H};\alpha}{\sghost{y}{\theta_2}{\theta_3}{\iproof}}{\pevolvein{y'=\theta,\alpha}{H};\beta'}&\\
\ \ ~ \text{for}~ \mknom{\pevolvein{x'=\theta}{H};\alpha}{\sproof}{\pevolvein{\alpha}{H};\beta'}\\
\mknom{\pevolvein{x'=\theta}{H};\alpha}{\sinv{pre}{\sproof_1}{ind}{\sproof_2}{\sproof}}{\mkfnom{\pevolvein{x'=\theta}{H};\alpha}{\sproof}}&\\
\mknom{\pevolvein{x'=\theta}{H};\alpha}{\finally{\sproof}}{\pevolvein{x'=\theta}{H};\mkfnom{\alpha}{\sproof}}&\\
\mknom{\alpha}{\slet{\pat}{\tilde{e}}{\sproof}}{\mkfnom{\alpha}{\sproof}}&\\
\mknom{\alpha}{\slet{t(?X)}{\tilde{e}}{\sproof}}{\mkfnom{\alpha}{\sproof}}&\\
\mknom{\alpha}{\snote{x}{\fproof}{\sproof}}{\mkfnom{\alpha}{\sproof}}&\\
\mknom{\alpha^*}{\bcase{p_\phi}{\sproof_\phi}{p_\alpha}{\sproof_\alpha}}{\alpha^*}& \\
\mknom{\alpha;\beta}{\bcon{\psi}{\sproof_\psi}{\sproof_\phi}}{\left(\mkfnom{\alpha}{\sproof_1}\right);\left(\mkfnom{\beta}{\sproof_2}\right)}& \\
\mknom{\alpha}{\sfocus{\pat}{\sproof}}{\alpha}&
\end{align*}

Throughout the proof, let $\eta'\equiv\eta,\omega,\omega^*_\nu$, that is the trace resulting from the IH, such that $\corr{\eta'}{H'}$.

\oldcase{$\mknom{?(\phi);\alpha}{\bassert{x}{\phi}{\sproof}}{?(\phi);\mkfnom{\alpha}{\sproof}}$}
By inversion $\omega\in\interp{\phi}$ and $(\omega,\nu)\in\interp{\alpha}$.
By IH, $\exists \alpha'~\mknom{\sproof}{\alpha}{\alpha'}$ and $(\omega,\nu)\in_{}\interp{}$.
Then $\mknom{\bassert{x}{\phi}{\sproof}}{?(\phi);\alpha}{?(\phi),\alpha'}$ and since $\omega\in\interp{\phi}$ and $\phi$ is nominal-free, then $\omega\in\interp{\phi}{\mkint{\eta'}}$ for any $\eta'$ and also $(\omega,\nu)\in\interp{\alpha'}$ so $(\omega,\nu)\in\interp{?(\phi);\alpha'}{\mkint{\eta'}}$.

\oldcase{$\mknom{\alpha}{\have{x}{\phi}{\sproof_1}{\sproof_2}}{\mkfnom{\alpha}{\sproof_2}}$}
By the IH, $\exists \alpha'~\mknom{\sproof}{\alpha}{\alpha'}$ and $(\omega,\nu)\in\interp{\alpha'}{\mkint{\eta'}}$ then by inversion $\mknom{\alpha}{\have{x}{\phi}{\sproof_1}{\sproof_2}}{\alpha'}$ for the same $\alpha'$, completing the case.

\oldcase{$\mknom{\alpha}{\slet{\pat}{\tilde{e}}{\sproof}}{\mkfnom{\alpha}{\sproof}}$}
Symmetric.

\oldcase{$\mknom{\alpha}{\slet{t(?X)}{\tilde{e}}{\sproof}}{\mkfnom{\alpha}{\sproof}}$}
Symmetric.

\oldcase{$\mknom{\alpha}{\snote{x}{\fproof}{\sproof}}{\mkfnom{\alpha}{\sproof}}$}
Symmetric.

\oldcase{$\mkfnom{\alpha^*}{\bcase{p_\phi}{\sproof_\phi}{p_\alpha}{\sproof_\alpha}}$}
Symmetric. 

\oldcase{$\mkfnom{\alpha}{\sfocus{\pat}{\sproof}}$}
Symmetric. 

\oldcase{$\bcon{\psi}{\sproof_\psi}{\sproof_\phi}$}
By inversion, $\exists\mu~((\omega,\mu)\in\interp{\alpha}\land(\mu,\nu)\in\interp{\beta})$.
Split the trace $\eta,\omega¸\eta_\nu$ into $\eta,\omega,\eta_\mu,\mu,\eta_\nu^*$.
By IH1, $(\omega,\mu)\in{\interp{\mkfnom{\sproof_1}{\alpha}}{\mkint{\eta,\omega,\eta_\mu,\mu}}}$.
By interpretation weakening, $(\omega,\mu)\in{\interp{\mkfnom{\sproof_1}{\alpha}}{\mkint{\eta,\omega¸\eta_\nu}}}$.
By Lemma~\label{lem:intermed-ghost}, IH2 is applicable and $(\mu,\nu)\in{\interp{\mkfnom{\sproof_2}{\beta}}{\mkint{\eta,\omega¸\eta_\nu}}}$.
The result holds by semantics of $[;]$.

%{\left(\mkfnom{\alpha}{\sproof_1}\right);\left(\mkfnom{\beta}{\sproof_2}\right)}& \\

\oldcase{$\mknom{\alpha}{\sstate{t}{\sproof}}{?t;\mkfnom{\alpha}{\sproof}}$}
%TODO: Note this means that \omega is not always the last thing in the trace
Then by IH with $\eta^*=\eta,\omega,t$ and $H^*=H,t$ have $\exists\alpha'~\mknom{\sproof}{\alpha}{\alpha'}$ for some $\alpha'$ where $(\omega,\nu)\in_{\mkint{\eta'}}\interp{\alpha'}$.
Then $\mknom{\alpha}{\sstate{t}{\sproof}}{?t;\alpha'}$.
We have $\mkint{\eta'}=\subst[\mkint{\eta^*}]{t}{\omega}$. Since $\mkint{\eta'}=\mkint{\eta^*}$ on $\Sigma(\alpha')$ then by coincidence~\cite{DBLP:conf/cade/Platzer15}, $(\omega,\nu)\in\interp{\alpha'}{\mkint{\eta^*}}$.
And by the definition of $\eta^*$ we have $\mkint{\eta^*}(t)=\omega$ so $\omega\in\interp{t}{\mkint{\eta^*}}$ and thus $(\omega,\nu)\in{?t;\alpha'}{\mkint{\eta^*}}$.

\oldcase{$\mknom{\alpha\cup\beta}{\sproof}{\alpha\cup\beta}$}
Since $(\omega,\nu)\interp{\alpha\cup\beta}$ then $(\omega,\nu)\in\interp{\alpha\cup\beta}{\mkint{\eta^*}}$ and since $\mknom{\alpha\cup\beta}{\sproof}{\alpha\cup\beta}$ the case is complete.

\oldcase{$\mknom{x:=\theta;\alpha}{\bassign{x}{\tilde{\theta}}{\sproof}}{x:=\theta;\mkfnom{\alpha}{\sproof}}$}
By IH $\exists \alpha'~\mknom{\sproof}{\alpha}{\alpha'}$ then $\mknom{x:=\theta;\alpha}{\bassign{x}{\tilde{\theta}}{\sproof}}{x:=\theta;\alpha'}$ and $(\subst[\omega]{x}{\interps{\theta}{\omega}},\nu)\in\interp{\alpha'}{\mkint{\eta^*}}$ and since (by inversion) $(\omega,\subst[\omega]{x}{\interps{\theta}{\omega}})\in{\interp{x:=\theta}}$ then also $(\omega,\subst[\omega]{x}{\interps{\theta}{\omega}})\in\interp{x:=\theta}{\mkint{\eta'}}$
and finally $(\omega,\nu)\in{\interp{x:=\theta;\alpha'}{\mkint{\eta'}}}$.

\oldcase{$\mknom{\prandom{x};\alpha}{\bassignany{x}{\sproof}}{\prandom{x};\mkfnom{\alpha}{\sproof}}$}
By IH $\exists \alpha'~\mknom{\sproof}{\alpha}{\alpha'}$ then $\mknom{x:=\theta;\alpha}{\bassignany{x}{\sproof}}{\prandom{x};\alpha'}$ and $(\subst[\omega]{x}{r},\nu)\in\interp{\alpha'}{\mkint{\eta^*}}$ and since (by inversion) $(\omega,\subst[\omega]{x}{r})\in{\interps{\prandom{x}}{\omega}}$ then also $(\omega,\subst[\omega]{x}{r})\in{\interp{\prandom{x}}{\mkint{\eta'}}}$
and finally $(\omega,\nu)\in{\interp{\prandom{x};\alpha'}{\mkint{\eta'}}}$.

\oldcase{$\mknom{\pevolvein{x'=\theta}{H};\alpha}{\bsolve{\pat_{ode}}{\pat_{t}}{\pat_{dom}}{\sproof}}{\pevolvein{x'=\theta}{H};\mkfnom{\alpha}{\sproof}}$}
By IH have $\exists \alpha'~\mknom{\sproof}{\alpha}{\alpha'}$ then $\mknom{\pevolvein{x'=\theta}{Q};\alpha}{\bsolve{\pat_{ode}}{\pat_t}{\pat_{dom}}{\sproof}}{\pevolvein{x'=\theta}{H};\alpha'}$ and $(\subst[\omega]{x}{\varphi(t)},\nu)\in\interp{\alpha'}{\mkint{\eta^*}}$ for some $t\geq 0$ such that $\forall~s\in[0,t]~\varphi(s)\in\interp{Q}$ and where $\varphi$ is the unique solution to $x'=\theta$ on $[0,t]$. 
Since (by inversion) $(\omega,\subst[\omega]{x}{\varphi(t)})\in{\interps{\pevolvein{x'=\theta}{Q}}{\omega}}$ then also $(\omega,\subst[\omega]{x}{\varphi(t)})\in{\interp{\pevolvein{x'=\theta}{Q}{\mkint{\eta'}}}}$
and finally $(\omega,\nu)\in{\interp{\pevolvein{x'=\theta}{Q}}{\mkint{\eta'}}}$.

\oldcase{$\mknom{\alpha^*;\beta}{\sinv{pre}{\sproof_1}{ind}{\sproof_2}{\sproof}}{\mkfnom{\alpha^*;\beta}{\sproof}}$}
By IH, have $\exists~\beta'~\mknom{\sproof}{\alpha^*;\beta}{\alpha^*;\beta'},$ and
$\mknom{\alpha^*;\beta}{\finally{\sproof}}{\alpha^*;\beta'}$ for the same $\beta'$
and by IH,  $(\omega,\nu)\in{\alpha^*;\beta'}{\mkint{\eta'}}$.

\oldcase{$\mknom{\alpha^*;\beta}{\finally{\sproof}}{\alpha^*;\mkfnom{\beta}{\sproof}}$}
By inversion, $\exists \mu~(\omega,\mu)\in\interp{\alpha^*}$ and $(\mu,\nu)\in\interp{\beta}$.
By IH, $\exists~\beta'~\mknom{\sproof}{\beta}{\beta'}$ and $(\mu,\nu)\in\interp{\beta'}{\mkint{\eta'}}$.
Then since $(\omega,\mu)\in\interp{\alpha^*}$ then also $(\omega,\mu)\in\interp{\alpha^*}{\mkint{\eta'}}$ and thus $(\omega,\nu)\in\interp{\alpha^*;\beta'}{\mkint{\eta'}}$ as desired.

\oldcase{$\mknom{\pevolvein{x'=\theta}{H};\alpha}{\sghost{y}{\theta_2}{\theta_3}{\iproof}}{\pevolvein{y'=\theta,\alpha}{Q};\beta'}$}
%TODO: Need to actually think about this case a bit more:
By inversion, $\exists \mu~(\omega,\mu)\in\interp{\pevolvein{x'=\theta}{Q}}$ and $(\mu,\nu)\in\interp{\beta}$ and $\mu=\subst[\omega]{x}{\varphi_x(t)}$ for some $t\geq 0$ such that $\forall~s\in[0,t]~\subst[\omega]{x}{\varphi_x(t)}\in\interp{Q}$.
By applying the IH, we then have $\exists~\beta'~\mknom{\sproof}{\pevolvein{x'=\theta,y'=\theta_2}{H};\beta}{\pevolvein{x'=\theta,y'=\theta_2}{H};\beta'},$ and
$\mknom{\pevolvein{x'=\theta,y'=\theta_2};\beta}{\sghost{y}{\theta_2}{\theta_3}{\iproof}}{\pevolvein{x'=\theta}{H};\beta'}$ for the same $\beta'$ and by IH, $(\subst[\omega]{y}{y_0},\nu)\in\interp{\pevolvein{x'=\theta,y'=\theta_2};\beta'}{\mkint{\eta'}}$.
By soundness for dG, $(\omega,\nu)\in\interp{\pevolvein{x'=\theta};\beta'}{\mkint{\eta'}}$.

\oldcase{$\mknom{\pevolvein{x'=\theta}{H};\alpha}{\sinv{pre}{\sproof_1}{ind}{\sproof_2}{\sproof}}{\mkfnom{\pevolvein{x'=\theta}{H};\alpha}{\sproof}}$}
By IH, $\exists~\beta'~\mknom{\sproof}{\pevolvein{x'=\theta}{H};\beta}{\pevolvein{x'=\theta}{H};\beta'},$
 and $\mknom{\alpha^*;\beta}{\finally{\sproof}}{\pevolvein{x'=\theta}{H};\beta'}$
 for the same $\beta'$ and by IH,  $(\omega,\nu)\in\interp{\pevolvein{x'=\theta}{H};\beta'}{\mkint{\eta'}}$.

 \oldcase{$\mknom{\pevolvein{x'=\theta}{H};\alpha}{\finally{\sproof}}{\pevolvein{x'=\theta}{H};\mkfnom{\alpha}{\sproof}}$}
By inversion, $\exists \mu~(\omega,\mu)\in\interp{\pevolvein{x'=\theta}{Q}}$ and $(\mu,\nu)\in\interp{\beta}$.
By IH, $\exists~\beta'~\mknom{\sproof}{\beta}{\beta'}$ and $(\mu,\nu)\in\interp{\beta'}{\mkint{\eta'}}$.
Then since $(\omega,\mu)\in\interp{\pevolvein{x'=\theta}{Q}}$ then also $(\omega,\mu)\in\interp{\pevolvein{x'=\theta}{Q}}{\mkint{\eta'}}$ and thus $(\omega,\nu)\in\interp{\pevolvein{x'=\theta}{Q};\beta'}{\eta'}$ as desired.

Note that, as with the traces themselves, nondeterministic choices and loops discard nominals from their subprograms.
As before, this does not weaken the theorem, it simply means those nominals are local in scope.
\end{proof}

\begin{theorem} Nominalization is sound when applied to proofs that check. 
That is, for all $\tpchk{\G}{H}{\sproof}{[\alpha]\phi}{H_\phi}$ and all \dLN interpretations $\mkint{\eta},$ we have $\G\vdash_{\mkint{\eta}}[\mkfnom{\alpha}{\sproof}]\phi$.
%That is, for all $\tpchk{\G}{H}{\sproof}{[\alpha]\phi}{H'}$ then $\mknom{\alpha}{\sproof}{\alpha'}$ for some $\alpha'$.
%Furthermore, $\G\vdash_{\eta}[\alpha']\phi$ for all $\eta$ in $\dL$.
\end{theorem}
\begin{proof}
  The  claim holds by observing that the transitions for $\alpha_1$ (which holds easily by induction) are a subset of those for $\alpha$, regardless of $\eta$, and the result holds from soundness for Kaisar:
   $[\alpha]\phi$ implies $[\beta]\phi$ for any $\beta$ where $\interp{\beta}{} \subseteq \interp{\alpha}{}$ by the semantics of the box modality.
\end{proof}

%Together these theorems complete the claim that denominalization correctly reduces program-level nominals to sequent-level terms.
%However, for maximum confidence in the utility of this theorem, we should also validate the very definition of sequent-level state itself:

%TODO: Actually do and then typeset proof
%TODO: Needs an execution of $\alpha$ to exist I think
%\textbf{Theorem:} If $\seqstate{\eta}{H} \in \interp{\G}$ and $\corr{\eta}{H}$ and $\spchk{\G;H}{\sproof}{[\alpha]\phi}{H'}$, then at every subderivation $\spchk{\G^*;H^*}{\sproof^*}{[\alpha^*]\phi^*}{H''}$, $\seqstate{\eta^*}{H^*} \in \interp{\G^*}$ for some $\corr{\eta^*}{H^*}$.
%\begin{proof}
%  TODO
%\end{proof}
%At the same time, we also show that the soundness theorem holds \emph{regardless} of the well-formedness of static traces.
%This shows that not only is our calculus sound, but that provability and nominals are \emph{independent}.

Lastly we show standard soundness and completeness theorems:
\subsection{Soundness and Completeness}
\begin{lemma}[Soundness of Pattern Matching]
If $\mmatch{\G_1}{p}{e}{\G_2}$ then $\interp{\G_1}=\interp{\G_2}$.
That is, pattern-matching never affects the assumptions of a context, only the abbreviations, which are not soundness-critical.
\end{lemma}
\begin{proof}
By induction on $\emmatch{\G_1}{p}{e}$.
In every case, $\G_1 = \G_2$ or $\G_1$ adds a definition to $\G_2$ or $\G_2$ comes from an inductive call.
\end{proof}
\begin{lemma}
\label{lem:ghost}
If $\omega\in\interp{\G}$ and $(\omega,\nu)\in\interp{\alpha}$ then $\arb{\nu}{\alpha}\in\interp{\arb{\G}{\alpha}}$, where $\arb{\nu}{\alpha}$ ghosts all $x\in\boundvars{\alpha}$ in $\omega$.
%is the ghosting in $\omega$ of all $x\in\boundvars{\alpha}$.
\end{lemma}
\begin{proof}
  %TODO: First claim
  By renaming, $\arb{\omega}{\alpha}\in\interp{\arb{\G}{\alpha}}$.
  By program coincidence~\cite{DBLP:journals/jar/Platzer16}, there exists $\tilde{\nu}$ where $(\arb{\omega}{\alpha},\tilde{\nu})\in\interp{\alpha}$ and $\tilde{\nu}$ agrees with $\arb{\omega}{\alpha}$ on $\dom{\omega}$.
  Then $\left(\dom{\tilde{\omega}}\setminus\dom{\omega}\right)\cap\boundvars{\alpha}=\emptyset$, so
  by bound effect lemma~\cite{DBLP:journals/jar/Platzer16} $\tilde{\nu}=\arb{\nu}{\alpha}$.
  Since $\arb{\omega}{\alpha}$ and $\arb{\nu}{\alpha}$ agree on ghosts, by formula coincidence $\arb{\omega}{\alpha}\in\interp{\arb{\G}{\alpha}}$.
\end{proof}

\begin{lemma}[Soundness of Forward Proof]
\label{lem:fp-sound}
 If $\fpchk{\G;H}{\fproof}{\Delta}$ then $\G\vdash\Delta$ is valid formula of \dL.
\end{lemma}
\begin{proof}
  \oldcase{$\pat$} If $\phi\in\G$ then the conclusion holds by the hypothesis rule, otherwise $\phi\in\Sigma$ and since $\Sigma$ consists only of the axiom schemata of first-order arithmetic, which are sound, the result also holds.

\oldcase{$\fproof_1\ \fproof_2$} Assume $\omega\in\interp{\Gamma}$ then by IH both $\omega\in\interp{\phi\limply\psi}$ and $\omega\in\interp{\phi}$ so by modus ponens $\omega\in\interp{\psi}$ as desired.

\oldcase{$\fproof_1\ \theta$} 
Assume $\omega\in\interp{\G}$ and then by IH
 $\omega\in\interp{\forall x~\phi}$ and
 by forall instantiation/substitution $\omega\in\interp{\subst[\phi]{x}{\theta}}$.
\end{proof}

\begin{theorem}[Soundness of SP]
If $\tpchk{\G}{H}{\sproof}{\Delta}{H''}$ then $\G\vdash\Delta$ is a valid formula of \dL.
%TODO: Why
Recall that while the LCF architecture ensures the implementation of Kaisar is sound, this soundness theorem is still essential to validate our calculus.
\end{theorem}
\begin{proof} %(Sketch, full proof TODO)
(Note: This proof has a lot of cases. Even for the extended version of the proof, we leave out the diamond rules, which are analogous to the box rules, and the implict rules, which follow directly from soundness of propositional logic and a handful of \dL axioms).
  By induction on the derivation.

\oldcase{$\bassert{x}{\phi}{\sproof}$} 
By definition of $[?(\psi)]\phi$ suffices to show $\G\vdash(\psi\limply\phi),\Delta$ is valid, i.e.
$\G,\psi\vdash\phi,\Delta$.
This holds directly by IH.

\oldcase{$\have{x}{\psi}{\sproof_1}{\sproof_2}$}
Assume $\omega\in\interp{\G}$.
By IH1, $\omega\in\interp{\psi,\Delta}$.
If $\omega\in\interp{\Delta}$ then we're done.
Else $\omega\in\interp{\psi}$ so $\omega\in\interp{\G,\psi}$.
Then IH2 is applicable and $\omega\in\interp{\Delta}$ as desired.
 
\oldcase{$\slet{\pat}{\tilde{e}}{\sproof}$} Note the context $\G'$ produced by pattern matching contains only definitions, so for all $\omega\in\interp{\G}$ also have $\omega\in\interp{\G,\G'}$ and the IH applies directly, giving $\omega\in\interp{\Delta}$ as desired.

\oldcase{$\slet{t(?x)}{\tilde{e}}{\sproof}$} Symmetric.

\oldcase{$\snote{x}{\fproof}{\sproof}$}
Assume $\omega\in\interp{\G}$ then by Lemma~\ref{lem:fp-sound}, $\omega\in\interp{\psi}$ and thus $\omega\in\interp{\psi\land\G}$ and by IH, $\omega\in\interp{\Delta}$ as desired.

\oldcase{$\sshow{x}{\phi}{\uproof}$} 
  First note all facts in the {\tt using} clause hold.
  Those on the LHS of the union are in $\G$ and thus hold by hypothesis rule.
  Those on the RHS are the results of forward proofs and thus hold by soundness of forward proof.
  Thus $\omega\in\interp{\G'}$.
  Proceed by cases on the proof method.

\oldcase{{\tt id}} Holds by hypothesis rule.

\oldcase{$\mathbb{R}$} Holds by the side condition that $\G\limply\Delta$ is valid in first-order logic over the reals and the fact that \dL conservatively extends FOL.

\oldcase{{\tt auto}} We do not give a precise proof rule for this case since {\tt auto} is ever-changing heuristic automation, we merely note that it is the result of operations in an LCF-style which itself has been verified, and thus soundness follows as a result.
        
 \oldcase{$\tpchk{\G}{H}{\bcase{p_\phi}{\sproof_\phi}{p_\alpha}{\sproof_\alpha}}{[\alpha^*]\phi,\Delta}{H}$}
 By the IH, $\G_\phi\vdash\phi,\Delta$ and $\G_\psi\vdash[\alpha][\alpha^*]\phi,\Delta$ where $\match{p_\phi}{\phi}{\G_\phi}$ $\match{p_\alpha}{\alpha}{\G_\alpha}$. By Lemma~\ref{lem:pat-mon}, $\G$ and $\G_\psi$ contain the same facts, as do $\G$ and$\G_\alpha$, so  $\G\vdash\phi,\Delta$ and $\G\vdash[\alpha][\alpha^*]\phi$.
 Then by the semantics of loops (or soundness of the $[*]$ axiom of \dL), $\G\vdash[\alpha^*]\phi,\Delta$. 

\oldcase{$\sstate{t}{\sproof}$}
Directly by the IH on  $\tpchk{\G}{\hhtime{H}{t}}{\sproof}{\phi,\Delta}{H_\phi}$.

\oldcase{$\sfocus{\pat}{\sproof}$ on the left}
By the IH, $\G_1,\G_2\vdash\neg\phi,\Delta$.
The inverse of  $\neg{R}$ is sound (derivable from cut and double negation elimination and $\neg{L}$),
so  $\G_1,\phi\G_2\vdash\Delta$ as desired.

\oldcase{$\sfocus{\pat}{\sproof}$ on the right}
By the IH, $\G\vdash\phi,\Delta_1,\Delta_2$ and
$\G\vdash\Delta_1,\phi,\Delta_2$ follows by soundness of exchange.

\oldcase{$\bcase{\pat_\alpha}{\sproof_\alpha}{\pat_\beta}{\sproof_\beta}$}
Assume $\omega\in\interp{\G}$.
Note $\G_1$ and $\G_2$ contain only pattern-matching definitions.
So we can apply both IHs and get $\omega\in\interp{[\alpha]\phi\lor\Delta}$ and $\omega\in\interp{[\beta]\phi\lor\Delta}$.
In the case that $\omega\in\interp{\Delta}$ then we're done, else $\omega\in\interp{[\alpha]\phi}$ and $\omega\in\interp{[\beta]\phi}$ and by the semantics of $\cup$ then $\omega\in\interp{[\alpha\cup\beta]\phi}$ as desired.

\oldcase{$\bassign{x}{\pat_\theta}{\sproof}, x\notin\boundvars{\phi}$}
Assume $\omega\in\interp{\G}$. 
By IH have $\omega\in\interp{\subst[\phi]{x}{\theta},\Delta}$.
If $\omega\in\interp{\Delta}$ we're done, else by substitution  $\subst[\omega]{x}{\interp{\theta}{\omega}}\in\interp{\phi}$ and by definition  $\omega\in\interp{[x:=\theta]\phi}$ as desired.

\oldcase{$\bassign{x}{\pat_\theta}{\sproof}, x\in\boundvars{\phi}$}
Assume $\omega\in\interp{\G}$. 
Define $\omega_1=\subst[\left(\omega,x_i=\omega(x)\right)]{x}{\interp{\omega}{\theta}}$ for fresh $x_i$.
Observe by construction and freshness of i.e. and soundness of renaming that $\omega_1\in\interp{\subst[\G]{x_i}{x}}$ and $\omega_1(x)=\interps{\subst[\theta]{x}{x_i}}{\omega_1}$ so the IH applies, giving
$\omega_1\in\interp{\phi,\subst[\Delta]{x_i}{x}}$ which again by renaming and construction of $\omega_1$ gives either $\subst[\omega]{x}{\interps{\theta}{\omega}}\in\interp{\phi}$ or $\omega\in\interp{\Delta}$ and in either case the result holds.

\oldcase{$\bassignany{x}{\sproof}$}
Assume $\omega\in\interp{\G}$. 
Define $\omega_1=\subst[\left(\omega,x_i=\omega(x)\right)]{x}{r}$ for fresh $x_i$ and some $r\in\mathbb{R}.$
Observe by construction and freshness of i.e. and soundness of renaming that $\omega_1\in\interp{\subst[\G]{x_i}{x}}$ so the IH applies, giving
$\omega_1\in\interp{\phi,\subst[\Delta]{x_i}{x}}$ which again by renaming and construction of $\omega_1$ gives either $\subst[\omega]{x}{r}\in\interp{\phi}$ or $\omega\in\interp{\Delta}$ and in either case the result holds (since this is true for all such $r$).

% label{lem:ghost}
\oldcase{$\tpchk{\G}{H}{\bcon{\psi}{\sproof_1}{\sproof_2}}{[\alpha]\phi,\Delta}{H_\alpha}$}
By IH, $\G\vdash[\alpha]\eag{\psi},\Delta$ and $\arb{\G}{\alpha},\eag{\psi}\vdash\phi,\arb{\Delta}{\alpha}$.
Assume $(\omega,\nu)\in\interp{\G}\setminus\interp{\Delta}$ (else the sequent is trivially true).
Then by IH1 $\nu\in\interp{\eag\psi}$.
Also by Lemma~\ref{lem:ghost} $\nu\in{\interp{\arb{\G}{\alpha}}}\setminus\interp{\arg{\Delta}}{\alpha}$ and so by IH2 $\nu\in\interp{\phi}$.
Since this holds for all such $\nu$ then $\omega\in\interp{\G\vdash[\alpha]\phi,\Delta}$ as desired.

\oldcase{$\bsolve{\pat_{ode}}{\pat_{t}}{\pat_{dom}}{\sproof}$}
Assume $\omega\in\interp{\G}$, then let $\nu=\subst[\omega]{x}{\varphi(t)}$ for any $t\geq 0$ where $\forall~s\in[0,t]~\subst[\omega]{x}{\varphi(s)}\in\interp{Q}$ and $\varphi$ is the solution for $x'=\theta$.
Then let $t$ be a fresh variable which agrees everywhere with the above time $t$.
Note that $\omega\in\interp{t\geq{0}}$ and $\omega\in\interp{\forall~s\in[0,t]~Q}$ by choice of $t$.
Then the IH applies and $\omega\in\interp{\subst[\phi]{x}{y(t)},\Delta}$.
If $\omega\in\interp{\Delta}$ we're done, else note by choice of $\nu$ we have $\nu\in\interp{\phi}$ by substitution  and this holds for all such $\nu$ and so $\omega\in\interp{[\pevolvein{x'=\theta}{Q}]\phi}$ as desired.

%TODO Stuff about mid\oldcase{$$}
The cases for invariants are proven by simultaneous induction, where for the additional context of invariants $\delta\in\Delta$ we assume $[\alpha]\delta$ instead of $\delta$.

\oldcase{$\sinv{x}{J}{\sproof_{pre}}{sproof_{inv}}{\iproof}$, loop $\alpha^*$}
Assume $\omega\in\interp{\Gamma}$.
By IH1, $\omega\in\interp{\psi,\Delta}$. 
If $\omega\in\interp{\Delta}$ we're done, else $\omega\in\interp{\psi}$.
%Let $\G_1=\subst[\G]{x_i}{x_{in}}\cdots$ for all $x_i\in\boundvars{\alpha}$ and fresh ghost variables $x_{in}$.
%Let $\omega_1=\subst[\omega]{x_{in}}{x}\cdots$ likewise.
By Lemma~\ref{lem:ghost} $\arb{\omega}{\alpha}\in\interp{\arb{\G}{\alpha}}$ and
$(\arb{\omega}{\alpha},\arb{\omega}{\alpha})\in\interp{\alpha^*}$
and $\arb{\omega}{\alpha}\in\interp{\arb{\G}{\alpha}}$.
%Then note $\freevars{\G'}\cap\boundvars{\alpha^*}=\emptyset$.
%By coincidence lemma~\cite{DBLP:conf/cade/Platzer15}, $\omega_2\in\interp{\G_1}$ 
Thus by IH2 $\arb{\omega}{\alpha}\in\interp{[\alpha]\phi,\arb{\Delta}{\alpha}}$.
If $\arb{\omega}{\alpha}\in\interp{\arb{\Delta}{\alpha}}$ then by coincedence and renaming $\omega\in\interp{\Delta}$ and we're done.
%TODO: Mention \phi in the proper place
Else $\arb{\omega}{\alpha}\in\interp{[\alpha]\phi}$ and this holds for all such $\arb{\omega}{\alpha}$ so $\arb{\omega}{\alpha}\in\interp{[\alpha^*](\phi\limply[\alpha]\phi)}$ and by renaming so is $\omega$, and by the induction axiom $\omega\in{\interp{[\alpha^*]\phi}}$.

%TODO: Carefully draw out what renaming result I'm using when and how.
\oldcase{$\finally{\sproof}$, loop $\alpha^*$}
Let $\omega\in\interp{\G}$ and let $\nu$ such that $(\omega,\nu)\in\interp{\alpha^*}$.
%Let $\G_1=\subst[\G]{x_i}{x_{in}}\cdots$ for all $x_i\in\boundvars{\alpha}$ and fresh ghost variables $x_{in}$.
%Let $\omega_1=\subst[\omega]{x_{in}}{x}\cdots$ and $\nu_1=\subst[\nu]{x_{in}}{x}$ likewise.
Then $\arb{\omega}{\alpha}\in\interp{\arb{\G}{\alpha}}$ 
and $\arb{\nu}{\alpha}\in\interp{\arb{\G}{\alpha}}$ by Lemma~\ref{lem:ghost}.
%by renaming and $\arb{\nu}{\alpha}\in\interp{\Delta}$ since $x_{in}$ are all fresh and $\nu_2\in\interp{\G_1}$ by coincidence~\cite{DBLP:conf/cade/Platzer15}.
%Let $\Delta_1=\subst[\Delta']{x_{in}}{x}\cdots$
By IH $\arb{\nu}{\alpha}\in\interp{\phi,\arb{\Delta}{\alpha}}$.
In the latter case by renaming we're done.
Else since this holds for all possible $\nu$ we have $\omega\in\interp{[\alpha^*]\phi}$ as desired (by renaming again also).

\oldcase{$\sinv{x}{J}{\sproof_{pre}}{\sproof_{inv}}{\iproof}$, ODE $\alpha=\pevolvein{x'=\theta}{Q}$}
Assume $\omega\in\interp{\G\land[\alpha]\Delta}$ then by IH $\omega\in\interp{\psi,\Delta}$.
If $\omega\in\interp{\Delta}$ we're done else $\omega\in\interp{\psi}$.
Then let $\nu=\subst[\omega]{x}{\varphi(t)}$ for any $t\geq 0$ where $\forall s~\in[0,t]\subst[\omega]{x}{\varphi(s)}\in\interp{Q}$ and $\varphi$ is the solution to $x'=\theta$ for time $t$.
By renaming and unpacking semantics in $\Delta$ then $\nu\in\interp{\G_1}{\Delta}$ and so by IH $\omega\in\interp{\subst[(\phi)']{x'}{\theta},\Delta_1}$.
In the case $\omega\in\interp{\Delta_1}$ then by renaming $\omega\in\interp{\Delta}$ and we're done, else by soundness of differential effects, $\nu\in\interp{(\phi)'}$.
Since this holds for any such $\nu$, then by soundness of differential invariants, $\nu\in\interp{\psi}$ for all such $\nu$ and so $\omega\in\interp{[\alpha]\psi}$ so we can apply the second IH from which the result follows.

\oldcase{$\finally{\sproof}$, ODE $\pevolvein{x'=\theta}{Q}$}
Assume $\omega\in\interp{\G}$.
Then let $\nu=\subst[\omega]{x}{\varphi(t)}$ for any $t\geq 0$ where $\forall s~\in[0,t]\subst[\omega]{x}{\varphi(s)}\in\interp{Q}$ and $\varphi$ is the solution to $x'=\theta$ for time $t$.
%TODO: Fix the inference rule because it's wrong wrt \G in IH
By renaming $\nu\in\interp{\G_1}$ and by diff weakening $\nu\in\interp{Q}$ and by unpacking semantics $\nu\in\interp{\Delta}$ so the IH applies giving $\nu\in\interp{\phi}$.
Since this held for all $\nu$ then by renaming again $\omega\in\interp{[\alpha]\phi}$ as desired.

\oldcase{$\sghost{y}{\theta_2}{\theta_3}{\iproof}$}
Let $\omega_1=\omega,y=\interp{\theta_3}{\omega}$.
Then by coincidence~\cite{DBLP:conf/cade/Platzer15} $\omega_1\in\interp{\G,\Delta}$ and by construction $\omega_I\in\interp{y=\theta_3}$
So by IH, $\omega_1\in\interp{\pevolvein{x'=\theta_1,y'=\theta_2}{Q}}$
By linearity of $\theta_2$ the existence interval of $\pevolvein{x'=\theta_1,y'=\theta_2}{Q}$ agrees with $\pevolvein{x'=\theta_1}{Q}$ so they have transitions for all the same times $t$.
And since $y$ was fresh they agree on all variables but $y$ and by coincidence~\cite{DBLP:conf/cade/Platzer15} $\nu\in\interp{\phi}$ whenever $\nu_1\in\interp{\phi}$ and so by renaming $\omega\in\interp{[\pevolvein{x'=\theta_1}{Q}]\phi}$ as desired.

%\oldcase{$$}
%\oldcase{$$}
%\oldcase{$$}

%  The structured symbolic execution cases each follow by the corresponding sequent calculus rules of \dL.
%  For the cases that modify the history, we additionally apply soundness of discrete ghosts, which follows from soundness of uniform and bound renaming.
%  The invariant case for loops uses axiom V to maintain truth of assumptions on old states.
%  The invariant case for ODEs additionally uses axioms dE,dW, and rule ['] to reduce ODEs to real arithmetic, and axiom dC to compose multiple differential invariants.
%  The structural rule {\tt have} is sound by propositional cut, {\tt show} is sound by the soundness of the underlying LCF prover core used for unstructured proofs, and {\tt let} and {\tt state} are sound because the do not modify the proof goal.
%  The {\tt note} rule is sound by the soundness of the modus ponens and forall instantiation rules.
\end{proof}

\newcommand{\scent}{\vdash_{\rm SC}}
\newcommand{\kent}{\vdash_{\rm K}}
\newcommand{\nnf}[1]{{\tt nnf}(#1)}
\newcommand{\fseq}[1]{\G\vdash{#1},\Delta}

\newcommand{\aseq}[1]{\G,{#1}\vdash\Delta}

\begin{theorem}[Completeness]  Kaisar is complete with respect to sequent calculus for \dL, that is for all $\Gamma,\Delta$ if $\Gamma\vdash\Delta$ is provable in \dL sequent calculus then  $\tpchk{\G}{H}{\sproof}{\Delta}{H_\sproof}$ for some $H,H_\sproof,\sproof$.
\end{theorem}
(Note: This proof has a lot of cases. Even for the extended version of the proof, we present only the cases for right rules for boxes and left rules for diamonds --- the other cases are analogous).
%(TODO: This requires adding a contextual equivalence construct) Kaisar is complete with respect to the sequent clculus, and thus relatively complete for both continuous and discrete dynamics.
  The proof proceeds by induction on sequent calculus proofs of \dL, but we first establish several preliminary observations resolving the key differences between sequent calculus and Kaisar.

\textbf{Observation 1:} Rules such as $\forall$ instantiation require specifying an input.
Inputs in Kaisar are extended terms, thus completeness requires all terms are expressible as extended terms (i.e. expansion is surjective).
This is true when all states are named, in which case we call the trace \emph{complete}, a property we maintain as an invariant.
\begin{proof}
  By induction on the term $\theta$ we wish to express.
  All cases except variables hold by IH.
  In the (ghost or program) variable case, produce a nominal for the state at which the name was introduced.
\end{proof}

\textbf{Observation 2:} All Kaisar proofs produced in this proof support weakening, so if $\tpchk{\G_1}{H}{\sproof}{\Delta_1}{H_\sproof}$ then  $\tpchk{\G_1,\G_2}{H}{\sproof}{\Delta_1,\Delta_2}{H_\sproof}$ for any contexts $\Gamma_2,\Delta_2$ that introduce only fresh variables, i.e. contexts where $\dom{\G_1}\cap\dom{\G_2}=\dom{\Delta_1}\cap\dom{\Delta_2}=\emptyset$.
\begin{proof}
  By inspection on the cases of the completeness proof.
  Weakening could only fail if adding assumptions caused a pattern-match to be ambiguous, but we fully disambiguate all patterns.
\end{proof}

\textbf{Observation 3:} Kaisar is complete even ignoring forward proof, abbreviations and {\tt using} clauses.
While these features are essential in practical usage, they will not be needed in this proof.

An essential part of the proof is our choice of induction metric, because we often induct on formulas that are not strict subformulas of the input.
We define a total ordering on sequents $(\G_1\vdash\Delta_1)\prec(\G_2\vdash\Delta_2)$ as the lexicographic ordering of:
\begin{enumerate}
\item The number of formulas containing modalities in the antecedent.
\item The number of nondeterministic choices $\alpha \cup \beta$.
\item The number of loops $\alpha^*$
\item The number of compositions $\alpha;\beta$
\item The total number of modalities $\dbox{\alpha}{\phi}$ or $\ddiamond{\alpha}{phi}$.
\item The number of symbols occuring under some negation (i.e. sum of the number of symbols in $\phi$ for each $\neg(\phi)$)
\item The total number of formula connectives, (which does not include the comma separator, nor term connectives.)
%TODO: Does this work? Don't think it's even necessary
%\item The structure of the sequent calculus derivation
\end{enumerate}
Rule (1) ensures the {\tt focus} rule can be used to implement antecedent structured execution.
Rules (2-4) support the cases for the compound programs and (5) supports the atomic programs.
Rule (6) supports the negation-normal form implicit rules, e.g. this definition ensures $(\G\vdash\neg{P}\lor\neg{P},\Delta)\prec(\G\vdash\neg(P\land{Q}),\Delta)$
Rule (7) supports the propositional and quantifier cases.
%Rule (8) supports the case where a modality occurs only in a succedent position other than the first.

%TODO: I think we need a rule to focus onto the antecedent
  The main proof now proceeds by induction on sequent calculus proofs of \dL under $\prec$.
\begin{proof}
In each case, let $\phi$ refer to the formula acted upon by the SC rule.
The structured execution rules assume a formula in the first succedent position, thus (at the beginning of each case) we immediately pull $\phi$ to first succedent position with {\tt focus}.
Well-ordering is preserved here: For right rules, the metric $\prec$ is unchanged when acting on the succedent.
For left rules, it strictly decreases by Rule (1).
%In each case we look at the first succedent formula, except in the case where that formula is in $\folr$ and some other formula is not, in which case we apply $\focus$ to bring some modal formula to the first succedent position.
%By Rule 1,8) this is well-founded.

\oldcase{cut} Can be assumed not to happen by admissibility of discrete cut elimination, but even without that, reduces to {\tt have}.

\oldcase{$\phi\in\folr$} By decidability of $\folr$, $\tpchk{\G}{H}{\sshow{x}{\_}{\tt auto}}{\phi,\Delta}{H}$.

\oldcase{$\phi\in\G$} By hypothesis rule, $\tpchk{\G,\phi}{H}{\sshow{x}{\_}{\tt id}}{\phi,\Delta}{H}$.

\oldcase{$\phi=\phi_1\land\phi_2$} 
By invertibility of the $\land{R}$ sequent calculus rule, $\G\vdash{\phi_1},\Delta$ and $\G\vdash{\phi_2},\Delta$ are derivable.
By IH, $\exists \sproof_1,\sproof_2$ where $\tpchk{\G}{H}{\sproof_1}{\phi_1}{H}$ and $\tpchk{\G}{H}{\sproof_2}{\phi_2}{H}$, from which we apply the case rule for conjunctions on the right: $\tpchk{\G}{H}{\bcase{\phi_1}{\sproof1}{\phi_2}{\sproof_2}}{\phi_1\land\phi_2,\Delta}{H}$.

\oldcase{$\phi=\phi_1\lor\phi_2$} Let $\Delta_1=\phi_1,\phi_2,\Delta$. 
Notice $\Delta_1\prec{\phi_1\lor\phi_2,\Delta}$ 
so apply the IH yielding $\tpchk{\G}{H}{\sproof}{\phi_1,\phi_2,\Delta}{H'}$ 
then apply the case rule for disjunction on the right and we get: 
$\tpchk{\G}{H}{\pcaseol{\phi_1}{\phi_2}{\sproof}}{\phi_1\lor\phi_2,\Delta}{H'}$.

\oldcase{$\phi=\neg\phi_1$} 
By rule (6) we apply whichever negation-normal form (NNF) implicit rule applies based on the shape of $\phi$.
%We choose $\prec$ such that the negation-normal form (NNF) of a formula $\phi_1$ is no more complex than $\phi_1$.
Because SC is sufficient to implement NNF normalization, if $\G\vdash\phi,\Delta$ is true in SC then $\G\vdash\nnf{\phi},\Delta$ is as well, and is simpler, so true in Kaisar as well, then by the implicit normalization rules of Kaisar, $\fseq{\phi}$ is as well.

\oldcase{$\phi=\phi_1\limply\phi_2$}
By IH (applicable because we removed propositional connectives), $\exists \sproof~ \tpchk{\G,\phi_1}{H}{\sproof}{\phi_2,\Delta}{H'}$,
Then $\tpchk{\G}{H}{\bassert{x}{\phi_1}{\sproof}}{\phi_1\limply\phi_2,\Delta}{H'}$

\oldcase{$\phi=\phi_1\bimply\phi_2$}
By IH (applicable because we removed propositional connectives), both $\tpchk{\G,\phi_1}{H}{\sproof_1}{\phi_2,\Delta}{H_1}$ and $\tpchk{\G,\phi_2}{H}{\sproof_2}{\phi_1,\Delta}{H_2}$
and then we apply the casing form for equivalences:
$\tpchk{\G}{H}{\bcase{\phi_1\limply\phi_2}{\bassert{x_1}{\phi_1}{\sproof_1}}{\phi_2\limply\phi_1}{\bassert{x_2}{\phi_2}{\sproof_2}}}{\phi_1\bimply\phi_2,\Delta}{H}$.

\oldcase{$\phi=\forall~x~\phi_1$} Symmetric with the case for $[\prandom{x}]\phi_1$, except the reason for well-foundedness is that we removed a quantifier and added only term connectives.

\oldcase{$\phi=\exists~x~\phi_1$} Symmetric with the case for $\ddia{\prandom{x}}{\phi_1}$.

\oldcase{$\phi=[x:=\theta]\phi_1$ where $x\notin\boundvars{\phi_1}$ }
By IH (because we remove a modality), we then have $\exists \sproof~\tpchk{\G}{\hhsub{H}{x}{\theta}}{\sproof}{\subst[\phi_1]{x}{\theta},\Delta}{H'}$.
Then $\tpchk{\G}{H}{\sstate{t}{\bassign{x}{\theta}{\sproof}}}{[x:=\theta]\phi_1,\Delta}{H'}$ and $\hhsub{\hhtime{H}{t}}{x}{\theta}$ is a complete history and the inner proof still holds by weakening (i.e. addition of state namesm).

\oldcase{$\phi=[x:=\theta]\phi_1$ where $x\notin\freevars{\G,\Delta,\theta}$}
By IH (because we remove a modality), we then have $\exists \sproof~\tpchk{\G}{\hheq{H}{x}{x_i}{\theta}}{\sproof}{\phi_1,\Delta}{H'}$.
Then $\tpchk{\G}{H}{\sstate{t}{\bassign{x}{\theta}{\sproof}}}{[x:=\theta]\phi_1,\Delta}{H'}$ and $\hheq{\hhtime{H}{t}}{x}{\theta}$ is a complete history and the inner proof still holds by weakening (i.e. addition of state names).

\oldcase{$\phi=[\prandom{x}]\phi_1$ where $x\notin\freevars{\G,\Delta,}$}
By IH (because we remove a modality), we then have $\exists \sproof~\tpchk{\G}{\hhany{H}{x}{x_i}}{\sproof}{\phi_1,\Delta}{H'}$.
Then $\tpchk{\G}{H}{\sstate{t}{\bassignany{x}{\sproof}}}{[\prandom{x}]\phi_1,\Delta}{H'}$ and $\hhany{\hhtime{H}{t}}{x}{x_i}$ is a complete history and the inner proof still holds by weakening (i.e. addition of state names).

\oldcase{$\phi=[\alpha;\beta]\phi_1$} 
By IH (because we remove a sequential composition), $\exists \sproof~\tpchk{\G}{H}{\sproof}{[\alpha][\beta]\phi_1,\Delta}{H'}$ then by the implicit rule for $[;]$ have $\tpchk{\G}{H}{\sproof}{[\alpha;\beta]\phi_1}{H'}$.

%TODO: [\alpha^*]\phi_1, etc
%By IH, $\exists \sproof~\spchk{\G;H}{\sproof}{\fseq{}}{H'}

\oldcase{$\phi=[\alpha\cup\beta]\phi_1$} 
By IH (because we remove a $\cup$ operator), $\exists \sproof_1~\tpchk{\G}{H}{\sproof_1}{[\alpha]\phi_1,\Delta}{H'}$ and $\exists \sproof_2~\tpchk{\G}{H}{\sproof_2}{[\beta]\phi_1,\Delta}{H'}$.
By  $[\cup]$ casing have $\tpchk{\G}{H}{\bcase{\alpha}{\sproof_1}{\beta}{\sproof_2}}{[\alpha\cup\beta]\phi_1,\Delta}{H}$.

\newcommand{\sder}[1]{\mathcal{D}_{#1}}

 \oldcase{$\infer{\aseq{\ddia{\alpha\cup\beta}{\phi}}}
               {\sder{1}=\aseq{\ddia{\alpha}{\phi}} & \sder{2}=\aseq{\ddia{\beta}{\phi}}}$}
 Observe in SC $\fseq{[\alpha\cup\beta]\neg\phi}$ by the following derivation:
\[\infer{\fseq{[\alpha\cup\beta]\neg\phi}}
   {\infer{\fseq{[\alpha]\neg\phi}} %andR boxAnd
   {\infer{\aseq{\neg[\alpha]\neg\phi}} %notL^-1
   {\infer{\ddia{\alpha}{\phi}}{\sder{1}}}}  %duality
&  {\infer{\fseq{[\beta]\neg\phi}}
   {\infer{\aseq{\neg[\alpha]\neg\phi}}
   {\infer{\ddia{\beta}{\phi}}{\sder{2}}}}}}\]
And then since this sequent is simpler (less modalities on the left), it's derivable by some $\sproof$ in Kaisar by IH, so we derive $\tpchk{\G}{H}{\sfocus{\ddia{\alpha\cup\beta}{\phi}}{\sproof}}{\ddia{\alpha\cup\beta}{\phi},\Delta}{H}$.

 \oldcase{$\infer{\aseq{\ddia{\alpha;\beta}{\phi}}}
         {\sder{1}=\aseq{\ddia{\alpha}{\ddia{\beta}{\phi}}}}$}
Observe in SC $\fseq{[\alpha][\beta]\neg\phi}$ by the following derivation:
\[\infer{\fseq{[\alpha][\beta]\neg\psi}}
{\infer{\fseq{[\alpha]\neg\neg[\beta]\neg\phi}} %DNE
{\infer{\fseq{\alpha}\neg\ddia{\beta}{\phi}} %dual
{\infer{\aseq{\neg[\alpha]\neg\ddia{\beta}{\phi}}} %!L^-1
{\infer{\aseq{\ddia{\alpha}{\ddia{\beta}{\phi}}}}
  {\sder{1}}}}}}\]
And then since this sequent is simpler (less modalities on the left), it's derivable by some $\sproof$ in Kaisar by IH, so we derive $\tpchk{\G}{H}{\sfocus{\ddia{\alpha;\beta}{\phi}}{\sproof}}{\ddia{\alpha;\beta}{\phi},\Delta}{H}$.

%TODO: Same nonsense for not bound case
\oldcase{$\infer{\aseq{\ddia{x:=\theta}{\phi}}}
     {\sder{1}=\aseq{\subst[\phi]{x}{\theta}}} (\text{if} x\notin\boundvars{\phi})$}
Observe $\fseq{[x:=\theta]\neg\phi}$ in SC by (noting $\neg\left(\subst[\phi]{x}{\theta}\right) = \subst[(\neg\phi)]{x}{\theta}$)
\[\infer{\fseq{[x:=\theta]\neg\phi}}
{\infer{\fseq{\subst[\left(\neg\phi\right)]{x}{\theta}}}
{\infer{\aseq{\subst[\phi]{x}{\theta}}}
{\sder{1}}}}\]
And then since this sequent is simpler (less modalities on the left), it's derivable by some $\sproof$ in Kaisar by IH, so we derive $\tpchk{\G}{H}{\sfocus{\ddia{x:=\theta}{\phi}}{\sproof}}{\ddia{x:=\theta}{\phi},\Delta}{H}$.

\oldcase{$\infer{\aseq{\ddia{x:=\theta}{\phi}}}
     {\sder{1}=\aseq{\subst[\phi]{x}{\theta}}} (\text{if} x\notin\freevars{\Gamma,\Delta})$}
Observe $\fseq{[x:=\theta]\neg\phi}$ in SC by (noting $\neg\left(\subst[\phi]{x}{\theta}\right) = \subst[(\neg\phi)]{x}{\theta}$)
\[\infer{\fseq{[x:=\theta]\neg\phi}}
{\infer{\subst[\G]{x}{x_i},x=\theta\vdash\left(\neg\phi\right),\subst[\Delta]{x}{x_i}}
{\infer{\subst[\G]{x}{x_i},x=\theta,\phi\vdash\subst[\Delta]{x}{x_i}}
{\infer{\G,\ddiamond{\humod{x}{\theta}}{\phi}\vdash\Delta}
{\sder{1}}}}}\]
And then since this sequent is simpler (less modalities on the left), it's derivable by some $\sproof$ in Kaisar by IH, so we derive $\tpchk{\G}{H}{\sfocus{\ddia{x:=\theta}{\phi}}{\sproof}}{\ddia{x:=\theta}{\phi},\Delta}{H}$.

\oldcase{$\infer{\aseq{\ddia{\prandom{x}}{\phi}}}
     {\sder{1}=\aseq{\exists~x~\phi}}$}
Observe $\fseq{[\prandom{x}]\neg\phi}$ in SC by
\[\infer{\fseq{[\prandom{x}]\neg\phi}}
{\infer{\fseq{\forall~x~\neg\phi}}%prandom
{\infer{\aseq{\neg\forall~x~\neg\phi}} %!L
{\infer{\aseq{\exists~x~\phi}} %duality
{\sder{1}}}}}
\]
And then since this sequent is simpler (less modalities on the left), it's derivable by some $\sproof$ in Kaisar by IH, so we derive $\tpchk{\G}{H}{\sfocus{\ddia{\prandom{x}}{\phi}}{\sproof}}{\ddia{\prandom{x}}{\phi},\Delta}{H}$.

\oldcase{$\infer{\aseq{\ddia{?(\psi)}{\phi}}}
 {\sder{1}=\aseq{\psi\land\phi}}$}
Observe $\fseq{[?(\psi)]\neg\phi}$ in SC by
\[\infer{\fseq{[?(\psi)]\neg\phi}}
{\infer{\fseq{\psi\limply\neg\phi}}
{\infer{\fseq{\neg\psi\lor\neg\phi}}
{\infer{\aseq{\neg\left(\neg\psi\lor\neg\phi\right)}}
{\infer{\aseq{\psi\land\phi}}
{\sder{1}}}}}}\]
And then since this sequent is simpler (less modalities on the left), it's derivable by some $\sproof$ in Kaisar by IH, so we derive $\tpchk{\G}{H}{\sfocus{\ddia{?(\psi)}{\phi}}{\sproof}}{\ddia{?(\psi)}{\phi},\Delta}{H}$.

\oldcase{$\infer{\aseq{\ddia{\alpha^*}{\phi}}}
{\aseq{\phi\lor\ddia{\alpha}{\ddia{\alpha^*}{\phi}}}}$}
Observe $\fseq{[\alpha^*]\neg\phi}$ in SC by
\[\infer{\fseq{[\alpha^*]\neg\phi}}
 {\infer{\fseq{\neg\phi\land[\alpha][\alpha^*]\neg\phi}}%[*]
 {\infer{\fseq{\neg\phi\land\neg\neg[\alpha^*]\neg\phi}}%DNE
 {\infer{\fseq{\neg\phi\and\neg\ddia{\alpha}{\ddia{\alpha^*}{\phi}}}}%DUAL
 {\infer{\aseq{\neg\left(\neg\phi\land\neg\ddia{\alpha}{\ddia{\alpha^*}{\phi}}\right)}}%
 {\infer{\aseq{\phi\lor\ddia{\alpha}{\ddia{\alpha^*}{\phi}}}}
 {\sder{1}}}}}}}\]
And then since this sequent is simpler (less modalities on the left), it's derivable by some $\sproof$ in Kaisar by IH, so we derive $\tpchk{\G}{H}{\sfocus{\ddia{\alpha^*}{\phi}}{\sproof}}{\ddia{\alpha^*}{\phi},\Delta}{H}$.

\oldcase{$\infer{\aseq{\ddia{\pevolvein{x'=\theta}{Q}}{\phi(x)}}}
{\G,t\geq0,\forall~s~\in[0,t]~\subst[Q]{x}{y(t)},\phi(y(t))\vdash\Delta}$}
Observe $\fseq{[\pevolvein{x'=\theta}{Q}]\neg\phi}$ in SC by
\[\infer{\fseq{[\pevolvein{x'=\theta}{Q}]\neg\phi}}
{\infer{\G,t\geq{0},\forall~s~\in[0,t]~\subst[Q]{x}{y(t)}\vdash p(y(t)),\Delta}
{\infer{\fseq{t\geq{0}\limply\forall~s~\in[0,t]~\subst[Q]{x}{y(t)}\limply{p(y(t))}}}
{\infer{\fseq{\neg{t\geq{0}}\lor{\neg{\forall~s~\in[0,t]~\subst[Q]{x}{y(t)}}}\lor{p(y(t))}}}
{\infer{\aseq{\neg\left(\neg{t\geq{0}}\lor{\neg{\forall~s~\in[0,t]~\subst[Q]{x}{y(t)}}}\lor{p(y(t))}\right)}}
{\infer{\aseq{{t\geq{0}}\land{\forall~s~\in[0,t]~\subst[Q]{x}{y(t)}}\land{p(y(t))}}}
{\infer{\Gamma,{t\geq{0}},{\forall~s~\in[0,t]~\subst[Q]{x}{y(t)}},{p(y(t))}\vdash\Delta}
{\sder{1}}}}}}}}\]
And then since this sequent is simpler (less modalities on the left), it's derivable by some $\sproof$ in Kaisar by IH, so we derive $\tpchk{\G}{H}{\sfocus{\ddia{\pevolvein{x'=\theta}{Q}}{\phi(x)}}{\sproof}}{\ddia{\pevolvein{x'=\theta}{Q}}{\phi(x)},\Delta}{H}$.

\oldcase{($\dbox{\alpha^*}{\phi}$, rule I)}
We have by assumption
\[\infer{\G,\phi\vdash[\alpha^*]\phi}
         {\arb{\G}{\alpha},\phi\vdash[\alpha]\phi,\arb{\Delta}{\alpha}}\]
And the proof follows trivially from Inv and finally.

\oldcase{$\dbox{M}{}$} In the case of a box monotonicity proof, note that the mid rule as presented in the paper is simply a special case of:
% label{lem:ghost}
\[\infer{\tpchk{\G}{H}{\bcon{\psi}{\sproof_1}{\sproof_2}}{[\alpha]\phi,\Delta}{H_\alpha}}
 {\tpchk{\G}{H}{\sproof_1}{[\alpha]\eag{\psi},\Delta}{H_\alpha} &
  \tpchk{\arb{\G}{\alpha},\eag{\psi}}{H_\alpha}{\sproof_2}{\phi,\arb{\Delta}{\alpha}}{H_\phi}}\]
We presented it in its special-case form due to its close analogy with Hoare logic for explanatory purposes, but structured symbolic execution works just as well with the general form, and with the general form $\dbox{M}{}$ falls out by propositional reasoning.
  
\oldcase{$\ddiamond{M}{}$} Symmetric.

\oldcase{${[*]}$ iter on the right}
Direct from the IH and from the rule:
\[\infer{\tpchk{\G}{H}{\bcase{p_\phi}{\sproof_\phi}{p_\alpha}{\sproof_\alpha}}{[\alpha^*]\phi,\Delta}{H}}
        {\deduce
          {\tpchk{\G_\phi}{H}{\sproof_\phi}{\phi,\Delta}{H_\phi}
           \hskip 0.1in
           \tpchk{\G_\alpha}{H}{\sproof_\alpha}{[\alpha][\alpha^*]\psi,\Delta}{H_\alpha}
          }
          {
            \match{p_\phi}{\phi}{\G_\phi}
            &\match{p_\alpha}{\alpha}{\G_\alpha}
          }
        }\]

\oldcase{$\ddiamond{\alpha^*}{\phi}$ on the right}
Rule
\[\infer{\G,\exists{v}~\varphi(v)\vdash\langle\alpha^*\rangle\exists{v\leq{0}}\varphi(v),\Delta}
{\arb{\G}{\alpha},v>0,\varphi(v)\vdash\langle\alpha\rangle\varphi(v-1)}\]
translates to
\[\infer
{\tpchk{\arb{\G}{\alpha},\eag{\psi}(x),x>0}{\arb{H}{\alpha}}{\sproof_{Inv}}{\ddiamond{\alpha}{\longeag{\psi(x-1)}},\arb{\Delta}{\alpha}}{H}}
{\tpchk{\Gamma}{H}{\sproof_{Pre}}{\exists~x~\eag{\phi}(x),\Delta}{H_{Pre}}}\]

%TODO: Make sure to do loop-invariant and variant-reasoning.

\oldcase{$\fseq{[\pevolvein{x'=\theta}{Q}]\phi}$}
% TODO: Make sure this assumption is valid
In the ODE case, we assume without loss of generality that the proof is not a proof by the solve axiom $[']$ since $[']$ is implementable with DG,DC, and DI~\cite{DBLP:conf/cade/Platzer15}.
Further note we can restrict our consideration to a simplied fragment of ODE proofs which is in turn complete for SC.
We say an ODE proof is in \emph{linear-normal} form if it consists of 0 or more DG's, followed by zero or more DCs (where the proof of each cut is a single DI), ending in DW.
\newcommand{\lnp}[4]{(#1,\deduce{#2}{#3},#4)}
\newcommand{\lnc}[2]{(#1{\rm::}#2)}
\newcommand{\lncc}[3]{(#1{\rm::}#2{\rm::}#3)}
\newcommand{\lnn}{\epsilon}
\newcommand{\lnnorm}[2]{{#1}\Mapsto{#2}}
\newcommand{\scd}{\mathcal{D}}
\newcommand{\fancydc}[3]{{\deduce{DC(#1)}{#2}}#3}
We notate linear-normal proofs as $\lnp{DGs}{DCs}{DIs}{DW}$ with empty lists denoted $\lnn$ and concatenation $\lnc{\scd_1}{\scd_2}$
\begin{lemma}Linear-normal ODE proofs in SC are complete for ODE proofs in SC.
\end{lemma}
\begin{proof}
  We define normalization $\lnnorm{\scd}{\lnp{DGs}{DCs}{DIs}{DW}}$, then show the proofs check:
\begin{center}
  \begin{tabular}{ccc}
  \infer[DW]{\lnnorm{DW}{\lnp{\lnn}{\lnn}{\lnn}{DW}}}{}
& \infer[DI]{\lnnorm{DI(\phi)}{\lnp{\lnn}{DC(\phi)}{DI(\phi)}{DW}}}{}   
& \infer[DG]{\lnnorm{DG(\scd)}{\lnp{\lnc{DG}{DGs}}{DCs}{DIs}{DW}}}
        {\lnnorm{\scd}{\lnp{DGs}{DCs}{DIs}{DW}}}\\
  \end{tabular}\\[0.1in]
$\infer[DC]{\lnnorm{\fancydc{\phi}{\scd_1}{\scd_2}}{\lnp{\lnc{DGs_1}{DGs_2}}{\lncc{DCs_1}{DCs_2}{DC(\phi)}}{\lncc{DIs_1}{DIs_2}{DW_1}}{DW_2}}}
  {\lnnorm{\scd{}_1}{\lnp{DGs_1}{DCs_1}{DIs_1}{DW_1}}  
  &\lnnorm{\scd{}_2}{\lnp{DCs_2}{DCs_2}{DIs_2}{DW_2}}}$
\end{center}
  Proceed by induction on the derivation.
  \oldcase{DW}: The proof is already linear-normal.\\
  \oldcase{DI}: We cut in $\phi$, the formula proved by $DI(\phi),$ thus the cut holds.
 The cut holds because $DI(\phi)$ is a proof of some 
  \oldcase{DG}: Direct by the IH.\\
  \oldcase{DC}: 
  $DGs_1$ check trivially by IH. $DGs_2$ check by IH and because the ghosted system after $DGs_1$ with the input system on non-ghost variables.
  $DCs_1$ and $DIs_1$ check because all invariants of the input system are invariants of the ghosted system.
  $DCs_2$ and $DIs_2$ check for this reason and because the addition of further invariants $DCs_1$ does not reduce provability.
  $DC(\phi)$ and $DW_1$ prove because addition of $DCs_1$ and $DCs_2$ never reduces provability.
  $DW_2$ proves because $\phi$ is available in the domain constraint and the addition of $DCs_1$ does not reduce provability.

\end{proof}

\begin{lemma} Kaisar is complete for linear-normal SC proofs.\end{lemma}
%We further assume without loss of generality that the proof consists of a chain of diffcuts where any dGs are first, then any dIs, then finally a single dW.
%Multiple dWs can be collapsed to one because any dW which would be provable earlier in an $\iproof$ is also provable at the end since the domain constraint only gets stronger.
%Further, if a dI were to use the fact proved by dW it could prove it directly itself by propositional cut.
%Furthermore any diffcut can be moved past dGs since any invariant of the original system is also an invariant of the ghosted system and the ghosts themselves have no dependencies.
\begin{proof}
  We proceed by a simultaneous induction on invariant proofs in normal
  form:

  % TODO: Do this in more detail
  \oldcase{$\infer{\fseq{[\pevolvein{x'=\theta}{Q}]\phi}} {Q\vdash\phi}$}
  By the outer IH $Q\vdash\phi$ is provable by some $\sproof$ because it eliminates a modality.
  The proof follows by applying the $\finally$ rule:
\[\linferenceRule[formula]
        {\tpchk{\arb{\G}{\alpha},\invlist{},Q}{\arb{H}{\alpha}}{\sproof}{\phi,\arb{\Delta}{\alpha}}{H_{x'}}}
        {\tpchk{\G,\invlist{}}{H}{\finally{\sproof}}{[\{x'=\theta~\&~Q\}]\phi,\Delta}{H_{x'}}}\]

  \oldcase{$\infer{\fseq{[\pevolvein{x'=\theta}{Q}]\phi}}
    {\infer{\fseq{[\pevolvein{x'=\theta}{Q}]\psi}}
      {Q\vdash{\subst[(\psi)']{x'}{\theta}}} &
      \fseq{[\pevolvein{x'=\theta}{Q\land\psi}]\phi}}$} 
  By the inner IH, the ``use'' case is provable by some $\iproof$.
  By the linear-normal assumption, the ``show'' case is a single DI, and thus follows from the Inv rule:
\[\linferenceRule[formula]
        {\deduce{\tpchk{\Gamma,\ext{\Delta}{x}{\eag{\psi}}}{H}{\iproof}{[\{x'=\theta~\&~Q\}]\phi,\Delta}{H_{x'}}}
{\tpchk{\Gamma,\invlist{}}{H}{\sproof_{Pre}}{\eag{\psi},\Delta}{H_{Pre}}
   \hskip 0.1in\tpchk{\arb{\G}{\alpha},\invlist{},Q}{\arb{H}{\alpha}}{\sproof_{Inv}}{[\humod{x'}{\theta}](\eag{\psi})',\arb{\Delta}{\alpha}}{H_{Inv}}
}}
        {\tpchk{\G,\Delta}{H}{\sinv{x}{\psi}{\sproof_{Pre}}{\sproof_{Inv}}{\iproof}}{[\{x'=\theta~\&~Q\}]\phi,\Delta}{H_{x'}}}\]
%        {\tpchk{\G,\Delta}{H}{\sinv{x}{\psi}{\sproof_{Pre}}{\sproof_{Inv}}{\iproof}}{[\{x'=\theta~\&~Q\}]\phi,\Delta}{H_{x'}}}\]
%\[\infer{\GDE\sinv{x}{\tilde{\psi}}{\sproof_1}{\sproof_2}{\iproof}:[\{x'=\theta~\&~H\}]\phi,\Delta'\Eval{H''}}
%   {\deduce{\Gamma;\ext{\Delta}{x}{\psi}\vdash\iproof:[\{x'=\theta~\&~H\}]\phi,\Delta'\Eval{H''}}
%{\deduce{\GHE\tilde\psi\eval\psi
%   \hskip 0.1in\Gamma,\Delta,H\vdash\sproof_1:\psi,\Delta' 
%   \hskip 0.1in[\sigma]\G,\Delta,H'\vdash\sproof_2:[x':=\theta](\psi)',[\sigma]\Delta'}
%{\deduce{\sigma\equiv[x_i/x] }
%{{H' = H,\hrany{x}{x_i}}}}
%}}\].

  \oldcase{$\infer{\fseq{[\pevolvein{x'=\theta}{Q}]\phi}}
    {\phi\bimply{\exists y~\psi} &
      \fseq{[\pevolvein{x'=\theta,y'=\theta_2}{Q}]\psi}}$} 
  By inner IH, the use case is provable by some $\iproof,$ and the result follows by ghosting:
\[\linferenceRule[formula]
        {\tpchk{\G,y=\eag{\theta_{y}}}{\Delta;H}{\iproof}{[\{x'=\eag{\theta_{x'}},y'=\eag{\theta_{y'}}~\&~H\}]\phi,\Delta}{H_{x'}} & \eag{\theta_{y'}}~{\rm linear\ in}~y}
        {\tpchk{\G,\invlist{}}{H}{\sghost{y}{\theta_{y'}}{\theta_{y}}{\iproof}}{[\{x'=\theta_{x'}~\&~H\}]\phi,\Delta}{H_{x'}}}\]
\end{proof}

By composing the above lemmas, Kaisar is complete for ODE proofs in sequent calculus.
  % By induction on derivations of \dL sequent calculus.
  % For rules that prove $[\alpha]\phi$ in the succedent, apply the corresponding structured execution rule.
  % If $\alpha$ is an ODE, the  proof (if in a nice normal form) ends in either $[']$, dI+dE+dW, dE+dW, dC, dG.
  % These map to solve, Inv+Finally, Finally, Inv1+IP, or Ghost+IP respectively.
  % For $\langle\alpha\rangle\phi$ on the right, use the analogous structured rules, except replace invariant reasoning with variant reasoning.
  % If proving a program modality on the left, use contextual reasoning to transform it into the opposing modality on the right.
  % Arithmetic is complete by completeness of QE.
\end{proof}

\begin{corollary} 
By G\"{o}del's incompleteness theorem, any sound calculus for \dL is incomplete in the absolute sense~\cite{DBLP:journals/jar/Platzer08} and thus so is Kaisar.
However, sequent calculus for \dL is relatively complete~\cite{DBLP:journals/jar/Platzer08} both with respect to any differentially expressive logic and with respect to discrete dynamics.
Thus Kaisar is as well.
\end{corollary}

%TODO: Variants

%\section{Applications}
%\label{sec:applications}

%TODO: Implement autocomplete.

%We wish to generalize it to address design issues which arise in these cases, as well as any further issues which arise for generalizations of 
%Furthermore, the proof rules as presented in this paper never remove assumptions from the context.
%We wish to implement optimizations which statically analyze the proof text to delete assumptions that are never used again, improving efficiency and scalability.

%Structured proof languages inherently come at the cost of some verbosity.
%To alleviate this issue, we wish to develop a proof IDE for Kaisar which can automatically generate proof skeletons following the structure of the proof, 

%The examples presented in this paper are adequate for purposes of explaining the language, but are significantly simpler than realistic hybrid systems case studies
%~\cite{DBLP:journals/corr/abs-1605-00604,DBLP:conf/emsoft/JeanninGKGSZP15,DBLP:conf/fm/LoosPN11}.
%We wish to further validate the design and implementation by performing an extended case study.

\section{Related Work}
%TODO: Cite standard hybrid systems stuff.
\label{sec:relatedwork}
Structured proof languages were first introduced in the theorem prover Mizar~\cite{conf/mkm/BancerekBGKMNPU15,Wenzel2002}, then expanded upon in systems such as Isar~\cite{DBLP:conf/mkm/Wenzel06,DBLP:conf/tphol/Wenzel99,Wenzel07isabelle/isar}. 
Other structured proof languages/extensions include the DECLARE~\cite{Syme1997DECLAREAP} proof system for HOL, TLAPS for ${\rm TLA}^+$~\cite{DBLP:conf/fm/CousineauDLMRV12,DBLP:conf/hybrid/Lamport92},  Coq's declarative proof language~\cite{DBLP:conf/types/Corbineau07} and SSReflect extension~\cite{DBLP:journals/jfrea/GonthierM10}, and several ``Mizar modes'' implementing structured languages in provers including Cambridge HOL~\cite{DBLP:conf/tphol/Harrison96}, Isabelle~\cite{DBLP:conf/cpp/KaliszykPU16}, and HOL Light~\cite{DBLP:conf/tphol/Wiedijk01}.
Kaisar is heavily inspired by Isar specifically, though we do not use every feature of Isar.
The \kwassume{}, \kwnote{}, \kwshow{}, and \kwhave{} are taken directly from Isar, and Kaisar's \kwlet{} construct is a straightforward generalization of Isar's \kwlet{} construct with pattern-matching.
The use of pattern-matching for formula selection has been investigated, e.g. by Traut ~\cite{traut2014pattern} and by Gonthier ~\cite{DBLP:conf/itp/GonthierT12}.
%and Noschinski  and Tassi

In contrast with all the above, Kaisar has an extensive metatheory to justify its defining features: nominal terms and structured symbolic execution.
While Isar and Coq's declarative language have formally defined semantics, we know of only one interactive proof language besides Kaisar with significant metatheoretic results:  VeriML~\cite{DBLP:conf/icfp/StampoulisS10,DBLP:conf/popl/StampoulisS12}.
We share with VeriML the goal of solving practical interactive proof problems via principled proof languages with metatheoretic guarantees.
We differ in that VeriML addresses extending logical frameworks with automation while we address verification of concrete systems in a domain-specific logic for CPS.

Nominal terms are unique to Kaisar among interactive proof languages, and give it a unique advantage in expressing the rich ghost state of hybrid systems proofs.
Other structured languages such as Isar have been used extensively for program verification (for example: ~\cite{Nipkow2002,DBLP:journals/cacm/KleinAEHCDEEKNSTW10,DBLP:journals/afp/Lochbihler07}).
However, because the above languages target general logics, they lack the language-level awareness of state change required for nominals.
We provide this language-level support through the novel technical features of structured symbolic execution and static traces.
We then implement the resulting language, reusing the infrastructure of an existing prover \KeYmaeraX.

An alternate approach is to implement a language at the user-level, in the tactics language of an existing prover.
This approach was taken, e.g. by the Iris Proof Mode (IPM) for Coq,~\cite{DBLP:conf/popl/KrebbersTB17} which implements reasoning for concurrent separation logic in Coq's \ltac language.
We work in the implementation language of \KeYmaeraX because it is far more expressive than its tactics language.
Language choice is incidental: the key is preserving the underlying prover's soundness guarantees.
As with Coq~\cite{Barras:1997}, \KeYmaeraX has an LCF-style core supported by mechanized soundness results for the underlying calculus~\cite{BohrerCPP17}, making the Kaisar implementation highly trustworthy.
%IPM simplifies reasoning in a  embedded in Coq.
Kaisar and IPM share a goal of building generalizable interactive proof technology for program logics, but they target vastly different logics and address different aspects of proof:
IPM uses Coq's unstructured proof style and focuses (a) on embedding object logics in metalogics and (b) on the concerns of separation logic (e.g. managing of different context types, state ownership).
In contrast, we augment structured proof with nominals to provide natural reasoning across states as needed in hybrid systems.
We conjecture that our basic approach applies to many logics, including separation logics.

%VeriML has a different focus: promoting prover extensibility while providing tactic efficiency with staging.
%Their metatheory validates the safety of staging, while our metatheory validates the correctness and generality of reasoning across states with nominals.

Two closely related language classes are tactics languages (which implement reuseable automation) and unstructured proof languages (which implement concrete proofs).
Automation can often be written in the prover's implementation language: OCaml in Coq, ML in Isabelle, or Scala in \KeYmaeraX.
Domain-specific languages for tactics include  untyped \ltac~\cite{Delahaye:2000:TLS:1765236.1765246}, reflective Rtac ~\cite{malecha2015rtac-coqpl}, and dependently typed Mtac~\cite{Ziliani:2013:MMT:2544174.2500579} in Coq, Eisbach~\cite{Matichuk:2016:EPM:2904234.2904264} in Isabelle, and VeriML~\cite{DBLP:conf/icfp/StampoulisS10,DBLP:conf/popl/StampoulisS12}.
Examples of unstructured languages are 
the Coq~\cite{COQ} script language and the Isabelle~\cite{Nipkow:2002:IPA:1791547} apply-script language.
 \KeYmaeraX features a language named Bellerophon~\cite{Belle} for unstructured proofs and tactics.
The Bellerophon language consists of regular expression-style tactic combinators (sequential composition, repetition, etc.) and a standard tactics library featuring, e.g. sequent calculus rules and general-purpose automation.
Bellerophon's strength is in tactics that compose the significant automation provided in its library.
Its weakness is in performing large-scale concrete proofs.
It lacks both the nominals unique to Kaisar and the constructs shared by Isar and Kaisar.
For example, assumptions in Bellerophon are unnamed and referred to by their index or by search, which can become unreadable or brittle at scale.

Our nominal terms relate to nominal differential dynamic logic~\cite{DBLP:journals/entcs/Platzer07} \dLN.
In \dLN{}, nominal formulas enable stating and proving theorems about named states.
Our goal differs: we apply named states to simplify proofs of theorems of plain \dL.
In our metatheory, \dLN{} formulas provide a clean specification for the nominal terms of Kaisar.

%\paragraph{Model Checking}
The main hybrid systems verification alternative to theorem proving is model-checking.
Because the uncountable state spaces of hybrid systems do not admit equivalent finite-state abstractions~\cite{DBLP:conf/lics/Henzinger96}, model-checking  approaches~\cite{DBLP:conf/hybrid/Frehse05,DBLP:conf/cav/FrehseGDCRLRGDM11,DBLP:conf/hybrid/Dreossi17,DBLP:conf/cav/HenzingerHW97,DBLP:conf/cade/GaoKC13,DBLP:conf/fmcad/GaoKC13,Chen2013} must approximate continuous dynamics, whereas \dL can reason about exact dynamics.
All of the above have limitations including (1) finite time horizons, (2) compact (and thus bounded) starting regions, (3) discrete notions of time, (4) and/or restriction to linear ODEs.
Restrictions (1) and (2) greatly reduce the scope of safety results, (3) reduces their accuracy and (4) reduces the class of systems considered. In contrast, $\dL$ supports unbounded continuous time with non-linear ODEs.

\section{Conclusion}
\label{sec:concl}
To simplify and systematize historical reference for verification of safety-critical CPS, we developed the Kaisar proof language for \dL, which introduces nominal terms supported by structured symbolic execution.
Our metatheory shows Kaisar is sound and as expressive as other calculi.
It shows that nominal automation is correct and nominals are the proof-language analog of nominal \dL.
In doing so we provide a foundation for ad-hoc historical reference in other provers.
%In developing the metatheory of Kaisar, we showed that those features have a solid theoretical foundation.
% is of interest beyond \dL, and could serve as a guide for applying the ideas of Kaisar to many other logics.
%In this paper, we have developed Kaisar, a structured proof language specialized to \dL, an expressive first-order logic for verification of hybrid systems.
%We have shown that in optimizing for a program logic, new proof language features become possible, such as structured symbolic execution and nominals.

Through our parachute example, we showed that nominals are desirable in CPS practice and that Kaisar proofs can be concise.
We provided empirical support by prototyping Kaisar in \KeYmaeraX, an implementation which supports the parachute example and other examples of this paper.
For evaluation, we reproduced a series of 5 safety proofs for ground robots~\cite{DBLP:conf/rss/MitschGP13}, combining differential invariant reasoning for nonsolvable dynamics with nontrivial arithmetic proofs, for models supporting avoidance of moving obstacles under position and actuator uncertainty.

\bibliographystyle{ACM-Reference-Format}
\bibliography{kaisar}


\begin{thebibliography}{00}

%%% the Web address.

\ifx \showCODEN    \undefined \def \showCODEN     #1{\unskip}     \fi
\ifx \showDOI      \undefined \def \showDOI       #1{#1}\fi
\ifx \showISBNx    \undefined \def \showISBNx     #1{\unskip}     \fi
\ifx \showISBNxiii \undefined \def \showISBNxiii  #1{\unskip}     \fi
\ifx \showISSN     \undefined \def \showISSN      #1{\unskip}     \fi
\ifx \showLCCN     \undefined \def \showLCCN      #1{\unskip}     \fi
\ifx \shownote     \undefined \def \shownote      #1{#1}          \fi
\ifx \showarticletitle \undefined \def \showarticletitle #1{#1}   \fi
\ifx \showURL      \undefined \def \showURL       {\relax}        \fi
% invisible to TeX
\providecommand\bibfield[2]{#2}
\providecommand\bibinfo[2]{#2}
\providecommand\natexlab[1]{#1}
\providecommand\showeprint[2][]{arXiv:#2}

\bibitem[\protect\citeauthoryear{Ahrendt, Beckert, Bubel, H\"{a}hnle, Schmitt,
  and Ulbrich}{Ahrendt et~al\mbox{.}}{2016}]%
        {KeYBook}
\bibfield{author}{\bibinfo{person}{Wolfgang Ahrendt}, \bibinfo{person}{Bernhard
  Beckert}, \bibinfo{person}{Richard Bubel}, \bibinfo{person}{Reiner
  H\"{a}hnle}, \bibinfo{person}{Peter~H. Schmitt}, {and}
  \bibinfo{person}{Mattias Ulbrich}.} \bibinfo{year}{2016}\natexlab{}.
\newblock \bibinfo{booktitle}{{\em Deductive Software Verification - The KeY
  Book}}.
\newblock \bibinfo{publisher}{Springer}.
\newblock
\showISBNx{9783319498119}


\bibitem[\protect\citeauthoryear{Apt, De~Boer, and Olderog}{Apt
  et~al\mbox{.}}{2010}]%
        {apt2010verification}
\bibfield{author}{\bibinfo{person}{Krzysztof Apt}, \bibinfo{person}{Frank~S
  De~Boer}, {and} \bibinfo{person}{Ernst-R{\"u}diger Olderog}.}
  \bibinfo{year}{2010}\natexlab{}.
\newblock \bibinfo{title}{Verification of sequential and concurrent programs}.
\newblock   (\bibinfo{year}{2010}).
\newblock


\bibitem[\protect\citeauthoryear{Apt, Bergstra, and Meertens}{Apt
  et~al\mbox{.}}{1979}]%
        {DBLP:journals/tcs/AptBM79}
\bibfield{author}{\bibinfo{person}{Krzysztof~R. Apt}, \bibinfo{person}{Jan~A.
  Bergstra}, {and} \bibinfo{person}{Lambert G. L.~T. Meertens}.}
  \bibinfo{year}{1979}\natexlab{}.
\newblock \showarticletitle{Recursive Assertions are not enough - or are they?}
\newblock \bibinfo{journal}{{\em Theor. Comput. Sci.\/}}  \bibinfo{volume}{8}
  (\bibinfo{year}{1979}), \bibinfo{pages}{73--87}.
\newblock
\showDOI{%
\url{https://doi.org/10.1016/0304-3975(79)90058-6}}


\bibitem[\protect\citeauthoryear{Arnon, Collins, and McCallum}{Arnon
  et~al\mbox{.}}{1984}]%
        {CAD}
\bibfield{author}{\bibinfo{person}{Dennis~S. Arnon}, \bibinfo{person}{George~E.
  Collins}, {and} \bibinfo{person}{Scott McCallum}.}
  \bibinfo{year}{1984}\natexlab{}.
\newblock \showarticletitle{Cylindrical Algebraic Decomposition I: The Basic
  Algorithm}.
\newblock \bibinfo{journal}{{\em SIAM J. Comput.\/}} \bibinfo{volume}{13},
  \bibinfo{number}{4} (\bibinfo{date}{Nov.} \bibinfo{year}{1984}),
  \bibinfo{pages}{865--877}.
\newblock
\showISSN{0097-5397}
\showDOI{%
\url{https://doi.org/10.1137/0213054}}


\bibitem[\protect\citeauthoryear{Bancerek, Bylinski, Grabowski, Kornilowicz,
  Matuszewski, Naumowicz, Pak, and Urban}{Bancerek et~al\mbox{.}}{2015}]%
        {conf/mkm/BancerekBGKMNPU15}
\bibfield{author}{\bibinfo{person}{Grzegorz Bancerek}, \bibinfo{person}{Czeslaw
  Bylinski}, \bibinfo{person}{Adam Grabowski}, \bibinfo{person}{Artur
  Kornilowicz}, \bibinfo{person}{Roman Matuszewski}, \bibinfo{person}{Adam
  Naumowicz}, \bibinfo{person}{Karol Pak}, {and} \bibinfo{person}{Josef
  Urban}.} \bibinfo{year}{2015}\natexlab{}.
\newblock \showarticletitle{Mizar: State-of-the-art and Beyond.}. In
  \bibinfo{booktitle}{{\em CICM}} {\em (\bibinfo{series}{Lecture Notes in
  Computer Science})}, \bibfield{editor}{\bibinfo{person}{Manfred Kerber},
  \bibinfo{person}{Jacques Carette}, \bibinfo{person}{Cezary Kaliszyk},
  \bibinfo{person}{Florian Rabe}, {and} \bibinfo{person}{Volker Sorge}} (Eds.),
  Vol.~\bibinfo{volume}{9150}. \bibinfo{publisher}{Springer},
  \bibinfo{pages}{261--279}.
\newblock
\showISBNx{978-3-319-20614-1}
\showDOI{%
\url{https://doi.org/10.1007/978-3-319-20615-8}}


\bibitem[\protect\citeauthoryear{Barnett, Leino, and Schulte}{Barnett
  et~al\mbox{.}}{2005}]%
        {Barnett2005}
\bibfield{author}{\bibinfo{person}{Mike Barnett}, \bibinfo{person}{K~Rustan~M
  Leino}, {and} \bibinfo{person}{Wolfram Schulte}.}
  \bibinfo{year}{2005}\natexlab{}.
\newblock \bibinfo{booktitle}{{\em {The Spec{\{}{\#}{\}} Programming System: An
  Overview}}}.
\newblock \bibinfo{publisher}{Springer Berlin Heidelberg},
  \bibinfo{address}{Berlin, Heidelberg}, \bibinfo{pages}{49--69}.
\newblock
\showISBNx{978-3-540-30569-9}
\showDOI{%
\url{https://doi.org/10.1007/978-3-540-30569-9_3}}


\bibitem[\protect\citeauthoryear{Barras and Werner}{Barras and Werner}{1997}]%
        {Barras:1997}
\bibfield{author}{\bibinfo{person}{Bruno Barras} {and}
  \bibinfo{person}{Benjamin Werner}.} \bibinfo{year}{1997}\natexlab{}.
\newblock \bibinfo{booktitle}{{\em {C}oq in {C}oq}}.
\newblock \bibinfo{type}{{T}echnical {R}eport}. \bibinfo{institution}{INRIA
  Rocquencourt}.
\newblock


\bibitem[\protect\citeauthoryear{Bohrer, Rahli, Vukotic, V{\"{o}}lp, and
  Platzer}{Bohrer et~al\mbox{.}}{2017}]%
        {BohrerCPP17}
\bibfield{author}{\bibinfo{person}{Brandon Bohrer}, \bibinfo{person}{Vincent
  Rahli}, \bibinfo{person}{Ivana Vukotic}, \bibinfo{person}{Marcus V{\"{o}}lp},
  {and} \bibinfo{person}{Andr{\'{e}} Platzer}.}
  \bibinfo{year}{2017}\natexlab{}.
\newblock \showarticletitle{{Formally Verified Differential Dynamic Logic}}. In
  \bibinfo{booktitle}{{\em Certified Programs and Proofs - 6th ACM SIGPLAN
  Conference, CPP 2017, Paris, France, January 16-17, 2017}},
  \bibfield{editor}{\bibinfo{person}{Yves Bertot} {and} \bibinfo{person}{Viktor
  Vafeiadis}} (Eds.). \bibinfo{publisher}{ACM}, \bibinfo{pages}{208--221}.
\newblock


\bibitem[\protect\citeauthoryear{Boker, Henzinger, and Radhakrishna}{Boker
  et~al\mbox{.}}{2014}]%
        {DBLP:conf/popl/BokerHR14}
\bibfield{author}{\bibinfo{person}{Udi Boker}, \bibinfo{person}{Thomas~A.
  Henzinger}, {and} \bibinfo{person}{Arjun Radhakrishna}.}
  \bibinfo{year}{2014}\natexlab{}.
\newblock \showarticletitle{Battery transition systems}. In
  \bibinfo{booktitle}{{\em The 41st Annual {ACM} {SIGPLAN-SIGACT} Symposium on
  Principles of Programming Languages, {POPL} '14, San Diego, CA, USA, January
  20-21, 2014}}, \bibfield{editor}{\bibinfo{person}{Suresh Jagannathan} {and}
  \bibinfo{person}{Peter Sewell}} (Eds.). \bibinfo{publisher}{{ACM}},
  \bibinfo{pages}{595--606}.
\newblock
\showISBNx{978-1-4503-2544-8}
\showDOI{%
\url{https://doi.org/10.1145/2535838.2535875}}


\bibitem[\protect\citeauthoryear{Chen, {\'{A}}brah{\'{a}}m, and
  Sankaranarayanan}{Chen et~al\mbox{.}}{2013}]%
        {Chen2013}
\bibfield{author}{\bibinfo{person}{Xin Chen}, \bibinfo{person}{Erika
  {\'{A}}brah{\'{a}}m}, {and} \bibinfo{person}{Sriram Sankaranarayanan}.}
  \bibinfo{year}{2013}\natexlab{}.
\newblock \bibinfo{booktitle}{{\em {Flow*: An Analyzer for Non-linear Hybrid
  Systems}}}.
\newblock \bibinfo{publisher}{Springer Berlin Heidelberg},
  \bibinfo{address}{Berlin, Heidelberg}, \bibinfo{pages}{258--263}.
\newblock
\showISBNx{978-3-642-39799-8}
\showDOI{%
\url{https://doi.org/10.1007/978-3-642-39799-8_18}}


\bibitem[\protect\citeauthoryear{Clarke}{Clarke}{1980}]%
        {DBLP:journals/acta/Clarke80}
\bibfield{author}{\bibinfo{person}{Edmund~M. Clarke}.}
  \bibinfo{year}{1980}\natexlab{}.
\newblock \showarticletitle{Proving Correctness of Coroutines Without History
  Variables}.
\newblock \bibinfo{journal}{{\em Acta Inf.\/}}  \bibinfo{volume}{13}
  (\bibinfo{year}{1980}), \bibinfo{pages}{169--188}.
\newblock
\showDOI{%
\url{https://doi.org/10.1007/BF00263992}}


\bibitem[\protect\citeauthoryear{Clint}{Clint}{1973}]%
        {DBLP:journals/acta/Clint73}
\bibfield{author}{\bibinfo{person}{Maurice Clint}.}
  \bibinfo{year}{1973}\natexlab{}.
\newblock \showarticletitle{Program Proving: Coroutines}.
\newblock \bibinfo{journal}{{\em Acta Inf.\/}}  \bibinfo{volume}{2}
  (\bibinfo{year}{1973}), \bibinfo{pages}{50--63}.
\newblock
\showDOI{%
\url{https://doi.org/10.1007/BF00571463}}


\bibitem[\protect\citeauthoryear{Collins and Hong}{Collins and Hong}{1991}]%
        {PCAD}
\bibfield{author}{\bibinfo{person}{George~E. Collins} {and}
  \bibinfo{person}{Hoon Hong}.} \bibinfo{year}{1991}\natexlab{}.
\newblock \showarticletitle{Partial Cylindrical Algebraic Decomposition for
  Quantifier Elimination}.
\newblock \bibinfo{journal}{{\em J. Symb. Comput.\/}} \bibinfo{volume}{12},
  \bibinfo{number}{3} (\bibinfo{date}{Sept.} \bibinfo{year}{1991}),
  \bibinfo{pages}{299--328}.
\newblock
\showISSN{0747-7171}
\showDOI{%
\url{https://doi.org/10.1016/S0747-7171(08)80152-6}}


\bibitem[\protect\citeauthoryear{Corbineau}{Corbineau}{2007}]%
        {DBLP:conf/types/Corbineau07}
\bibfield{author}{\bibinfo{person}{Pierre Corbineau}.}
  \bibinfo{year}{2007}\natexlab{}.
\newblock \showarticletitle{A Declarative Language for the Coq Proof
  Assistant}. In \bibinfo{booktitle}{{\em Types for Proofs and Programs,
  International Conference, {TYPES} 2007, Cividale del Friuli, Italy, May 2-5,
  2007, Revised Selected Papers}} {\em (\bibinfo{series}{Lecture Notes in
  Computer Science})}, \bibfield{editor}{\bibinfo{person}{Marino Miculan},
  \bibinfo{person}{Ivan Scagnetto}, {and} \bibinfo{person}{Furio Honsell}}
  (Eds.), Vol.~\bibinfo{volume}{4941}. \bibinfo{publisher}{Springer},
  \bibinfo{pages}{69--84}.
\newblock
\showISBNx{978-3-540-68084-0}
\showDOI{%
\url{https://doi.org/10.1007/978-3-540-68103-8_5}}


\bibitem[\protect\citeauthoryear{Cousineau, Doligez, Lamport, Merz, Ricketts,
  and Vanzetto}{Cousineau et~al\mbox{.}}{2012}]%
        {DBLP:conf/fm/CousineauDLMRV12}
\bibfield{author}{\bibinfo{person}{Denis Cousineau}, \bibinfo{person}{Damien
  Doligez}, \bibinfo{person}{Leslie Lamport}, \bibinfo{person}{Stephan Merz},
  \bibinfo{person}{Daniel Ricketts}, {and} \bibinfo{person}{Hern{\'{a}}n
  Vanzetto}.} \bibinfo{year}{2012}\natexlab{}.
\newblock \showarticletitle{{TLA} + Proofs}. In \bibinfo{booktitle}{{\em {FM}
  2012: Formal Methods - 18th International Symposium, Paris, France, August
  27-31, 2012. Proceedings}} {\em (\bibinfo{series}{Lecture Notes in Computer
  Science})}, \bibfield{editor}{\bibinfo{person}{Dimitra Giannakopoulou} {and}
  \bibinfo{person}{Dominique M{\'{e}}ry}} (Eds.), Vol.~\bibinfo{volume}{7436}.
  \bibinfo{publisher}{Springer}, \bibinfo{pages}{147--154}.
\newblock
\showISBNx{978-3-642-32758-2}
\showDOI{%
\url{https://doi.org/10.1007/978-3-642-32759-9_14}}


\bibitem[\protect\citeauthoryear{Davenport and Heintz}{Davenport and
  Heintz}{1988}]%
        {Davenport1988}
\bibfield{author}{\bibinfo{person}{James~H. Davenport} {and}
  \bibinfo{person}{Joos Heintz}.} \bibinfo{year}{1988}\natexlab{}.
\newblock \showarticletitle{{Real quantifier elimination is doubly
  exponential}}.
\newblock \bibinfo{journal}{{\em Journal of Symbolic Computation\/}}
  \bibinfo{volume}{5}, \bibinfo{number}{1-2} (\bibinfo{year}{1988}),
  \bibinfo{pages}{29--35}.
\newblock
\showISSN{07477171}
\showDOI{%
\url{https://doi.org/10.1016/S0747-7171(88)80004-X}}


\bibitem[\protect\citeauthoryear{Davis}{Davis}{1981}]%
        {DBLP:conf/ijcai/Davis81}
\bibfield{author}{\bibinfo{person}{Martin Davis}.}
  \bibinfo{year}{1981}\natexlab{}.
\newblock \showarticletitle{Obvious Logical Inferences}. In
  \bibinfo{booktitle}{{\em Proceedings of the 7th International Joint
  Conference on Artificial Intelligence, {IJCAI} '81, Vancouver, BC, Canada,
  August 24-28, 1981}}, \bibfield{editor}{\bibinfo{person}{Patrick~J. Hayes}}
  (Ed.). \bibinfo{publisher}{William Kaufmann}, \bibinfo{pages}{530--531}.
\newblock
\showURL{%
\url{http://ijcai.org/Proceedings/81-1/Papers/095.pdf}}


\bibitem[\protect\citeauthoryear{Delahaye}{Delahaye}{2000}]%
        {Delahaye:2000:TLS:1765236.1765246}
\bibfield{author}{\bibinfo{person}{David Delahaye}.}
  \bibinfo{year}{2000}\natexlab{}.
\newblock \showarticletitle{A Tactic Language for the System Coq}. In
  \bibinfo{booktitle}{{\em Proceedings of the 7th International Conference on
  Logic for Programming and Automated Reasoning}} {\em
  (\bibinfo{series}{LPAR'00})}. \bibinfo{publisher}{Springer-Verlag},
  \bibinfo{address}{Berlin, Heidelberg}, \bibinfo{pages}{85--95}.
\newblock
\showISBNx{3-540-41285-9}
\showURL{%
\url{http://dl.acm.org/citation.cfm?id=1765236.1765246}}


\bibitem[\protect\citeauthoryear{Dreossi}{Dreossi}{2017}]%
        {DBLP:conf/hybrid/Dreossi17}
\bibfield{author}{\bibinfo{person}{Tommaso Dreossi}.}
  \bibinfo{year}{2017}\natexlab{}.
\newblock \showarticletitle{Sapo: Reachability Computation and Parameter
  Synthesis of Polynomial Dynamical Systems}. In \bibinfo{booktitle}{{\em
  Proceedings of the 20th International Conference on Hybrid Systems:
  Computation and Control, {HSCC} 2017, Pittsburgh, PA, USA, April 18-20,
  2017}}, \bibfield{editor}{\bibinfo{person}{Goran Frehse} {and}
  \bibinfo{person}{Sayan Mitra}} (Eds.). \bibinfo{publisher}{{ACM}},
  \bibinfo{pages}{29--34}.
\newblock
\showISBNx{978-1-4503-4590-3}
\showDOI{%
\url{https://doi.org/10.1145/3049797.3049824}}


\bibitem[\protect\citeauthoryear{Frehse}{Frehse}{2005}]%
        {DBLP:conf/hybrid/Frehse05}
\bibfield{author}{\bibinfo{person}{Goran Frehse}.}
  \bibinfo{year}{2005}\natexlab{}.
\newblock \showarticletitle{PHAVer: Algorithmic Verification of Hybrid Systems
  Past HyTech}. In \bibinfo{booktitle}{{\em Hybrid Systems: Computation and
  Control, 8th International Workshop, {HSCC} 2005, Zurich, Switzerland, March
  9-11, 2005, Proceedings}} {\em (\bibinfo{series}{Lecture Notes in Computer
  Science})}, \bibfield{editor}{\bibinfo{person}{Manfred Morari} {and}
  \bibinfo{person}{Lothar Thiele}} (Eds.), Vol.~\bibinfo{volume}{3414}.
  \bibinfo{publisher}{Springer}, \bibinfo{pages}{258--273}.
\newblock
\showISBNx{3-540-25108-1}
\showDOI{%
\url{https://doi.org/10.1007/978-3-540-31954-2_17}}


\bibitem[\protect\citeauthoryear{Frehse, Guernic, Donz{\'{e}}, Cotton, Ray,
  Lebeltel, Ripado, Girard, Dang, and Maler}{Frehse et~al\mbox{.}}{2011}]%
        {DBLP:conf/cav/FrehseGDCRLRGDM11}
\bibfield{author}{\bibinfo{person}{Goran Frehse}, \bibinfo{person}{Colas~Le
  Guernic}, \bibinfo{person}{Alexandre Donz{\'{e}}}, \bibinfo{person}{Scott
  Cotton}, \bibinfo{person}{Rajarshi Ray}, \bibinfo{person}{Olivier Lebeltel},
  \bibinfo{person}{Rodolfo Ripado}, \bibinfo{person}{Antoine Girard},
  \bibinfo{person}{Thao Dang}, {and} \bibinfo{person}{Oded Maler}.}
  \bibinfo{year}{2011}\natexlab{}.
\newblock \showarticletitle{SpaceEx: Scalable Verification of Hybrid Systems}.
  In \bibinfo{booktitle}{{\em Computer Aided Verification - 23rd International
  Conference, {CAV} 2011, Snowbird, UT, USA, July 14-20, 2011. Proceedings}}
  {\em (\bibinfo{series}{Lecture Notes in Computer Science})},
  \bibfield{editor}{\bibinfo{person}{Ganesh Gopalakrishnan} {and}
  \bibinfo{person}{Shaz Qadeer}} (Eds.), Vol.~\bibinfo{volume}{6806}.
  \bibinfo{publisher}{Springer}, \bibinfo{pages}{379--395}.
\newblock
\showISBNx{978-3-642-22109-5}
\showDOI{%
\url{https://doi.org/10.1007/978-3-642-22110-1_30}}


\bibitem[\protect\citeauthoryear{Fulton, Mitsch, Bohrer, and Platzer}{Fulton
  et~al\mbox{.}}{2017}]%
        {Belle}
\bibfield{author}{\bibinfo{person}{Nathan Fulton}, \bibinfo{person}{Stefan
  Mitsch}, \bibinfo{person}{Brandon Bohrer}, {and} \bibinfo{person}{Andr{\'e}
  Platzer}.} \bibinfo{year}{2017}\natexlab{}.
\newblock \showarticletitle{{Bellerophon: Tactical Theorem Proving for Hybrid
  Systems}}. In \bibinfo{booktitle}{{\em Interactive Theorem Proving - Eighth
  International Conference, {ITP} 2017, Brasilia, Brasil, September 26-29,
  2017. To Appear}}.
\newblock
\showURL{%
\url{https://nfulton.org/papers/bellerophon.pdf}}


\bibitem[\protect\citeauthoryear{Fulton, Mitsch, Quesel, V{\"o}lp, and
  Platzer}{Fulton et~al\mbox{.}}{2015}]%
        {DBLP:conf/cade/FultonMQVP15}
\bibfield{author}{\bibinfo{person}{Nathan Fulton}, \bibinfo{person}{Stefan
  Mitsch}, \bibinfo{person}{Jan-David Quesel}, \bibinfo{person}{Marcus
  V{\"o}lp}, {and} \bibinfo{person}{Andr{\'e} Platzer}.}
  \bibinfo{year}{2015}\natexlab{}.
\newblock \showarticletitle{{KeYmaera X}: An Axiomatic Tactical Theorem Prover
  for Hybrid Systems}. In \bibinfo{booktitle}{{\em CADE}} {\em
  (\bibinfo{series}{LNCS})}, \bibfield{editor}{\bibinfo{person}{Amy~P. Felty}
  {and} \bibinfo{person}{Aart Middeldorp}} (Eds.), Vol.~\bibinfo{volume}{9195}.
  \bibinfo{publisher}{Springer}, \bibinfo{pages}{527--538}.
\newblock
\showDOI{%
\url{https://doi.org/10.1007/978-3-319-21401-6_36}}


\bibitem[\protect\citeauthoryear{Gao, Kong, and Clarke}{Gao
  et~al\mbox{.}}{2013a}]%
        {DBLP:conf/cade/GaoKC13}
\bibfield{author}{\bibinfo{person}{Sicun Gao}, \bibinfo{person}{Soonho Kong},
  {and} \bibinfo{person}{Edmund~M. Clarke}.} \bibinfo{year}{2013}\natexlab{a}.
\newblock \showarticletitle{dReal: An {SMT} Solver for Nonlinear Theories over
  the Reals}. In \bibinfo{booktitle}{{\em Automated Deduction - {CADE-24} -
  24th International Conference on Automated Deduction, Lake Placid, NY, USA,
  June 9-14, 2013. Proceedings}} {\em (\bibinfo{series}{Lecture Notes in
  Computer Science})}, \bibfield{editor}{\bibinfo{person}{Maria~Paola
  Bonacina}} (Ed.), Vol.~\bibinfo{volume}{7898}. \bibinfo{publisher}{Springer},
  \bibinfo{pages}{208--214}.
\newblock
\showISBNx{978-3-642-38573-5}
\showDOI{%
\url{https://doi.org/10.1007/978-3-642-38574-2_14}}


\bibitem[\protect\citeauthoryear{Gao, Kong, and Clarke}{Gao
  et~al\mbox{.}}{2013b}]%
        {DBLP:conf/fmcad/GaoKC13}
\bibfield{author}{\bibinfo{person}{Sicun Gao}, \bibinfo{person}{Soonho Kong},
  {and} \bibinfo{person}{Edmund~M. Clarke}.} \bibinfo{year}{2013}\natexlab{b}.
\newblock \showarticletitle{Satisfiability modulo ODEs}. In
  \bibinfo{booktitle}{{\em Formal Methods in Computer-Aided Design, {FMCAD}
  2013, Portland, OR, USA, October 20-23, 2013}}. \bibinfo{publisher}{{IEEE}},
  \bibinfo{pages}{105--112}.
\newblock
\showURL{%
\url{http://ieeexplore.ieee.org/document/6679398/}}


\bibitem[\protect\citeauthoryear{Gonthier and Mahboubi}{Gonthier and
  Mahboubi}{2010}]%
        {DBLP:journals/jfrea/GonthierM10}
\bibfield{author}{\bibinfo{person}{Georges Gonthier} {and}
  \bibinfo{person}{Assia Mahboubi}.} \bibinfo{year}{2010}\natexlab{}.
\newblock \showarticletitle{An introduction to small scale reflection in Coq}.
\newblock \bibinfo{journal}{{\em J. Formalized Reasoning\/}}
  \bibinfo{volume}{3}, \bibinfo{number}{2} (\bibinfo{year}{2010}),
  \bibinfo{pages}{95--152}.
\newblock
\showDOI{%
\url{https://doi.org/10.6092/issn.1972-5787/1979}}


\bibitem[\protect\citeauthoryear{Gonthier and Tassi}{Gonthier and
  Tassi}{2012}]%
        {DBLP:conf/itp/GonthierT12}
\bibfield{author}{\bibinfo{person}{Georges Gonthier} {and}
  \bibinfo{person}{Enrico Tassi}.} \bibinfo{year}{2012}\natexlab{}.
\newblock \showarticletitle{A Language of Patterns for Subterm Selection}. In
  \bibinfo{booktitle}{{\em Interactive Theorem Proving - Third International
  Conference, {ITP} 2012, Princeton, NJ, USA, August 13-15, 2012. Proceedings}}
  {\em (\bibinfo{series}{Lecture Notes in Computer Science})},
  \bibfield{editor}{\bibinfo{person}{Lennart Beringer} {and}
  \bibinfo{person}{Amy~P. Felty}} (Eds.), Vol.~\bibinfo{volume}{7406}.
  \bibinfo{publisher}{Springer}, \bibinfo{pages}{361--376}.
\newblock
\showISBNx{978-3-642-32346-1}
\showDOI{%
\url{https://doi.org/10.1007/978-3-642-32347-8_25}}


\bibitem[\protect\citeauthoryear{Harel, Tiuryn, and Kozen}{Harel
  et~al\mbox{.}}{2000}]%
        {Harel}
\bibfield{author}{\bibinfo{person}{David Harel}, \bibinfo{person}{Jerzy
  Tiuryn}, {and} \bibinfo{person}{Dexter Kozen}.}
  \bibinfo{year}{2000}\natexlab{}.
\newblock \bibinfo{booktitle}{{\em Dynamic Logic}}.
\newblock \bibinfo{publisher}{MIT Press}, \bibinfo{address}{Cambridge, MA,
  USA}.
\newblock
\showISBNx{0262082896}


\bibitem[\protect\citeauthoryear{Harrison}{Harrison}{1996}]%
        {DBLP:conf/tphol/Harrison96}
\bibfield{author}{\bibinfo{person}{John Harrison}.}
  \bibinfo{year}{1996}\natexlab{}.
\newblock \showarticletitle{A Mizar Mode for {HOL}}. In
  \bibinfo{booktitle}{{\em Theorem Proving in Higher Order Logics, 9th
  International Conference, TPHOLs'96, Turku, Finland, August 26-30, 1996,
  Proceedings}} {\em (\bibinfo{series}{Lecture Notes in Computer Science})},
  \bibfield{editor}{\bibinfo{person}{Joakim von Wright}, \bibinfo{person}{Jim
  Grundy}, {and} \bibinfo{person}{John Harrison}} (Eds.),
  Vol.~\bibinfo{volume}{1125}. \bibinfo{publisher}{Springer},
  \bibinfo{pages}{203--220}.
\newblock
\showISBNx{3-540-61587-3}
\showDOI{%
\url{https://doi.org/10.1007/BFb0105406}}


\bibitem[\protect\citeauthoryear{Henzinger}{Henzinger}{1996}]%
        {DBLP:conf/lics/Henzinger96}
\bibfield{author}{\bibinfo{person}{Thomas~A. Henzinger}.}
  \bibinfo{year}{1996}\natexlab{}.
\newblock \showarticletitle{The Theory of Hybrid Automata}. In
  \bibinfo{booktitle}{{\em Proceedings, 11th Annual {IEEE} Symposium on Logic
  in Computer Science, New Brunswick, New Jersey, USA, July 27-30, 1996}}.
  \bibinfo{publisher}{{IEEE} Computer Society}, \bibinfo{pages}{278--292}.
\newblock
\showISBNx{0-8186-7463-6}
\showDOI{%
\url{https://doi.org/10.1109/LICS.1996.561342}}


\bibitem[\protect\citeauthoryear{Henzinger, Ho, and Wong{-}Toi}{Henzinger
  et~al\mbox{.}}{1997}]%
        {DBLP:conf/cav/HenzingerHW97}
\bibfield{author}{\bibinfo{person}{Thomas~A. Henzinger},
  \bibinfo{person}{Pei{-}Hsin Ho}, {and} \bibinfo{person}{Howard Wong{-}Toi}.}
  \bibinfo{year}{1997}\natexlab{}.
\newblock \showarticletitle{{HYTECH:} {A} Model Checker for Hybrid Systems}. In
  \bibinfo{booktitle}{{\em Computer Aided Verification, 9th International
  Conference, {CAV} '97, Haifa, Israel, June 22-25, 1997, Proceedings}} {\em
  (\bibinfo{series}{Lecture Notes in Computer Science})},
  \bibfield{editor}{\bibinfo{person}{Orna Grumberg}} (Ed.),
  Vol.~\bibinfo{volume}{1254}. \bibinfo{publisher}{Springer},
  \bibinfo{pages}{460--463}.
\newblock
\showISBNx{3-540-63166-6}
\showDOI{%
\url{https://doi.org/10.1007/3-540-63166-6_48}}


\bibitem[\protect\citeauthoryear{Hoare}{Hoare}{1969}]%
        {Hoare:1969:ABC:363235.363259}
\bibfield{author}{\bibinfo{person}{C.~A.~R. Hoare}.}
  \bibinfo{year}{1969}\natexlab{}.
\newblock \showarticletitle{An Axiomatic Basis for Computer Programming}.
\newblock \bibinfo{journal}{{\em Commun. ACM\/}} \bibinfo{volume}{12},
  \bibinfo{number}{10} (\bibinfo{date}{Oct.} \bibinfo{year}{1969}),
  \bibinfo{pages}{576--580}.
\newblock
\showISSN{0001-0782}
\showDOI{%
\url{https://doi.org/10.1145/363235.363259}}


\bibitem[\protect\citeauthoryear{Jeannin, Ghorbal, Kouskoulas, Gardner,
  Schmidt, Zawadzki, and Platzer}{Jeannin et~al\mbox{.}}{2015}]%
        {DBLP:conf/emsoft/JeanninGKGSZP15}
\bibfield{author}{\bibinfo{person}{Jean{-}Baptiste Jeannin},
  \bibinfo{person}{Khalil Ghorbal}, \bibinfo{person}{Yanni Kouskoulas},
  \bibinfo{person}{Ryan Gardner}, \bibinfo{person}{Aurora Schmidt},
  \bibinfo{person}{Erik Zawadzki}, {and} \bibinfo{person}{Andr{\'e} Platzer}.}
  \bibinfo{year}{2015}\natexlab{}.
\newblock \showarticletitle{Formal Verification of {ACAS X}, an Industrial
  Airborne Collision Avoidance System}. In \bibinfo{booktitle}{{\em EMSOFT}},
  \bibfield{editor}{\bibinfo{person}{Alain Girault} {and} \bibinfo{person}{Nan
  Guan}} (Eds.). \bibinfo{publisher}{IEEE Press}, \bibinfo{pages}{127--136}.
\newblock
\showDOI{%
\url{https://doi.org/10.1109/EMSOFT.2015.7318268}}


\bibitem[\protect\citeauthoryear{Kaliszyk, Pak, and Urban}{Kaliszyk
  et~al\mbox{.}}{2016}]%
        {DBLP:conf/cpp/KaliszykPU16}
\bibfield{author}{\bibinfo{person}{Cezary Kaliszyk}, \bibinfo{person}{Karol
  Pak}, {and} \bibinfo{person}{Josef Urban}.} \bibinfo{year}{2016}\natexlab{}.
\newblock \showarticletitle{Towards a mizar environment for isabelle:
  foundations and language}. In \bibinfo{booktitle}{{\em Proceedings of the 5th
  {ACM} {SIGPLAN} Conference on Certified Programs and Proofs, Saint
  Petersburg, FL, USA, January 20-22, 2016}},
  \bibfield{editor}{\bibinfo{person}{Jeremy Avigad} {and} \bibinfo{person}{Adam
  Chlipala}} (Eds.). \bibinfo{publisher}{{ACM}}, \bibinfo{pages}{58--65}.
\newblock
\showISBNx{978-1-4503-4127-1}
\showDOI{%
\url{https://doi.org/10.1145/2854065.2854070}}


\bibitem[\protect\citeauthoryear{Kido, Chaudhuri, and Hasuo}{Kido
  et~al\mbox{.}}{2016}]%
        {DBLP:conf/vmcai/KidoCH16}
\bibfield{author}{\bibinfo{person}{Kengo Kido}, \bibinfo{person}{Swarat
  Chaudhuri}, {and} \bibinfo{person}{Ichiro Hasuo}.}
  \bibinfo{year}{2016}\natexlab{}.
\newblock \showarticletitle{Abstract Interpretation with Infinitesimals -
  Towards Scalability in Nonstandard Static Analysis}. In
  \bibinfo{booktitle}{{\em Verification, Model Checking, and Abstract
  Interpretation - 17th International Conference, {VMCAI} 2016, St. Petersburg,
  FL, USA, January 17-19, 2016. Proceedings}} {\em (\bibinfo{series}{Lecture
  Notes in Computer Science})}, \bibfield{editor}{\bibinfo{person}{Barbara
  Jobstmann} {and} \bibinfo{person}{K.~Rustan~M. Leino}} (Eds.),
  Vol.~\bibinfo{volume}{9583}. \bibinfo{publisher}{Springer},
  \bibinfo{pages}{229--249}.
\newblock
\showISBNx{978-3-662-49121-8}
\showDOI{%
\url{https://doi.org/10.1007/978-3-662-49122-5_11}}


\bibitem[\protect\citeauthoryear{Klein, Andronick, Elphinstone, Heiser, Cock,
  Derrin, Elkaduwe, Engelhardt, Kolanski, Norrish, Sewell, Tuch, and
  Winwood}{Klein et~al\mbox{.}}{2010}]%
        {DBLP:journals/cacm/KleinAEHCDEEKNSTW10}
\bibfield{author}{\bibinfo{person}{Gerwin Klein}, \bibinfo{person}{June
  Andronick}, \bibinfo{person}{Kevin Elphinstone}, \bibinfo{person}{Gernot
  Heiser}, \bibinfo{person}{David Cock}, \bibinfo{person}{Philip Derrin},
  \bibinfo{person}{Dhammika Elkaduwe}, \bibinfo{person}{Kai Engelhardt},
  \bibinfo{person}{Rafal Kolanski}, \bibinfo{person}{Michael Norrish},
  \bibinfo{person}{Thomas Sewell}, \bibinfo{person}{Harvey Tuch}, {and}
  \bibinfo{person}{Simon Winwood}.} \bibinfo{year}{2010}\natexlab{}.
\newblock \showarticletitle{seL4: formal verification of an operating-system
  kernel}.
\newblock \bibinfo{journal}{{\em Commun. {ACM}\/}} \bibinfo{volume}{53},
  \bibinfo{number}{6} (\bibinfo{year}{2010}), \bibinfo{pages}{107--115}.
\newblock
\showDOI{%
\url{https://doi.org/10.1145/1743546.1743574}}


\bibitem[\protect\citeauthoryear{Kouskoulas, Renshaw, Platzer, and
  Kazanzides}{Kouskoulas et~al\mbox{.}}{2013}]%
        {DBLP:conf/hybrid/KouskoulasRPK13}
\bibfield{author}{\bibinfo{person}{Yanni Kouskoulas}, \bibinfo{person}{David~W.
  Renshaw}, \bibinfo{person}{Andr{\'{e}} Platzer}, {and} \bibinfo{person}{Peter
  Kazanzides}.} \bibinfo{year}{2013}\natexlab{}.
\newblock \showarticletitle{Certifying the safe design of a virtual fixture
  control algorithm for a surgical robot}. In \bibinfo{booktitle}{{\em
  Proceedings of the 16th international conference on Hybrid systems:
  computation and control, {HSCC} 2013, April 8-11, 2013, Philadelphia, PA,
  {USA}}}, \bibfield{editor}{\bibinfo{person}{Calin Belta} {and}
  \bibinfo{person}{Franjo Ivancic}} (Eds.). \bibinfo{publisher}{{ACM}},
  \bibinfo{pages}{263--272}.
\newblock
\showISBNx{978-1-4503-1567-8}
\showDOI{%
\url{https://doi.org/10.1145/2461328.2461369}}


\bibitem[\protect\citeauthoryear{Krebbers, Timany, and Birkedal}{Krebbers
  et~al\mbox{.}}{2017}]%
        {DBLP:conf/popl/KrebbersTB17}
\bibfield{author}{\bibinfo{person}{Robbert Krebbers}, \bibinfo{person}{Amin
  Timany}, {and} \bibinfo{person}{Lars Birkedal}.}
  \bibinfo{year}{2017}\natexlab{}.
\newblock \showarticletitle{Interactive proofs in higher-order concurrent
  separation logic}. In \bibinfo{booktitle}{{\em Proceedings of the 44th {ACM}
  {SIGPLAN} Symposium on Principles of Programming Languages, {POPL} 2017,
  Paris, France, January 18-20, 2017}},
  \bibfield{editor}{\bibinfo{person}{Giuseppe Castagna} {and}
  \bibinfo{person}{Andrew~D. Gordon}} (Eds.). \bibinfo{publisher}{{ACM}},
  \bibinfo{pages}{205--217}.
\newblock
\showISBNx{978-1-4503-4660-3}
\showDOI{%
\url{https://doi.org/10.1145/3009837}}


\bibitem[\protect\citeauthoryear{Lamport}{Lamport}{1992}]%
        {DBLP:conf/hybrid/Lamport92}
\bibfield{author}{\bibinfo{person}{Leslie Lamport}.}
  \bibinfo{year}{1992}\natexlab{}.
\newblock \showarticletitle{Hybrid Systems in TLA\({}^{\mbox{+}}\)}. In
  \bibinfo{booktitle}{{\em Hybrid Systems}} {\em (\bibinfo{series}{Lecture
  Notes in Computer Science})}, \bibfield{editor}{\bibinfo{person}{Robert~L.
  Grossman}, \bibinfo{person}{Anil Nerode}, \bibinfo{person}{Anders~P. Ravn},
  {and} \bibinfo{person}{Hans Rischel}} (Eds.), Vol.~\bibinfo{volume}{736}.
  \bibinfo{publisher}{Springer}, \bibinfo{pages}{77--102}.
\newblock
\showISBNx{3-540-57318-6}
\showDOI{%
\url{https://doi.org/10.1007/3-540-57318-6_25}}


\bibitem[\protect\citeauthoryear{Lamport}{Lamport}{1995}]%
        {Lamport95}
\bibfield{author}{\bibinfo{person}{Leslie Lamport}.}
  \bibinfo{year}{1995}\natexlab{}.
\newblock \showarticletitle{How to Write a Proof}.
\newblock \bibinfo{journal}{{\it Amer. Math. Monthly}} \bibinfo{volume}{102},
  \bibinfo{number}{7} (\bibinfo{year}{1995}), \bibinfo{pages}{600--608}.
\newblock
\showURL{%
\url{http://lamport.azurewebsites.net/pubs/lamport-how-to-write.pdf}}


\bibitem[\protect\citeauthoryear{Lamport}{Lamport}{2012}]%
        {Lamport12}
\bibfield{author}{\bibinfo{person}{Leslie Lamport}.}
  \bibinfo{year}{2012}\natexlab{}.
\newblock \showarticletitle{How to Write a 21st Century Proof}.
\newblock \bibinfo{journal}{{\em Journal of Fixed Point Theory and
  Applications\/}} (\bibinfo{year}{2012}).
\newblock
\showDOI{%
\url{https://doi.org/10.1007/s11784-012-0071-6}}


\bibitem[\protect\citeauthoryear{Leavens, Baker, and Ruby}{Leavens
  et~al\mbox{.}}{1999}]%
        {DBLP:books/daglib/p/LeavensBR99}
\bibfield{author}{\bibinfo{person}{Gary~T. Leavens}, \bibinfo{person}{Albert~L.
  Baker}, {and} \bibinfo{person}{Clyde Ruby}.} \bibinfo{year}{1999}\natexlab{}.
\newblock \showarticletitle{{JML:} {A} Notation for Detailed Design}.
\newblock In \bibinfo{booktitle}{{\em Behavioral Specifications of Businesses
  and Systems}}, \bibfield{editor}{\bibinfo{person}{Haim Kilov},
  \bibinfo{person}{Bernhard Rumpe}, {and} \bibinfo{person}{Ian Simmonds}}
  (Eds.). \bibinfo{series}{The Kluwer International Series in Engineering and
  Computer Science}, Vol.~\bibinfo{volume}{523}. \bibinfo{publisher}{Springer},
  \bibinfo{pages}{175--188}.
\newblock
\showISBNx{978-1-4613-7383-4}
\showDOI{%
\url{https://doi.org/10.1007/978-1-4615-5229-1_12}}


\bibitem[\protect\citeauthoryear{Leino}{Leino}{2010}]%
        {DBLP:conf/lpar/Leino10}
\bibfield{author}{\bibinfo{person}{K.~Rustan~M. Leino}.}
  \bibinfo{year}{2010}\natexlab{}.
\newblock \showarticletitle{Dafny: An Automatic Program Verifier for Functional
  Correctness}. In \bibinfo{booktitle}{{\em Logic for Programming, Artificial
  Intelligence, and Reasoning - 16th International Conference, LPAR-16, Dakar,
  Senegal, April 25-May 1, 2010, Revised Selected Papers}} {\em
  (\bibinfo{series}{Lecture Notes in Computer Science})},
  \bibfield{editor}{\bibinfo{person}{Edmund~M. Clarke} {and}
  \bibinfo{person}{Andrei Voronkov}} (Eds.), Vol.~\bibinfo{volume}{6355}.
  \bibinfo{publisher}{Springer}, \bibinfo{pages}{348--370}.
\newblock
\showISBNx{978-3-642-17510-7}
\showDOI{%
\url{https://doi.org/10.1007/978-3-642-17511-4_20}}


\bibitem[\protect\citeauthoryear{Leino, M{\"{u}}ller, and Smans}{Leino
  et~al\mbox{.}}{2009}]%
        {DBLP:conf/fosad/LeinoMS09}
\bibfield{author}{\bibinfo{person}{K.~Rustan~M. Leino}, \bibinfo{person}{Peter
  M{\"{u}}ller}, {and} \bibinfo{person}{Jan Smans}.}
  \bibinfo{year}{2009}\natexlab{}.
\newblock \showarticletitle{Verification of Concurrent Programs with Chalice}.
  In \bibinfo{booktitle}{{\em Foundations of Security Analysis and Design V,
  {FOSAD} 2007/2008/2009 Tutorial Lectures}} {\em (\bibinfo{series}{Lecture
  Notes in Computer Science})}, \bibfield{editor}{\bibinfo{person}{Alessandro
  Aldini}, \bibinfo{person}{Gilles Barthe}, {and} \bibinfo{person}{Roberto
  Gorrieri}} (Eds.), Vol.~\bibinfo{volume}{5705}.
  \bibinfo{publisher}{Springer}, \bibinfo{pages}{195--222}.
\newblock
\showISBNx{978-3-642-03828-0}
\showDOI{%
\url{https://doi.org/10.1007/978-3-642-03829-7_7}}


\bibitem[\protect\citeauthoryear{Leino}{Leino}{2008}]%
        {this-is-boogie-2-2}
\bibfield{author}{\bibinfo{person}{Rustan Leino}.}
  \bibinfo{year}{2008}\natexlab{}.
\newblock \showarticletitle{{This is Boogie 2}}. \bibinfo{publisher}{Microsoft
  Research}.
\newblock
\showURL{%
\url{https://www.microsoft.com/en-us/research/publication/this-is-boogie-2-2/}}


\bibitem[\protect\citeauthoryear{Lochbihler}{Lochbihler}{2007}]%
        {DBLP:journals/afp/Lochbihler07}
\bibfield{author}{\bibinfo{person}{Andreas Lochbihler}.}
  \bibinfo{year}{2007}\natexlab{}.
\newblock \showarticletitle{Jinja with Threads}.
\newblock \bibinfo{journal}{{\em Archive of Formal Proofs\/}}
  \bibinfo{volume}{2007} (\bibinfo{year}{2007}).
\newblock
\showURL{%
\url{https://www.isa-afp.org/entries/JinjaThreads.shtml}}


\bibitem[\protect\citeauthoryear{Loos, Platzer, and Nistor}{Loos
  et~al\mbox{.}}{2011}]%
        {DBLP:conf/fm/LoosPN11}
\bibfield{author}{\bibinfo{person}{Sarah~M. Loos}, \bibinfo{person}{Andr{\'e}
  Platzer}, {and} \bibinfo{person}{Ligia Nistor}.}
  \bibinfo{year}{2011}\natexlab{}.
\newblock \showarticletitle{Adaptive Cruise Control: Hybrid, Distributed, and
  Now Formally Verified}. In \bibinfo{booktitle}{{\em FM}} {\em
  (\bibinfo{series}{LNCS})}, \bibfield{editor}{\bibinfo{person}{Michael Butler}
  {and} \bibinfo{person}{Wolfram Schulte}} (Eds.), Vol.~\bibinfo{volume}{6664}.
  \bibinfo{publisher}{Springer}, \bibinfo{pages}{42--56}.
\newblock
\showDOI{%
\url{https://doi.org/10.1007/978-3-642-21437-0_6}}


\bibitem[\protect\citeauthoryear{Loos, Renshaw, and Platzer}{Loos
  et~al\mbox{.}}{2013}]%
        {DBLP:conf/hybrid/LoosRP13}
\bibfield{author}{\bibinfo{person}{Sarah~M. Loos}, \bibinfo{person}{David~W.
  Renshaw}, {and} \bibinfo{person}{Andr{\'e} Platzer}.}
  \bibinfo{year}{2013}\natexlab{}.
\newblock \showarticletitle{Formal Verification of Distributed Aircraft
  Controllers}. In \bibinfo{booktitle}{{\em Hybrid Systems: Computation and
  Control (part of CPS Week 2013), HSCC'13, Philadelphia, PA, USA, April 8-13,
  2013}}, \bibfield{editor}{\bibinfo{person}{Calin Belta} {and}
  \bibinfo{person}{Franjo Ivancic}} (Eds.). \bibinfo{publisher}{ACM},
  \bibinfo{pages}{125--130}.
\newblock
\showDOI{%
\url{https://doi.org/10.1145/2461328.2461350}}


\bibitem[\protect\citeauthoryear{Malecha and Bengtson}{Malecha and
  Bengtson}{2015}]%
        {malecha2015rtac-coqpl}
\bibfield{author}{\bibinfo{person}{Gregory Malecha} {and}
  \bibinfo{person}{Jesper Bengtson}.} \bibinfo{year}{2015}\natexlab{}.
\newblock \showarticletitle{{Rtac}: A Fully Reflective Tactic Language}. In
  \bibinfo{booktitle}{{\em CoqPL'15}}.
\newblock


\bibitem[\protect\citeauthoryear{Matichuk, Murray, and Wenzel}{Matichuk
  et~al\mbox{.}}{2016}]%
        {Matichuk:2016:EPM:2904234.2904264}
\bibfield{author}{\bibinfo{person}{Daniel Matichuk}, \bibinfo{person}{Toby
  Murray}, {and} \bibinfo{person}{Makarius Wenzel}.}
  \bibinfo{year}{2016}\natexlab{}.
\newblock \showarticletitle{Eisbach: A Proof Method Language for {Isabelle}}.
\newblock \bibinfo{journal}{{\em J. Autom. Reason.\/}} \bibinfo{volume}{56},
  \bibinfo{number}{3} (\bibinfo{date}{March} \bibinfo{year}{2016}),
  \bibinfo{pages}{261--282}.
\newblock
\showISSN{0168-7433}
\showDOI{%
\url{https://doi.org/10.1007/s10817-015-9360-2}}


\bibitem[\protect\citeauthoryear{Mitsch, Ghorbal, and Platzer}{Mitsch
  et~al\mbox{.}}{2013}]%
        {DBLP:conf/rss/MitschGP13}
\bibfield{author}{\bibinfo{person}{Stefan Mitsch}, \bibinfo{person}{Khalil
  Ghorbal}, {and} \bibinfo{person}{Andr{\'{e}} Platzer}.}
  \bibinfo{year}{2013}\natexlab{}.
\newblock \showarticletitle{On Provably Safe Obstacle Avoidance for Autonomous
  Robotic Ground Vehicles}. In \bibinfo{booktitle}{{\em Robotics: Science and
  Systems IX, Technische Universit{\"{a}}t Berlin, Berlin, Germany, June 24 -
  June 28, 2013}}, \bibfield{editor}{\bibinfo{person}{Paul Newman},
  \bibinfo{person}{Dieter Fox}, {and} \bibinfo{person}{David Hsu}} (Eds.).
\newblock
\showISBNx{978-981-07-3937-9}
\showURL{%
\url{http://www.roboticsproceedings.org/rss09/p14.html}}


\bibitem[\protect\citeauthoryear{Mitsch and Platzer}{Mitsch and
  Platzer}{2016}]%
        {DBLP:conf/fide/MitschP16}
\bibfield{author}{\bibinfo{person}{Stefan Mitsch} {and}
  \bibinfo{person}{Andr{\'e} Platzer}.} \bibinfo{year}{2016}\natexlab{}.
\newblock \showarticletitle{The {KeYmaera X} proof {IDE}: Concepts on usability
  in hybrid systems theorem proving}. In \bibinfo{booktitle}{{\em 3rd Workshop
  on Formal Integrated Development Environment}} {\em
  (\bibinfo{series}{EPTCS})}, \bibfield{editor}{\bibinfo{person}{Catherine
  Dubois}, \bibinfo{person}{Dominique Mery}, {and} \bibinfo{person}{Paolo
  Masci}} (Eds.), Vol.~\bibinfo{volume}{240}. \bibinfo{pages}{67--81}.
\newblock
\showDOI{%
\url{https://doi.org/10.4204/EPTCS.240.5}}


\bibitem[\protect\citeauthoryear{M{\"u}ller, Mitsch, and Platzer}{M{\"u}ller
  et~al\mbox{.}}{2015}]%
        {DBLP:conf/itsc/MullerMP15}
\bibfield{author}{\bibinfo{person}{Andreas M{\"u}ller}, \bibinfo{person}{Stefan
  Mitsch}, {and} \bibinfo{person}{Andr{\'e} Platzer}.}
  \bibinfo{year}{2015}\natexlab{}.
\newblock \showarticletitle{Verified Traffic Networks: Component-Based
  Verification of Cyber-Physical Flow Systems}. In \bibinfo{booktitle}{{\em
  ITSC}}. \bibinfo{pages}{757--764}.
\newblock
\showDOI{%
\url{https://doi.org/10.1109/ITSC.2015.128}}


\bibitem[\protect\citeauthoryear{Nipkow}{Nipkow}{2002}]%
        {Nipkow2002}
\bibfield{author}{\bibinfo{person}{Tobias Nipkow}.}
  \bibinfo{year}{2002}\natexlab{}.
\newblock \bibinfo{booktitle}{{\em {Hoare Logics in Isabelle/HOL}}}.
\newblock \bibinfo{publisher}{Springer Netherlands},
  \bibinfo{address}{Dordrecht}, \bibinfo{pages}{341--367}.
\newblock
\showISBNx{978-94-010-0413-8}
\showDOI{%
\url{https://doi.org/10.1007/978-94-010-0413-8_11}}


\bibitem[\protect\citeauthoryear{Nipkow, Wenzel, and Paulson}{Nipkow
  et~al\mbox{.}}{2002}]%
        {Nipkow:2002:IPA:1791547}
\bibfield{author}{\bibinfo{person}{Tobias Nipkow}, \bibinfo{person}{Markus
  Wenzel}, {and} \bibinfo{person}{Lawrence~C. Paulson}.}
  \bibinfo{year}{2002}\natexlab{}.
\newblock \bibinfo{booktitle}{{\em Isabelle/HOL: A Proof Assistant for
  Higher-order Logic}}.
\newblock \bibinfo{publisher}{Springer-Verlag}, \bibinfo{address}{Berlin,
  Heidelberg}.
\newblock
\showISBNx{3-540-43376-7}


\bibitem[\protect\citeauthoryear{Owicki and Gries}{Owicki and Gries}{1976}]%
        {Owicki1976}
\bibfield{author}{\bibinfo{person}{Susan Owicki} {and} \bibinfo{person}{David
  Gries}.} \bibinfo{year}{1976}\natexlab{}.
\newblock \showarticletitle{{An axiomatic proof technique for parallel programs
  I}}.
\newblock \bibinfo{journal}{{\em Acta Informatica\/}} \bibinfo{volume}{6},
  \bibinfo{number}{4} (\bibinfo{date}{dec} \bibinfo{year}{1976}),
  \bibinfo{pages}{319--340}.
\newblock
\showISSN{1432-0525}
\showDOI{%
\url{https://doi.org/10.1007/BF00268134}}


\bibitem[\protect\citeauthoryear{Owicki}{Owicki}{1975}]%
        {DBLP:books/garland/Owicki75}
\bibfield{author}{\bibinfo{person}{Susan~S. Owicki}.}
  \bibinfo{year}{1975}\natexlab{}.
\newblock \bibinfo{booktitle}{{\em Axiomatic Proof Techniques for Parallel
  Programs}}.
\newblock \bibinfo{publisher}{Garland Publishing, New York}.
\newblock
\showISBNx{0-8240-4413-4}


\bibitem[\protect\citeauthoryear{Platzer}{Platzer}{2007a}]%
        {DBLP:conf/tableaux/Platzer07}
\bibfield{author}{\bibinfo{person}{Andr{\'e} Platzer}.}
  \bibinfo{year}{2007}\natexlab{a}.
\newblock \showarticletitle{Differential Dynamic Logic for Verifying Parametric
  Hybrid Systems.}. In \bibinfo{booktitle}{{\em TABLEAUX}} {\em
  (\bibinfo{series}{LNCS})}, \bibfield{editor}{\bibinfo{person}{Nicola
  Olivetti}} (Ed.), Vol.~\bibinfo{volume}{4548}. \bibinfo{publisher}{Springer},
  \bibinfo{pages}{216--232}.
\newblock
\showISBNx{978-3-540-73098-9}
\showDOI{%
\url{https://doi.org/10.1007/978-3-540-73099-6_17}}


\bibitem[\protect\citeauthoryear{Platzer}{Platzer}{2007b}]%
        {DBLP:journals/entcs/Platzer07}
\bibfield{author}{\bibinfo{person}{Andr{\'e} Platzer}.}
  \bibinfo{year}{2007}\natexlab{b}.
\newblock \showarticletitle{Towards a Hybrid Dynamic Logic for Hybrid Dynamic
  Systems}, In \bibinfo{booktitle}{International Workshop on Hybrid Logic,
  HyLo'06, Seattle, USA, Proceedings},
  \bibfield{editor}{\bibinfo{person}{Patrick Blackburn},
  \bibinfo{person}{Thomas Bolander}, \bibinfo{person}{Torben Bra\"{u}ner},
  \bibinfo{person}{Valeria de~Paiva}, {and} \bibinfo{person}{J{\o}rgen
  Villadsen}} (Eds.).
\newblock \bibinfo{journal}{{\em Electr. Notes Theor. Comput. Sci.\/}}
  \bibinfo{volume}{174}, \bibinfo{number}{6}, \bibinfo{pages}{63--77}.
\newblock
\showISSN{1571-0661}
\showDOI{%
\url{https://doi.org/10.1016/j.entcs.2006.11.026}}


\bibitem[\protect\citeauthoryear{Platzer}{Platzer}{2008}]%
        {DBLP:journals/jar/Platzer08}
\bibfield{author}{\bibinfo{person}{Andr{\'{e}} Platzer}.}
  \bibinfo{year}{2008}\natexlab{}.
\newblock \showarticletitle{Differential Dynamic Logic for Hybrid Systems}.
\newblock \bibinfo{journal}{{\em J. Autom. Reasoning\/}} \bibinfo{volume}{41},
  \bibinfo{number}{2} (\bibinfo{year}{2008}), \bibinfo{pages}{143--189}.
\newblock
\showDOI{%
\url{https://doi.org/10.1007/s10817-008-9103-8}}


\bibitem[\protect\citeauthoryear{Platzer}{Platzer}{2010}]%
        {DBLP:journals/logcom/Platzer10}
\bibfield{author}{\bibinfo{person}{Andr{\'{e}} Platzer}.}
  \bibinfo{year}{2010}\natexlab{}.
\newblock \showarticletitle{Differential-algebraic Dynamic Logic for
  Differential-algebraic Programs}.
\newblock \bibinfo{journal}{{\em J. Log. Comput.\/}} \bibinfo{volume}{20},
  \bibinfo{number}{1} (\bibinfo{year}{2010}), \bibinfo{pages}{309--352}.
\newblock
\showDOI{%
\url{https://doi.org/10.1093/logcom/exn070}}


\bibitem[\protect\citeauthoryear{Platzer}{Platzer}{2011}]%
        {DBLP:journals/corr/abs-1104-1987}
\bibfield{author}{\bibinfo{person}{Andr{\'{e}} Platzer}.}
  \bibinfo{year}{2011}\natexlab{}.
\newblock \showarticletitle{The Structure of Differential Invariants and
  Differential Cut Elimination}.
\newblock \bibinfo{journal}{{\em Logical Methods in Computer Science\/}}
  \bibinfo{volume}{8}, \bibinfo{number}{4} (\bibinfo{year}{2011}).
\newblock
\showDOI{%
\url{https://doi.org/10.2168/LMCS-8(4:16)2012}}


\bibitem[\protect\citeauthoryear{Platzer}{Platzer}{2012a}]%
        {DBLP:journals/corr/abs-1206-3357}
\bibfield{author}{\bibinfo{person}{Andr{\'{e}} Platzer}.}
  \bibinfo{year}{2012}\natexlab{a}.
\newblock \showarticletitle{A Complete Axiomatization of Quantified
  Differential Dynamic Logic for Distributed Hybrid Systems}.
\newblock \bibinfo{journal}{{\em Logical Methods in Computer Science\/}}
  \bibinfo{volume}{8}, \bibinfo{number}{4} (\bibinfo{year}{2012}).
\newblock
\showDOI{%
\url{https://doi.org/10.2168/LMCS-8(4:17)2012}}


\bibitem[\protect\citeauthoryear{Platzer}{Platzer}{2012b}]%
        {DBLP:conf/lics/Platzer12}
\bibfield{author}{\bibinfo{person}{Andr{\'{e}} Platzer}.}
  \bibinfo{year}{2012}\natexlab{b}.
\newblock \showarticletitle{Logics of Dynamical Systems}. In
  \bibinfo{booktitle}{{\em Proceedings of the 27th Annual {IEEE} Symposium on
  Logic in Computer Science, {LICS} 2012, Dubrovnik, Croatia, June 25-28,
  2012}}. \bibinfo{publisher}{{IEEE} Computer Society},
  \bibinfo{pages}{13--24}.
\newblock
\showISBNx{978-1-4673-2263-8}
\showDOI{%
\url{https://doi.org/10.1109/LICS.2012.13}}


\bibitem[\protect\citeauthoryear{Platzer}{Platzer}{2015}]%
        {DBLP:conf/cade/Platzer15}
\bibfield{author}{\bibinfo{person}{Andr{\'e} Platzer}.}
  \bibinfo{year}{2015}\natexlab{}.
\newblock \showarticletitle{A Uniform Substitution Calculus for Differential
  Dynamic Logic}. In \bibinfo{booktitle}{{\em CADE}} {\em
  (\bibinfo{series}{LNCS})}, \bibfield{editor}{\bibinfo{person}{Amy~P. Felty}
  {and} \bibinfo{person}{Aart Middeldorp}} (Eds.), Vol.~\bibinfo{volume}{9195}.
  \bibinfo{publisher}{Springer}, \bibinfo{pages}{467--481}.
\newblock
\showDOI{%
\url{https://doi.org/10.1007/978-3-319-21401-6_32}}
\showeprint{1503.01981}


\bibitem[\protect\citeauthoryear{Platzer}{Platzer}{2016}]%
        {DBLP:journals/jar/Platzer16}
\bibfield{author}{\bibinfo{person}{Andr{\'e} Platzer}.}
  \bibinfo{year}{2016}\natexlab{}.
\newblock \showarticletitle{A Complete Uniform Substitution Calculus for
  Differential Dynamic Logic}.
\newblock \bibinfo{journal}{{\em J. Autom. Reas.\/}} (\bibinfo{year}{2016}).
\newblock
\showDOI{%
\url{https://doi.org/10.1007/s10817-016-9385-1}}


\bibitem[\protect\citeauthoryear{Platzer and Clarke}{Platzer and
  Clarke}{2009}]%
        {DBLP:conf/fm/PlatzerC09}
\bibfield{author}{\bibinfo{person}{Andr{\'e} Platzer} {and}
  \bibinfo{person}{Edmund~M. Clarke}.} \bibinfo{year}{2009}\natexlab{}.
\newblock \showarticletitle{Formal Verification of Curved Flight Collision
  Avoidance Maneuvers: A Case Study}. In \bibinfo{booktitle}{{\em FM}} {\em
  (\bibinfo{series}{LNCS})}, \bibfield{editor}{\bibinfo{person}{Ana Cavalcanti}
  {and} \bibinfo{person}{Dennis Dams}} (Eds.), Vol.~\bibinfo{volume}{5850}.
  \bibinfo{publisher}{Springer}, \bibinfo{pages}{547--562}.
\newblock
\showDOI{%
\url{https://doi.org/10.1007/978-3-642-05089-3_35}}


\bibitem[\protect\citeauthoryear{Platzer and Quesel}{Platzer and
  Quesel}{2008}]%
        {DBLP:conf/cade/PlatzerQ08}
\bibfield{author}{\bibinfo{person}{Andr{\'e} Platzer} {and}
  \bibinfo{person}{Jan-David Quesel}.} \bibinfo{year}{2008}\natexlab{}.
\newblock \showarticletitle{{KeYmaera}: A Hybrid Theorem Prover for Hybrid
  Systems.}. In \bibinfo{booktitle}{{\em IJCAR}} {\em
  (\bibinfo{series}{LNCS})}, \bibfield{editor}{\bibinfo{person}{Alessandro
  Armando}, \bibinfo{person}{Peter Baumgartner}, {and} \bibinfo{person}{Gilles
  Dowek}} (Eds.), Vol.~\bibinfo{volume}{5195}. \bibinfo{publisher}{Springer},
  \bibinfo{pages}{171--178}.
\newblock
\showISBNx{978-3-540-71069-1}
\showISSN{0302-9743}
\showDOI{%
\url{https://doi.org/10.1007/978-3-540-71070-7_15}}


\bibitem[\protect\citeauthoryear{Platzer and Quesel}{Platzer and
  Quesel}{2009}]%
        {DBLP:conf/icfem/PlatzerQ09}
\bibfield{author}{\bibinfo{person}{Andr{\'e} Platzer} {and}
  \bibinfo{person}{Jan-David Quesel}.} \bibinfo{year}{2009}\natexlab{}.
\newblock \showarticletitle{{European Train Control System}: A Case Study in
  Formal Verification}. In \bibinfo{booktitle}{{\em ICFEM}} {\em
  (\bibinfo{series}{LNCS})}, \bibfield{editor}{\bibinfo{person}{Karin Breitman}
  {and} \bibinfo{person}{Ana Cavalcanti}} (Eds.), Vol.~\bibinfo{volume}{5885}.
  \bibinfo{publisher}{Springer}, \bibinfo{pages}{246--265}.
\newblock
\showDOI{%
\url{https://doi.org/10.1007/978-3-642-10373-5_13}}


\bibitem[\protect\citeauthoryear{Pratt}{Pratt}{1976}]%
        {DBLP:conf/focs/Pratt76}
\bibfield{author}{\bibinfo{person}{Vaughan~R. Pratt}.}
  \bibinfo{year}{1976}\natexlab{}.
\newblock \showarticletitle{Semantical Considerations on Floyd-Hoare Logic}. In
  \bibinfo{booktitle}{{\em 17th Annual Symposium on Foundations of Computer
  Science, Houston, Texas, USA, 25-27 October 1976}}.
  \bibinfo{publisher}{{IEEE} Computer Society}, \bibinfo{pages}{109--121}.
\newblock
\showDOI{%
\url{https://doi.org/10.1109/SFCS.1976.27}}


\bibitem[\protect\citeauthoryear{Rudnicki}{Rudnicki}{1987}]%
        {DBLP:journals/jar/Rudnicki87}
\bibfield{author}{\bibinfo{person}{Piotr Rudnicki}.}
  \bibinfo{year}{1987}\natexlab{}.
\newblock \showarticletitle{Obvious Inferences}.
\newblock \bibinfo{journal}{{\em J. Autom. Reasoning\/}} \bibinfo{volume}{3},
  \bibinfo{number}{4} (\bibinfo{year}{1987}), \bibinfo{pages}{383--393}.
\newblock
\showDOI{%
\url{https://doi.org/10.1007/BF00247436}}


\bibitem[\protect\citeauthoryear{Stampoulis and Shao}{Stampoulis and
  Shao}{2010}]%
        {DBLP:conf/icfp/StampoulisS10}
\bibfield{author}{\bibinfo{person}{Antonis Stampoulis} {and}
  \bibinfo{person}{Zhong Shao}.} \bibinfo{year}{2010}\natexlab{}.
\newblock \showarticletitle{VeriML: typed computation of logical terms inside a
  language with effects}. In \bibinfo{booktitle}{{\em Proceeding of the 15th
  {ACM} {SIGPLAN} {I}nternational {C}onference on {F}unctional {P}rogramming,
  {ICFP} 2010, Baltimore, Maryland, USA, September 27-29, 2010}},
  \bibfield{editor}{\bibinfo{person}{Paul Hudak} {and}
  \bibinfo{person}{Stephanie Weirich}} (Eds.). \bibinfo{publisher}{{ACM}},
  \bibinfo{pages}{333--344}.
\newblock
\showISBNx{978-1-60558-794-3}
\showDOI{%
\url{https://doi.org/10.1145/1863543.1863591}}


\bibitem[\protect\citeauthoryear{Stampoulis and Shao}{Stampoulis and
  Shao}{2012}]%
        {DBLP:conf/popl/StampoulisS12}
\bibfield{author}{\bibinfo{person}{Antonis Stampoulis} {and}
  \bibinfo{person}{Zhong Shao}.} \bibinfo{year}{2012}\natexlab{}.
\newblock \showarticletitle{Static and user-extensible proof checking}. In
  \bibinfo{booktitle}{{\em Proceedings of the 39th {ACM} {SIGPLAN-SIGACT}
  Symposium on Principles of Programming Languages, {POPL} 2012, Philadelphia,
  Pennsylvania, USA, January 22-28, 2012}},
  \bibfield{editor}{\bibinfo{person}{John Field} {and} \bibinfo{person}{Michael
  Hicks}} (Eds.). \bibinfo{publisher}{{ACM}}, \bibinfo{pages}{273--284}.
\newblock
\showISBNx{978-1-4503-1083-3}
\showDOI{%
\url{https://doi.org/10.1145/2103656.2103690}}


\bibitem[\protect\citeauthoryear{Suenaga and Hasuo}{Suenaga and Hasuo}{2011}]%
        {DBLP:conf/icalp/SuenagaH11}
\bibfield{author}{\bibinfo{person}{Kohei Suenaga} {and} \bibinfo{person}{Ichiro
  Hasuo}.} \bibinfo{year}{2011}\natexlab{}.
\newblock \showarticletitle{Programming with Infinitesimals: {A} While-Language
  for Hybrid System Modeling}. In \bibinfo{booktitle}{{\em Automata, Languages
  and Programming - 38th International Colloquium, {ICALP} 2011, Zurich,
  Switzerland, July 4-8, 2011, Proceedings, Part {II}}} {\em
  (\bibinfo{series}{Lecture Notes in Computer Science})},
  \bibfield{editor}{\bibinfo{person}{Luca Aceto}, \bibinfo{person}{Monika
  Henzinger}, {and} \bibinfo{person}{Jir{\'{\i}} Sgall}} (Eds.),
  Vol.~\bibinfo{volume}{6756}. \bibinfo{publisher}{Springer},
  \bibinfo{pages}{392--403}.
\newblock
\showISBNx{978-3-642-22011-1}
\showDOI{%
\url{https://doi.org/10.1007/978-3-642-22012-8_31}}


\bibitem[\protect\citeauthoryear{Syme}{Syme}{1997}]%
        {Syme1997DECLAREAP}
\bibfield{author}{\bibinfo{person}{Donald Syme}.}
  \bibinfo{year}{1997}\natexlab{}.
\newblock \showarticletitle{{DECLARE}: A Prototype Declarative Proof System for
  Higher Order Logic}.
\newblock


\bibitem[\protect\citeauthoryear{Team}{Team}{2017}]%
        {COQ}
\bibfield{author}{\bibinfo{person}{The Coq~Development Team}.}
  \bibinfo{year}{2017}\natexlab{}.
\newblock \bibinfo{title}{{C}oq {P}roof {A}ssistant}.
\newblock   (\bibinfo{year}{2017}).
\newblock
\showURL{%
\url{http://coq.inria.fr/}}
\newblock
\shownote{Accessed: 2017-05-25.}


\bibitem[\protect\citeauthoryear{Traut and Noschinski}{Traut and
  Noschinski}{2014}]%
        {traut2014pattern}
\bibfield{author}{\bibinfo{person}{Christoph Traut} {and} \bibinfo{person}{Lars
  Noschinski}.} \bibinfo{year}{2014}\natexlab{}.
\newblock \showarticletitle{Pattern-based Subterm Selection in {I}sabelle}. In
  \bibinfo{booktitle}{{\em Proceedings of Isabelle Workshop 2014}}.
\newblock


\bibitem[\protect\citeauthoryear{Wenzel}{Wenzel}{1999}]%
        {DBLP:conf/tphol/Wenzel99}
\bibfield{author}{\bibinfo{person}{Markus Wenzel}.}
  \bibinfo{year}{1999}\natexlab{}.
\newblock \showarticletitle{Isar - {A} Generic Interpretative Approach to
  Readable Formal Proof Documents}. In \bibinfo{booktitle}{{\em Theorem Proving
  in Higher Order Logics, 12th International Conference, TPHOLs'99, Nice,
  France, September, 1999, Proceedings}} {\em (\bibinfo{series}{Lecture Notes
  in Computer Science})}, \bibfield{editor}{\bibinfo{person}{Yves Bertot},
  \bibinfo{person}{Gilles Dowek}, \bibinfo{person}{Andr{\'{e}} Hirschowitz},
  \bibinfo{person}{Christine Paulin{-}Mohring}, {and} \bibinfo{person}{Laurent
  Th{\'{e}}ry}} (Eds.), Vol.~\bibinfo{volume}{1690}.
  \bibinfo{publisher}{Springer}, \bibinfo{pages}{167--184}.
\newblock
\showISBNx{3-540-66463-7}
\showDOI{%
\url{https://doi.org/10.1007/3-540-48256-3_12}}


\bibitem[\protect\citeauthoryear{Wenzel}{Wenzel}{2006}]%
        {DBLP:conf/mkm/Wenzel06}
\bibfield{author}{\bibinfo{person}{Makarius Wenzel}.}
  \bibinfo{year}{2006}\natexlab{}.
\newblock \showarticletitle{Structured Induction Proofs in Isabelle/Isar}. In
  \bibinfo{booktitle}{{\em Mathematical Knowledge Management, 5th International
  Conference, {MKM} 2006, Wokingham, UK, August 11-12, 2006, Proceedings}} {\em
  (\bibinfo{series}{Lecture Notes in Computer Science})},
  \bibfield{editor}{\bibinfo{person}{Jonathan~M. Borwein} {and}
  \bibinfo{person}{William~M. Farmer}} (Eds.), Vol.~\bibinfo{volume}{4108}.
  \bibinfo{publisher}{Springer}, \bibinfo{pages}{17--30}.
\newblock
\showISBNx{3-540-37104-4}
\showDOI{%
\url{https://doi.org/10.1007/11812289_3}}


\bibitem[\protect\citeauthoryear{Wenzel}{Wenzel}{2007}]%
        {Wenzel07isabelle/isar}
\bibfield{author}{\bibinfo{person}{Makarius Wenzel}.}
  \bibinfo{year}{2007}\natexlab{}.
\newblock \showarticletitle{{I}sabelle/{I}sar -- a generic framework for
  human-readable proof documents}. In \bibinfo{booktitle}{{\em UNIVERSITY OF
  BIA\L{}YSTOK}}.
\newblock


\bibitem[\protect\citeauthoryear{Wenzel and Wiedijk}{Wenzel and
  Wiedijk}{2002}]%
        {Wenzel2002}
\bibfield{author}{\bibinfo{person}{Markus Wenzel} {and} \bibinfo{person}{Freek
  Wiedijk}.} \bibinfo{year}{2002}\natexlab{}.
\newblock \showarticletitle{A Comparison of Mizar and Isar}.
\newblock \bibinfo{journal}{{\em J. Autom. Reasoning\/}} \bibinfo{volume}{29},
  \bibinfo{number}{3-4} (\bibinfo{year}{2002}), \bibinfo{pages}{389--411}.
\newblock
\showDOI{%
\url{https://doi.org/10.1023/A:1021935419355}}


\bibitem[\protect\citeauthoryear{Wiedijk}{Wiedijk}{2001}]%
        {DBLP:conf/tphol/Wiedijk01}
\bibfield{author}{\bibinfo{person}{Freek Wiedijk}.}
  \bibinfo{year}{2001}\natexlab{}.
\newblock \showarticletitle{Mizar Light for {HOL} Light}. In
  \bibinfo{booktitle}{{\em Theorem Proving in Higher Order Logics, 14th
  International Conference, TPHOLs 2001, Edinburgh, Scotland, UK, September
  3-6, 2001, Proceedings}} {\em (\bibinfo{series}{Lecture Notes in Computer
  Science})}, \bibfield{editor}{\bibinfo{person}{Richard~J. Boulton} {and}
  \bibinfo{person}{Paul~B. Jackson}} (Eds.), Vol.~\bibinfo{volume}{2152}.
  \bibinfo{publisher}{Springer}, \bibinfo{pages}{378--394}.
\newblock
\showISBNx{3-540-42525-X}
\showDOI{%
\url{https://doi.org/10.1007/3-540-44755-5_26}}


\bibitem[\protect\citeauthoryear{Ziliani, Dreyer, Krishnaswami, Nanevski, and
  Vafeiadis}{Ziliani et~al\mbox{.}}{2013}]%
        {Ziliani:2013:MMT:2544174.2500579}
\bibfield{author}{\bibinfo{person}{Beta Ziliani}, \bibinfo{person}{Derek
  Dreyer}, \bibinfo{person}{Neelakantan~R. Krishnaswami},
  \bibinfo{person}{Aleksandar Nanevski}, {and} \bibinfo{person}{Viktor
  Vafeiadis}.} \bibinfo{year}{2013}\natexlab{}.
\newblock \showarticletitle{Mtac: A Monad for Typed Tactic Programming in
  {C}oq}.
\newblock \bibinfo{journal}{{\em SIGPLAN Not.\/}} \bibinfo{volume}{48},
  \bibinfo{number}{9} (\bibinfo{date}{Sept.} \bibinfo{year}{2013}),
  \bibinfo{pages}{87--100}.
\newblock
\showISSN{0362-1340}
\showDOI{%
\url{https://doi.org/10.1145/2544174.2500579}}


\end{thebibliography}

\newpage
\appendix 

\section{Sequent Calculus for \dL}
\newcommand{\seqq}[2]{\Gamma{#1}\vdash{#2}\Delta}
\label{app:sequent}
\paragraph{First-Order Rules}
\begin{center}
  \begin{calculuscollections}{\columnwidth}
    \begin{calculus}
      \cinferenceRule[negR|$\neg{R}$]{}
      {\linferenceRule[formula]
        {\seqq{,\phi}{}}
        {\seqq{}{\neg\phi,}}
      }{}
      \cinferenceRule[orR|$\vee{R}$]{}
      {\linferenceRule[formula]
        {\seqq{}{\phi,\psi,}}
        {\seqq{}{\phi\lor\psi,}}
      }{}
      \cinferenceRule[andR|$\wedge{R}$]{}
      {\linferenceRule[formula]
        {\seqq{}{\phi,} & \seqq{}{\psi,}}
        {\seqq{}{\phi\land\psi,}}
      }{}
      \cinferenceRule[impR|$\limply{R}$]{}
      {\linferenceRule[formula]
        {\seqq{,\phi}{\psi,}}
        {\seqq{}{\phi\limply\psi,}}
      }{}
      \cinferenceRule[id|id]{}
      {\linferenceRule[formula]
        {*}
        {\seqq{,\phi}{\phi,}}
      }{}
    \end{calculus}
    \begin{calculus}
      \cinferenceRule[negL|$\neg{L}$]{}
      {\linferenceRule[formula]
        {\seqq{}{\phi,}}
        {\seqq{,\neg{\phi}}{}}
      }{}
      \cinferenceRule[orL|$\vee{L}$]{}
      {\linferenceRule[formula]
        {\seqq{,\phi}{} & \seqq{,\psi}{}}
        {\seqq{,\phi\vee\psi}{}}
      }{}
      \cinferenceRule[andL|$\wedge{L}$]{}
      {\linferenceRule[formula]
        {\seqq{,\phi}{} & \seqq{,\psi}{}}
        {\seqq{,\phi\land\psi}{}}
      }{}
      \cinferenceRule[impL|$\limply{L}$]{}
      {\linferenceRule[formula]
        {\seqq{}{\phi,} & \seqq{,\psi}{}}
        {\seqq{,\phi\limply\psi}{}}
      }{}
      \cinferenceRule[cut|cut]{}
      {\linferenceRule[formula]
        {\seqq{}{\phi,} & \seqq{,\phi}{}}
        {\seqq{}{}}
      }{}
    \end{calculus}
    \begin{calculus}
      \cinferenceRule[allR|$\forall{R}$]{}
      {\linferenceRule[formula]
        {\seqq{}{\phi(y),}}
        {\seqq{}{\forall~x~\phi(x),}}
      }{Where $y\notin\freevars{\G}\cup\freevars{\Delta}$}
      \cinferenceRule[existsR|$\exists{R}$]{}
      {\linferenceRule[formula]
        {\seqq{}{\phi(\theta),}}
        {\seqq{}{\exists~x~\phi(x),}}
      }{}
      \cinferenceRule[existsL|$\exists{L}$]{}
      {\linferenceRule[formula]
        {\seqq{,\phi(y)}{}}
        {\seqq{,\exists~x~\phi(x)}{}}
      }{Where $y\in\freevars{\G}\cup\freevars{\Delta}$}
      \cinferenceRule[allL|$\forall{L}$]{}
      {\linferenceRule[formula]
        {\seqq{}{\phi(\theta),}}
        {\seqq{}{\exists~x~\phi(x),}}
      }{}
      \cinferenceRule[QE|QE]{}
      {\linferenceRule[formula]
        {*}
        {\seqq{}{}}
      }{$\bigwedge_{\phi\in\G}\phi\limply\bigvee_{\psi\in\Delta}\psi$ valid in $\folr{}$}
    \end{calculus}
  \end{calculuscollections}
\end{center}
\paragraph{Symmetric Program Rules}
\begin{center}
  \begin{calculuscollections}{\columnwidth}
    \begin{calculus}
      \cinferenceRule[boxode|${[']}$]{}
      {\linferenceRule[formula]
        {\forall~t\geq{0}~((\forall{0}{\leq}{s}{\leq}{t}\dbox{\humod{x}{y(t)}}{\psi}{})\limply{\dbox{\humod{x}{y(t)}}{\psi}})}
        {\dbox{\pevolvein{x'=\theta}{\psi}}{\phi}}
      }{}
     \cinferenceRule[diaode|$\langle{'}\rangle$]{}
      {\linferenceRule[formula]
        {\exists~t\geq{0}~((\forall{0}\leq{s}\leq{t}\dbox{\humod{x}{y(t)}}{\psi}{})\land{\dbox{\humod{x}{y(t)}}{\psi}})}
        {\ddiamond{\pevolvein{x'=\theta}{\psi}}{\phi}}
      }{}
    \end{calculus}
  \end{calculuscollections}
\end{center}
\begin{center}
  \begin{calculuscollections}{\columnwidth}
    \begin{calculus}
%      \cinfa
      \cinferenceRule[diacomp|$\langle;\rangle$]{}
      {\linferenceRule[formula]
        {\ddiamond{\alpha}{\ddiamond{\beta}{\phi}}}
        {\ddiamond{\alpha;\beta}{\phi}}
      }{}
      \cinferenceRule[dialoop|$\langle{*}\rangle$]{}
      {\linferenceRule[formula]
        {\phi\lor\ddiamond{\alpha}{\ddiamond{\alpha^*}{\phi}}}
        {\ddiamond{\alpha^*}{\phi}}
      }{}
      \cinferenceRule[diaasgn|$\langle{x}:={\theta}\rangle$]{}
      {\linferenceRule[formula]
        {\subst[\phi]{x}{\theta}}
        {\ddiamond{\humod{x}{\theta}}{\phi}}
      }{Where $\freevars{\theta}\cap\boundvars{\phi}=\emptyset$}
      \cinferenceRule[anyboxasgneq|${[:={*}]}$]{}
      {\linferenceRule[formula]
        {\forall~y~\phi(y)}
        {\dbox{\humod{x}{\theta}}{\phi(x)}}
      }{}
    \end{calculus}
    \begin{calculus}
      \cinferenceRule[boxcomp|${[;]}$]{}
      {\linferenceRule[formula]
        {\dbox{\alpha}{\dbox{\beta}{\phi}}}
        {\dbox{\alpha;\beta}{\phi}}
      }{}
      \cinferenceRule[boxloop|${[*]}$]{}
      {\linferenceRule[formula]
        {\phi\land\dbox{\alpha}{\dbox{\alpha^*}{\phi}}}
        {\dbox{\alpha^*}{\phi}}
      }{}
      \cinferenceRule[boxassign|${[:=]}$]{}
      {\linferenceRule[formula]
        {\subst[\phi]{x}{\theta}}
        {\dbox{\humod{x}{\theta}}{\phi}}
      }{Where $\freevars{\theta}\cap\boundvars{\phi}=\emptyset$}
      \cinferenceRule[anydiaasgneq|$\langle{:=}{*}\rangle$]{}
      {\linferenceRule[formula]
        {\exists~y~\phi(y)}
        {\ddiamond{\humod{x}{\theta}}{\phi(x)}}
      }{}
    \end{calculus}
    \begin{calculus}
      \cinferenceRule[diachoice|$\langle\cup\rangle$]{}
      {\linferenceRule[formula]
        {\ddiamond{\alpha}\phi\lor\ddiamond{\beta}{\phi}}
        {\ddiamond{\alpha\cup\beta}{\phi}}
      }{}
%      \cinferenceRule[diachoice|$\langle\cup\rangle$]{}
 %     {\linferenceRule[formula]
%        {}
%        {}
%      }{}
      \cinferenceRule[diatest|$\langle?\rangle$]{}
      {\linferenceRule[formula]
        {\psi\land\phi}
        {\ddiamond{?(\psi)}{\phi}}
      }{}
      \cinferenceRule[diaasgneq|$\langle:=\rangle=$]{}
      {\linferenceRule[formula]
        {\exists~y~(y=\theta\land\phi(y))}
        {\ddiamond{\humod{x}{\theta}}{\phi(x)}}
      }{}   
      \cinferenceRule[boxasgneq|${[:=]=}$]{}
      {\linferenceRule[formula]
        {\forall~y~(y=\theta\limply\phi(y))}
        {\dbox{\humod{x}{\theta}}{\phi(x)}}
      }{}
\end{calculus} 
\begin{calculus} 
      \cinferenceRule[boxchoice|${[\cup]}$]{}
      {\linferenceRule[formula]
        {\dbox{\alpha}{\phi}\land\dbox{\beta}{\phi}}
        {\dbox{\alpha\cup\beta}{\phi}}
      }{}
      \cinferenceRule[boxtest|${[?]}$]{}
      {\linferenceRule[formula]
        {\psi\limply\phi}
        {\dbox{\psi}{\phi}}
      }{}
    \end{calculus}
  \end{calculuscollections}
\end{center}\paragraph{Asymetric Program Rules}
\newcommand{\gconst}{\G_{\tt{const}}}
\begin{center}
%\begin{center}
  \begin{calculuscollections}{\columnwidth}
    \begin{calculus}
     \cinferenceRule[diaInd|con]{}
      {\linferenceRule[formula]
        {\gconst,v>0,\varphi(v)\vdash\langle\alpha\rangle\varphi(v-1)}
        {\seqq{,\exists{v}~\varphi(v)}{\langle\alpha^*\rangle\exists{v\leq{0}}\varphi(v),}}
      }{Where $\gconst = \{ \phi\in\G~|~\freevars{\phi}\cap\boundvars{\alpha}=\emptyset\}$ }
      \cinferenceRule[boxMon|${[M]}$]{}
      {\linferenceRule[formula]
        {\phi\vdash\psi}
        {[\alpha]\phi\vdash[\alpha]\psi}
      }{}
     \end{calculus}
    \begin{calculus}
      \cinferenceRule[boxInd|I]{}
      {\linferenceRule[formula]
        {\gconst,J\vdash[\alpha]J}
        {\seqq{,J}{[\alpha^*]J}}
      }{}
      \cinferenceRule[diaMon|$\langle{M}\rangle$]{}
      {\linferenceRule[formula]
        {\phi\vdash\psi}
        {\langle\alpha\rangle\phi\vdash\langle\alpha\rangle\psi}
      }{}
    \end{calculus}
  \end{calculuscollections}
\end{center}

\newpage
\section{Denotational Semantics}
\label{app:denotational}
\newcommand{\tsem}[2]{{\ivaluation{}{#1}}#2}
\newcommand{\psem}[3]{(#2,#3)\in{\ivaluation{}{#1}}}
\newcommand{\fsem}[1]{\ivaluation{}{#1}}
The denotational semantics of \dL are given as interpretation functions $\tsem{\theta}{\omega}$ for terms, $\fsem{\phi}$ for formulas, and $\fsem{\alpha}$ for programs. 
Terms denote reals, formulas denote sets of states, and programs denote transition relations (i.e. sets of pairs of states).

\paragraph{Term Semantics}
\begin{align*}
  \tsem{x}{\omega} &= \omega(x)\\
  \tsem{q}{\omega} &= q\\
  \tsem{\theta_1+\theta_2}{\omega} &= \tsem{\theta_1}{\omega}+\tsem{\theta_2}{\omega}\\
  \tsem{\theta_1\cdot{}\theta_2}{\omega} &= \tsem{\theta_1}{\omega}\cdot{}\tsem{\theta_2}{\omega}\\
  \tsem{\theta_1/\theta_2}{\omega} &= \tsem{\theta_1}{\omega}~/~\tsem{\theta_2}{\omega}\\
  \tsem{\theta^q}{\omega} &= \left(\tsem{\theta_1}{\omega}\right)^q &\text{for $q\in\mathbb{Q}$}\\
  \tsem{\der{\theta}}{\omega} &=\sum\limits_{x\in\allvars}\frac{(\partial{\tsem{\theta}{}})(\omega)}{\partial{x}}\cdot\omega(x')
\end{align*}
The differential $\der{\theta}$ of a term $\theta$ is denoted by the \emph{total derivative}, the sum of all partial derivatives.

\paragraph{Formula Semantics}
\begin{align*}
  \fsem{\phi\land\psi} &= \fsem{\phi} \cap \fsem{\psi}\\
  \fsem{\phi\lor\psi} &= \fsem{\phi} \cup \fsem{\psi}\\
  \fsem{\neg\phi} &= \fsem{\phi}^C\\
  \fsem{\forall~x~\phi} &= \{~\omega~|~\mforall~r\in\mathbb{R}~\subst[\omega]{x}{r}\in\fsem{\phi}~\}\\
  \fsem{\exists~x~\phi} &= \{~\omega~|~\mexists~r\in\mathbb{R}~\subst[\omega]{x}{r}\in\fsem{\phi}~\}\\
  \fsem{\dbox{\alpha}{\phi}} &= \{\omega~|~\mforall~\nu~\psem{\alpha}{\omega}{\nu}~\text{implies}~\nu\in\fsem{\phi}~\}\\
  \fsem{\ddiamond{\alpha}{\phi}} &= \{\omega~|~\mexists~\nu~\psem{\alpha}{\omega}{\nu}~\text{and}~\nu\in\fsem{\phi}~\}\\
  \fsem{\theta_1\sim\theta_2} &= \{\omega~|~ \tsem{\theta_1}{\omega}\sim\tsem{\theta_2}{\omega}~\}\\
 &\text{for $\sim\ \in\{<,\leq,=,\geq,>,\neq\}$}\\
\end{align*}

\paragraph{Program Semantics}
%~\text{}~
\begin{align*}
\psem{\humod{x}{\theta}}{\omega}{\nu} &~\text{iff}~ \subst[\omega]{x}{\tsem{\theta}{\omega}}\\
\psem{\prandom{x}}{\omega}{\nu} &~\text{iff}~ \exists~r\in\mathbb{R}~\subst[\omega]{x}{r} = \nu\\
\psem{?(\phi)}{\omega}{\nu}     &~\text{iff}~ \omega = \nu~\text{and}~\omega\in\fsem{\phi}\\
\psem{\alpha;\beta}{\omega}{\nu}&~\text{iff}~ \exists~\mu~\psem{\alpha}{\omega}{\mu}~\text{and}~\psem{\beta}{\mu}{\nu}\\
\psem{\alpha\cup\beta}{\omega}{\nu} &~\text{iff}~ \psem{\alpha}{\omega}{\nu}~\text{or}~\psem{\beta}{\omega}{\nu}\\
\psem{\alpha^*}{\omega}{\nu} &~\text{iff}~ \psem{\alpha}{\omega}{\nu}^*\ \ \ \ \text{(i.e. transitive, reflexive closure of $\interp{\alpha}$)}\\
\psem{\pevolvein{x'=\theta}{\psi}}{\omega}{\nu} &~\text{iff}~ \exists~t\in\mathbb{R}_{\geq{0}}~ \left(\nu=\varphi(t)\right) \land \forall~s\in[0,t]~ \varphi(s)\in\fsem{\psi} \\
&\text{where $\varphi$ is the unique solution to $x'=\theta$ with $\varphi(0)=\omega$ }
\end{align*}

\newpage
\section{Complete Proof Rules of Kaisar}
\label{app:diamond}
\small

\paragraph{Pattern-Matching}
\begin{center}
  \begin{calculuscollections}{\columnwidth}
    \begin{calculus}
      \cinferenceRule[matchOp|op]{}
      {\linferenceRule[formula]
        {\mmatch{\G}{\pat}{e}{\Delta_\pat} & \mmatch{\Delta_\pat}{\qat}{f}{\Delta_\qat}}
        {\match{\otimes(\pat,\qat)}{\otimes(e,f)}{\Delta_\qat}}
      }{}
      \cinferenceRule[matchUnionOne|$\cup$1]{}
      {\linferenceRule[formula]
        {\match{\pat}{e}{\Delta}}
        {\match{\pat\cup\qat}{e}{\Delta}}
      }{}
      \cinferenceRule[matchUnionTwo|$\cup$2]{}
      {\linferenceRule[formula]
        {\nmatch{\pat}{e} &  \match{\qat}{e}{\Delta}}
        {\match{\pat\cup\qat}{e}{\Delta}}
      }{}
      \cinferenceRule[matchInter|$\cap$]{}
      {\linferenceRule[formula]
        {\mmatch{\G}{\pat}{e}{\Delta_\pat} \hskip 0.1in \mmatch{\Delta_\pat}{\qat}{e}{\Delta_\qat}}
        {\mmatch{\G}{\pat\cap\qat}{e}{\Delta_\qat}}
      }{}
      \cinferenceRule[matchIdent|ident]{}
      {\linferenceRule[formula]
        {\Gamma(\ident) = e}
        {\mmatch{\G}{\freeident}{e}{\G}}
      }{}
      \cinferenceRule[matchVars|vars]{}
      {\linferenceRule[formula]
        {vars\subseteq\freevars{e}}
        {\mmatch{\G}{p(vars)}{e}{\G}}
      }{}
\end{calculus}
\begin{calculus}
      \cinferenceRule[matchNvars|nVars]{}
      {\linferenceRule[formula]
        {vars\cap\freevars{e}=\emptyset}
        {\mmatch{\G}{p(\neg vars)}{e}{\G}}
      }{}
      \cinferenceRule[matchFree|free]{}
      {\linferenceRule[formula]
        {{\ident} \notin \G}
        {\mmatch{\G}{\freeident{}}{e}{\rext{\G}{\freeident{}}{e}}}
      }{}
      \cinferenceRule[matchWild|wild]{}
      {\linferenceRule[formula]
        {}
        {\mmatch{\G}{\wildpat}{e}{\G}}
      }{}
      \cinferenceRule[matchNeg|$\neg$]{}
      {\linferenceRule[formula]
        {\nmatch{\pat}{e} & \boundvars{\pat}=\emptyset}
        {\mmatch{\G}{\neg\pat}{e}{\G}}
      }{}
      \cinferenceRule[matchNow|now]{}
      {\linferenceRule[formula]
        {\ematch{x}{\longeag{now(x)}}}
        {\ematch{x}{e}}
      }{}
    \end{calculus}
  \end{calculuscollections}
\end{center}

\paragraph{Forward-Chaining Proof}
\begin{center}
  \begin{calculuscollections}{\columnwidth}
    \begin{calculus}
      \cinferenceRule[fpPat|pat]{}
      {\linferenceRule[formula]
        {\phi \in (\Sigma,\G) & \ematch{\phi}{\pat}}
        {\fpchks{\G}{\pat}{\phi}{\Sigma}}
      }{}
\end{calculus}
\begin{calculus}
      \cinferenceRule[fpMP|MP]{}
      {\linferenceRule[formula]
        {\fpchk{\G}{\fproof_2}{\phi}&\fpchk{\G}{\fproof_1}{\phi\limply\psi}}
        {\fpchk{\G}{\fproof_1~\fproof_2}{\psi}}
      }{}
\end{calculus}
\begin{calculus}
      \cinferenceRule[mpInst|$\forall$]{}
      {\linferenceRule[formula]
        {\fpchk{\G}{\fproof}{\forall x~\phi}}
        {\fpchk{\G}{\fproof~\theta}{\subst[\phi]{x}{\eag{\theta}}}}
      }{}
    \end{calculus}
  \end{calculuscollections}
\end{center}

\paragraph{Structural Rules}
{\scriptsize\begin{center}
  \begin{calculuscollections}{\columnwidth}
    \begin{calculus}
      \cinferenceRule[spFocL|focusL]{}
      {\linferenceRule[formula]
        {\ematch{\pat}{\phi} & \tpchk{\G_1,\G_2}{H}{\sproof}{\neg{\phi},\Delta}{H_{\neg{\phi}}}}
        {\tpchk{\G_1,\phi,\G_2}{H}{\sfocus{\pat}{\sproof}}{\Delta}{H}}
      }{}
      \cinferenceRule[spFocR|focusR]{}
      {\linferenceRule[formula]
        {\ematch{\pat}{\phi} & \tpchk{\G}{H}{\sproof}{\phi,\Delta_1,\Delta_2}{H_{\neg{\phi}}}}
        {\tpchk{\G}{H}{\sfocus{\pat}{\sproof}}{\Delta_1,\phi,\Delta_2}{H}}
      }{}
      \cinferenceRule[spShowId|id]{}
      {\linferenceRule[formula]
        {\ematch{\pat}{\phi} & \phi \in \upfacts{{\pat}s}{{\fproof}s}~\text{valid in}~\folr{}}
        {\tpchk{\G}{H}{\sshow{\pat}{\using~{\pat}s~{\fproof}s~\by~\closeid}}{\phi,\Delta}{H}}
      }{}
      \cinferenceRule[spShowR|R]{}
      {\linferenceRule[formula]
        {\ematch{\pat}{\phi} & \upfacts{{\pat}s}{{\fproof}s}~\text{valid in}~\folr{} }
        {\tpchk{\G}{H}{\sshow{\pat}{\using~{\pat}s~{\fproof}s~\by~\mathbb{R}}}{\phi,\Delta}{H}}
      }{}
      \cinferenceRule[spFlet|flet]{}
      {\linferenceRule[formula]
        {\tpchk{\ext{\Gamma}{t({\tt{x\_}})}{\llabs{\cemp{}}{e}{t_{now}}}}{H,t_{now}}{\sproof}{\phi,\Delta}{H_\sproof}}
        {\tpchk{\G}{H}{\slet{t({\tt{x\_}})}{e}{\sproof}}{\phi,\Delta}{H_\sproof}}
      }{}
\end{calculus}
\begin{calculus}
      \cinferenceRule[spState|state]{}
      {\linferenceRule[formula]
        {\tpchk{\G}{\hhtime{H}{t}}{\sproof}{\phi,\Delta}{H_\phi}}
        {\tpchk{\G}{H} {\sstate{t}{\sproof}}{\phi,\Delta}{H_\phi}}
      }{}
      \cinferenceRule[spNote|note]{}
      {\linferenceRule[formula]
        {\fpchk{\G}{\fproof}{\psi} & {\tpchk{\ext{\Gamma}{x}{\psi}}{H}{\sproof}{\phi,\Delta}{H_\phi}}}
        {\tpchk{\G}{H}{\snote{x}{\fproof}{\sproof}}{\phi,\Delta}{H_\phi}}
      }{}
      \cinferenceRule[spLet|let]{}
      {\linferenceRule[formula]
        {\match{p}{\eag{e}}{\G_1} & \tpchk{\G_1}{H}{\sproof}{\phi,\Delta}{H_\phi}}
        {\tpchk{\G}{H}{\slet{p}{e}{\sproof}}{\phi,\Delta}{H_\phi}}
      }{}
      \cinferenceRule[spHave|have]{}
      {\linferenceRule[formula]
        { \tpchk{\G}{H}{\sproof_1}{\eag{\psi}}{H_\psi} & 
          \tpchk{\ext{\Gamma}{x}{\eag{\psi}}}{H}{\sproof_2}{\phi,\Delta}{H_\phi}}
        {\tpchk{\G}{H}{\have{x}{\psi}{\sproof_1}{\sproof_2}}{\phi,\Delta}{H_\phi}}
      }{}
    \end{calculus}
  \end{calculuscollections}
\end{center}}

%\newpage
\paragraph{Backward-Chaining First-Order Proof}
\begin{center}
  \begin{calculuscollections}{\columnwidth}
    \begin{calculus}
      \cinferenceRule[spbtest|$\limply{R}$]{}
      {\linferenceRule[formula]
        {\match{\pat}{\psi}{\Gamma_1} & \ppchk{\ext{\G_1}{x}{\psi}}{\sproof}{\phi,\Delta}}
        {\ppchk{\G}{\bassert{x}{\pat}{\sproof}}{(\psi\limply\phi),\Delta}}
      }{}
      \cinferenceRule[sporr|$\vee$R]{}
      {\linferenceRule[formula]
        {\ppchk{\G}{\sproof}{\rext{x}{\phi}{\rext{y}{\psi}{\Delta}}}}
        {\ppchk{\G}{\pcaseor{x:\phi}{y:\psi}{\sproof}}{(\phi\lor\psi),\Delta}}
      }{}
\end{calculus}
\begin{calculus}
      \cinferenceRule[spiff|$\bimply$]{}
      {\linferenceRule[formula]
        {\deduce{\ppchk{\rext{\G}{x}{\phi}}{\sproof_1}{\lext{y}{\psi}{\Delta}}\hskip 0.1in \ppchk{\rext{\G}{y}{\psi}}{\sproof_2}{\lext{x}{\phi}{\Delta}}} 
             {\match{p\bimply{q}}{\psi\bimply\phi}{\Gamma_1}}}
        {\ppchk{\G}{\pcaseiffr{p}{q}{\sproof_1}{\sproof_2}}{(\phi\bimply\psi),\Delta}}
      }{}
      \cinferenceRule[spandr|$\wedge$R]{}
      {\linferenceRule[formula]
        {\deduce{\ppchk{\G}{\sproof_1}{\lext{x}{\phi}{\Delta}}
               \hskip 0.1in \ppchk{\G}{\sproof_2}{\lext{y}{\psi}{\Delta}}}
         {\match{p\land{q}}{\psi\land\phi}{\Gamma_1}}}
        {\ppchk{\G}{\pcasear{p}{q}{\sproof_1}{\sproof_2}}{(\phi\land\psi),\Delta}}
      }{}
    \end{calculus}
  \end{calculuscollections}
\end{center}

\newpage
\paragraph{Backward-Chaining Structured Box Execution}
\begin{center}
  \begin{calculuscollections}{\columnwidth}
    \begin{calculus}
      \cinferenceRule[spbeqsub|${[:=]sub}$]{}
      {\linferenceRule[formula]
        {\tpchk{\G}{\hhsub{H}{x}{\eag{\theta}}}{\sproof}{\subst[\phi]{x}{\eag{\theta}},\Delta}{H_1}}
        {\tpchk{\G}{H}{\bassign{x}{\theta}{\sproof}}{[x:=\eag{\theta}]\phi,\Delta}{H_1}}
      }{$\text{if}~\phi_{x}^{\theta}~\text{admissible}$}
      \cinferenceRule[spbeqeq|${[:=]eq}$]{}
      {\linferenceRule[formula]
        {\tpchk{\subst[\G]{x}{x_i},x=\eag{\theta}}{\hheq{H}{x}{x_i}{\eag{\theta}}}{\sproof}{\phi,\subst[\Delta]{x}{x_i}}{H_1}}
        {\tpchk{\G}{H}{\bassign{x}{\theta}{\sproof}}{[x:=\eag{\theta}]\phi,\Delta}{H_1}}
      }{$\text{if}~x_i~\text{fresh}$}
      \cinferenceRule[spbany|${[:=^*]}$]{}
      {\linferenceRule[formula]
        {\tpchk{\subst[\G]{x}{x_i}}{\hhany{H}{x}{x_i}}{\sproof}{\phi,\subst[\Delta]{x}{x_i}}{H_1}}
        {\tpchk{\G}{H}{\bassignany{x}{\sproof}}{[\prandom{x}]\phi,\Delta}{H_1}}
      }{$\text{if}~x_i~\text{fresh}$}
      \cinferenceRule[spbcase|${[\cup]}$]{}
      {\linferenceRule[formula]
        {\deduce{\tpchk{\G_\alpha}{H}{\sproof_\alpha}{[\alpha]\phi,\Delta}{H_\alpha} \hskip 0.1in \tpchk{\G_\beta}{H}{\sproof_\beta}{[\beta]\phi,\Delta}{H_\beta}}
{\match{p}{\alpha}{\G_\alpha} &
 \match{q} {\beta}{\G_\beta}}}
        {\tpchk{\G}{H}{\bcase{p}{\sproof_\alpha}{q}{\sproof_\beta}}{[\alpha\cup\beta]\phi,\Delta}{H}}
      }{}
      \cinferenceRule[spHM|mid]{}
      {\linferenceRule[formula]
        {\tpchk{\G}{H}{\sproof_\psi}{[\alpha]\eag{\psi},\Delta}{H_\psi} \hskip 0.1in \tpchk{\arb{\G}{\alpha},\eag{\psi}}{\arb{H_\psi}{\alpha}}{\sproof_\phi}{[\beta]\phi,\arb{\Delta}{\alpha}}{H_\phi}}
        {\tpchk{\G}{H}{\bcon{\psi}{\sproof_\psi}{\sproof_\phi}}{[\alpha][\beta]\phi,\Delta}{H_\phi}}
      }{}
      \cinferenceRule[spsolver|${[']}$]{}
      {\linferenceRule[formula]
        {\deduce
{\tpchk{\ext{\ext{\Gamma_Q}{dom}{\forall s \in [0,t]~Q}}{t}{t\geq 0}}{\hhsub{H}{x}{y(t)}}{\sproof}{\subst[\phi]{x}{y(t)},\Delta}{H_\phi}}
{
 \match{\pat}{\{x'=\theta~\&~Q\}}{\Gamma_{x'}}
&\mmatch{\G_{x'}}{\pat_t}{t\geq{0}}{\Gamma_{t}}
&\mmatch{\G_{t}}{\pat_{dom}}{Q}{\G_Q}
}}
        {\tpchk{\G}{H}{\bsolve{\pat}{\pat_{t}}{\pat_{dom}}{\sproof}}{[\{x'=\theta~\&~Q\}]\phi,\Delta}{H_\phi}}
      }{$y(0)=x,~y'(t)=\theta(y)$}

      \cinferenceRule[spboxand|${[\land]}$]{}
      {\linferenceRule[formula]
        {\deduce
          {\tpchk{\G_\phi}{H}{\sproof_\phi}{[\alpha]\phi,\Delta}{H_\phi}
            \hskip 0.1in
            \tpchk{\G_\psi}{H}{\sproof_\psi}{[\alpha]\psi,\Delta}{H_\psi}
          }
          {
            \match{p_\phi}{\phi}{\G_\phi}
            &\match{p_\psi}{\psi}{\G_\psi}
          }}
        {\tpchk{\G}{H}{\bcase{p_\phi}{\sproof_\phi}{p_\psi}{\sproof_\psi}}{[\alpha](\phi\land\psi),\Delta}{H}}
      }{}

      \cinferenceRule[spboxiter|${[*]}$]{}
      {\linferenceRule[formula]
        {\deduce
          {\tpchk{\G_\phi}{H}{\sproof_\phi}{\phi,\Delta}{H_\phi}
           \hskip 0.1in
           \tpchk{\G_\alpha}{H}{\sproof_\alpha}{[\alpha][\alpha^*]\psi,\Delta}{H_\alpha}
          }
          {
            \match{p_\phi}{\phi}{\G_\phi}
            &\match{p_\alpha}{\alpha}{\G_\alpha}
          }
        }
        {\tpchk{\G}{H}{\bcase{p_\phi}{\sproof_\phi}{p_\alpha}{\sproof_\alpha}}{[\alpha^*]\phi,\Delta}{H}}
      }{}
    \end{calculus}
  \end{calculuscollections}
\end{center}

\newpage
\paragraph{Backward-Chaining Structured Diamond Execution}
\begin{center}
  \begin{calculuscollections}{\columnwidth}
    \begin{calculus}
      \cinferenceRule[spdeqsub|${\langle{}:=\rangle{}sub}$]{}
      {\linferenceRule[formula]
        {\tpchk{\G}{\hhsub{H}{x}{\eag{\theta}}}{\sproof}{\subst[\phi]{x}{\eag{\theta}},\Delta}{H_1}}
        {\tpchk{\G}{H}{\dassign{x}{\theta}{\sproof}}{\ddiamond{x:=\eag{\theta}}{\phi},\Delta}{H_1}}
      }{$\text{if}~\phi_{x}^{\theta}~\text{admissible}$}
      \cinferenceRule[spdeqeq|${\langle{}:=\rangle{}eq}$]{}
      {\linferenceRule[formula]
        {\tpchk{\subst[\G]{x}{x_i},x=\eag{\theta}}{\hheq{H}{x}{x_i}{\eag{\theta}}}{\sproof}{\phi,\subst[\Delta]{x}{x_i}}{H_1}}
        {\tpchk{\G}{H}{\dassign{x}{\theta}{\sproof}}{\ddiamond{\humod{x}{\eag{\theta}}}{\phi},\Delta}{H_1}}
      }{$\text{if}~x_i~\text{fresh}$}
      \cinferenceRule[spdeqsub|${\langle{}:=\rangle{}sub}$]{}
      {\linferenceRule[formula]
        {\tpchk{\G}{\hhsub{H}{x}{\eag{\theta}}}{\sproof}{\subst[\phi]{x}{\eag{\theta}},\Delta}{H_1}}
        {\tpchk{\G}{H}{\dassignany{x}{\theta}{\sproof}}{\ddiamond{x:=\eag{\theta}}{\phi},\Delta}{H_1}}
      }{$\text{if}~\phi_{x}^{\theta}~\text{admissible}$}
%      \cinferenceRule[spdany|${\langle:=^*\rangle}$]{}
%      {\linferenceRule[formula]
%        {\tpchk{\subst[\G]{x}{x_i}}{\hhany{H}{x}{x_i}}{\sproof}{\phi,\subst[\Delta]{x}{x_i}}{H_1}}
%        {\tpchk{\G}{H}{\bassignany{x}{\sproof}}{\ddiamond{\prandom{x}}{\phi},\Delta}{H_1}}
%      }{$\text{if}~x_i~\text{fresh}$}
      \cinferenceRule[spdcase|${\langle\cup\rangle}$]{}
      {\linferenceRule[formula]
        {\deduce{\tpchk{\G}{H}{\sproof}{x:\ddiamond{\alpha}{\phi},y:\ddiamond{\beta}{\phi},\Delta}{H_{\cup}}}
               {\match{p}{\alpha}{\G_\alpha} &  \match{q}{\beta}{\G_\beta}}}
        {\tpchk{\G}{H}{\dcase{x:p}{y:q}{\sproof}}{\ddiamond{\alpha\cup\beta}{\phi},\Delta}{H}}
      }{}
      \cinferenceRule[spdHM|mid]{}
      {\linferenceRule[formula]
        {\tpchk{\G}{H}{\sproof_\psi}{\ddiamond{\alpha}{\eag{\psi}},\Delta}{H_\psi} \hskip 0.1in \tpchk{\arb{\G}{\alpha},\eag{\psi}}{\arb{H_\psi}{\alpha}}{\sproof_\phi}{\ddiamond{\beta}{\phi},\arb{\Delta}{\alpha}}{H_\phi}}
        {\tpchk{\G}{H}{\bcon{\psi}{\sproof_\psi}{\sproof_\phi}}{\ddiamond{\alpha}{\ddiamond{\beta}{\phi}},\Delta}{H_\phi}}
      }{}
      \cinferenceRule[spdsolver|${\langle'\rangle}$]{}
      {\linferenceRule[formula]
        {\deduce{\tpchk{\G}{H}{\sproof_{dom}}{\theta_{t}\geq{0}\land\forall~s~s\in[0,\theta_{}]~Q,\Delta}{H_{dom}}}
 {\deduce
{\tpchk{\Gamma_Q}{\hhsub{H}{x}{y(\eag{\theta_t})}}{\sproof}{\subst[\phi]{x}{y(\eag{\theta_t})},\Delta}{H_\phi}}
{
 \match{\pat}{\{x'=\theta~\&~Q\}}{\Gamma_{x'}}
%&\mmatch{\G_{x'}}{\pat_t}{t\geq{0}}{\Gamma_{t}}
&\mmatch{\G_{t}}{\pat_{dom}}{Q}{\G_Q}
}}}
        {\tpchk{\G}{H}{\dsolve{\pat}{\theta_{t}}{\pat_{dom}}{\sproof_{dom}}{\sproof}}{\ddiamond{\{x'=\theta~\&~Q\}}{\phi},\Delta}{H_\phi}}
      }{$y(0)=x,~y'(t)=\theta(y)$}
      \cinferenceRule[spdiaand|${\langle*\rangle}$]{}
      {\linferenceRule[formula]
        {\deduce
          {\tpchk{\G_\phi}{H}{\sproof_\phi}{\phi,\ddiamond{\alpha}{\ddiamond{\alpha^*}{\phi}},\Delta}{H_\phi}
          }
          {\match{p_\phi}{\phi}{\G_\phi}
            &\match{p_\alpha}{\alpha}{\G_\alpha}
          }}
        {\tpchk{\G}{H}{\pcasediaiter{\alpha}{\phi}{\sproof_\phi}{\sproof_\alpha}}{\ddiamond{\alpha^*}{\phi},\Delta}{H_\phi}}
      }{}
    \end{calculus}
  \end{calculuscollections}
\end{center}
%\newpage
\paragraph{Invariant (and Variant) Execution}
\begin{center}
  \begin{calculuscollections}{\columnwidth}
    \begin{calculus}
      \cinferenceRule[spFinLoop|fin${[*]}$]{}
      {\linferenceRule[formula]
        {\tpchk{\arb{\G}{\alpha},\invlist{}}{\arb{H}{\alpha}}{\sproof}{\phi,\arb{\Delta}{\alpha}}{H_\alpha}}
        {\tpchk{\G,\invlist{}}{H}{\finally{\sproof}}{[\alpha^*]\phi,\Delta}{H_\alpha}}
      }{}
    \end{calculus}
    \begin{calculus}
      \cinferenceRule[spFinOde|fin${[']}$]{}
      {\linferenceRule[formula]
        {\tpchk{\arb{\G}{\alpha},\invlist{},Q}{\arb{H}{\alpha}}{\sproof}{\phi,\arb{\Delta}{\alpha}}{H_{x'}}}
        {\tpchk{\G,\invlist{}}{H}{\finally{\sproof}}{[\{x'=\theta~\&~Q\}]\phi,\Delta}{H_{x'}}}
      }{}
    \end{calculus}
  \end{calculuscollections}
\end{center}
\begin{center}
  \begin{calculuscollections}{\columnwidth}
    \begin{calculus}
      \cinferenceRule[spInvLoop|inv${[*]}$]{}
      {\linferenceRule[formula]
        {\deduce{\tpchk{\Gamma,\invlist{}}{H}{\sproof_{Pre}}{\eag{\psi},\Delta}{H_{Pre}} \hskip 0.1in  \tpchk{\arb{\G}{\alpha},\invlist{},\eag{\phi}}{\arb{H}{\alpha}}{\sproof_{Inv}}{[\alpha]\eag{\psi},\arb{\Delta}{\alpha}}{H_{Inv}}}
 %,\hrany{y}{y_j}
           {\tpchk{\Gamma,\invlist{},x:\psi}{H}{\iproof}{[\alpha^*]\phi,\Delta}{H_{\it{Tail}}}}}
        {\tpchk{\G,\invlist{}}{H}{\sinv{x}{\psi}{\sproof_{Pre}}{\sproof_{Inv}}{\iproof}}{[\alpha^*]\phi,\Delta}{H}}
      }{}
      \cinferenceRule[spInvOde|inv${[']}$]{}
      {\linferenceRule[formula]
        {\deduce{\tpchk{\Gamma,\ext{\Delta}{x}{\eag{\psi}}}{H}{\iproof}{[\{x'=\theta~\&~Q\}]\phi,\Delta}{H_{x'}}}
{\tpchk{\Gamma,\invlist{}}{H}{\sproof_{Pre}}{\eag{\psi},\Delta}{H_{Pre}}
   \hskip 0.1in\tpchk{\arb{\G}{\alpha},\invlist{},Q}{\arb{H}{\alpha}}{\sproof_{Inv}}{[\humod{x'}{\theta}](\eag{\psi})',\arb{\Delta}{\alpha}}{H_{Inv}}
}}
        {\tpchk{\G,\Delta}{H}{\sinv{x}{\psi}{\sproof_{Pre}}{\sproof_{Inv}}{\iproof}}{[\{x'=\theta~\&~Q\}]\phi,\Delta}{H_{x'}}}
      }{}
      \cinferenceRule[spGhost|ghost${[']}$]{}
      {\linferenceRule[formula]
        {\tpchk{\G,y=\eag{\theta_{y}}}{\Delta;H}{\iproof}{[\{x'=\eag{\theta_{x'}},y'=\eag{\theta_{y'}}~\&~H\}]\phi,\Delta}{H_{x'}} & \eag{\theta_{y'}}~{\rm linear\ in}~y}
        {\tpchk{\G,\invlist{}}{H}{\sghost{y}{\theta_{y'}}{\theta_{y}}{\iproof}}{[\{x'=\theta_{x'}~\&~H\}]\phi,\Delta}{H_{x'}}}
      }{}
      \cinferenceRule[spConLoop|con${\langle*\rangle}$]{}
      {\linferenceRule[formula]
        {\deduce
    {\tpchk{\Gamma}{H}{\sproof_{Pre}}{\exists~x~\eag{\phi}(x),\Delta}{H_{Pre}} 
     {\hskip 0.1in}
     \tpchk{\arb{\G}{\alpha},\eag{\psi}(x),x>0}{\arb{H}{\alpha}}{\sproof_{Inv}}{\ddiamond{\alpha}{\longeag{\psi(x-1)}},\arb{\Delta}{\alpha}}{H_{Inv}}}
 %,\hrany{y}{y_j}
           {\tpchk{\arb{\Gamma}{\alpha},x:\psi,x\leq{0}}{H}{\sproof_{Post}}{\ddiamond{\alpha^*}{\phi}}{H_{\it{Post}}}}}
        {\tpchk{\G}{H}{\scon{x}{\phi}{\sproof_{Pre}}{\sproof_{Post}}{\sproof_{Inv}}}{\ddiamond{\alpha^*}{\phi},\Delta}{H}}
      }{}

    \end{calculus}
  \end{calculuscollections}
\end{center}

\newpage
\paragraph{Implicit Conversion Rules}
The implicit conversions leave the proof and history untouched, so we write only their effect on the sequent.
\begin{center}
  \begin{calculuscollections}{\columnwidth}
    \begin{calculus}
      \cinferenceRule[imSeq|{imp$[;]$}]{}
      {\linferenceRule[formula]
        {\seqq{}{\dbox{\alpha}{\dbox{\beta}{\phi}},}}
        {\seqq{}{\dbox{\alpha;\beta}{\phi},}}
      }{}
      \cinferenceRule[imDSeq|imp$\langle{;}\rangle$]{}
      {\linferenceRule[formula]
        {\seqq{}{\ddiamond{\alpha}{\ddiamond{\beta}{\phi}},}}
        {\seqq{}{\ddiamond{\alpha;\beta}{\phi},}}
      }{}
      \cinferenceRule[imAll|imp$\forall$]{}
      {\linferenceRule[formula]
        {\seqq{}{\dbox{\prandom{x}}{\phi},}}
        {\seqq{}{\forall~x~\phi,}}
      }{}
      \cinferenceRule[imEx|imp$\exists$]{}
      {\linferenceRule[formula]
        {\seqq{}{\ddiamond{\prandom{x}}{\phi},}}
        {\seqq{}{\exists~x~\phi,}}
      }{}
      \cinferenceRule[imNAll|imp$\neg\forall$]{}
      {\linferenceRule[formula]
        {\seqq{}{\exists~x~\neg\phi,}}
        {\seqq{}{\neg\forall~x~\phi,}}
      }{}
    \end{calculus}
    \begin{calculus}
      \cinferenceRule[imNExists|imp$\neg\exists$]{}
      {\linferenceRule[formula]
        {\seqq{}{\forall~x~\neg\phi,}}
        {\seqq{}{\neg\exists~x~\phi,}}
      }{}
      \cinferenceRule[imNN|imp$\neg\neg$]{}
      {\linferenceRule[formula]
        {\seqq{}{\phi,}}
        {\seqq{}{\neg\neg\phi,}}
      }{}
      \cinferenceRule[imNAnd|imp$\neg\land$]{}
      {\linferenceRule[formula]
        {\seqq{}{\neg\phi\lor\neg\psi,}}
        {\seqq{}{\neg(\phi\land\psi),}}
      }{}
      \cinferenceRule[imNOr|imp$\neg\lor$]{}
      {\linferenceRule[formula]
        {\seqq{}{\neg\phi\land\neg\psi,}}
        {\seqq{}{\neg(\phi\lor\psi),}}
      }{}
      \cinferenceRule[imNImp|imp$\neg\limply$]{}
      {\linferenceRule[formula]
        {\seqq{}{\phi\land\neg\psi,}}
        {\seqq{}{\neg(\phi\limply\psi),}}
      }{}
    \end{calculus}
    \begin{calculus}
      \cinferenceRule[imNEquiv|imp$\neg\bimply$]{}
      {\linferenceRule[formula]
        {\seqq{}{(\phi\land\neg\psi)\lor(\neg\phi\land\psi)}}
        {\seqq{}{\neg(\phi\bimply\psi),}}
      }{}
      \cinferenceRule[imNBox|imp$\neg{[]}$]{}
      {\linferenceRule[formula]
        {\seqq{}{\ddiamond{\alpha}{\neg\psi},}}
        {\seqq{}{\neg\dbox{\alpha}{\psi},}}
      }{}
      \cinferenceRule[imNDia|imp$\neg{\langle\rangle}$]{}
      {\linferenceRule[formula]
        {\seqq{}{\dbox{\alpha}{\neg\psi},}}
        {\seqq{}{\neg\ddiamond{\alpha}{\psi},}}
      }{}
      \cinferenceRule[imAssert|imp$?$]{}
      {\linferenceRule[formula]
        {\seqq{}{\dbox{?(\phi)}{\psi},}}
        {\seqq{}{\phi\limply\psi,}}
      }{}
    \end{calculus}
  \end{calculuscollections}
\end{center}

\end{document}